\documentclass[11pt]{article}
\usepackage{geometry}
\usepackage[utf8]{inputenc}
\usepackage{graphics}
\usepackage{amsmath}
\usepackage{amssymb}
\newtheorem{theorem}{Theorem}

\newtheorem{definition}{Definition}
\newtheorem{remark}{Remark}
\newtheorem{example}{Example}
\newtheorem{corollary}{Corollary}
\newtheorem{assumption}{Assumption}
\newtheorem{proposition}{Proposition}
\newenvironment{proof}{{\noindent\it Proof}\quad}{\hfill $\square$\par}
\usepackage{mathrsfs}
\usepackage{natbib}
\usepackage{graphicx}
\usepackage{float} 
\usepackage{bbding}
\usepackage{color}
\usepackage{algorithm}
\usepackage{algpseudocode}
\usepackage{makecell}
\usepackage{subfig}
\usepackage[colorlinks, linkcolor=red, anchorcolor=blue,citecolor=green]{hyperref}

\title{Joint Pricing in SPX and VIX Derivative Markets with Composite Change of Time Models}

\author{Liexin Cheng\textsuperscript{a}\textsuperscript{b} and Xue Cheng\textsuperscript{a}\textsuperscript{b} and Xianhua Peng\textsuperscript{c}}

\begin{document}

\maketitle

\vspace{-5pt}
\begin{center}
    \textsuperscript{a} School of Mathematical Science, Peking University, China\\
    \textsuperscript{b} Center for Statistical Science \& LMEQF, Peking University, China\\
    \textsuperscript{c} HSBC Business School, Peking University, China 
\end{center}
\begin{abstract}
The Chicago Board Options Exchange Volatility Index (VIX) is calculated from SPX options and derivatives of VIX are also traded in market, which leads to the so-called ``consistent modeling" problem. This paper proposes a time-changed Lévy model for log price with a composite change of time structure to capture both features of the implied SPX volatility and the implied volatility of volatility. Consistent modeling is achieved naturally via flexible choices of jumps and leverage effects, as well as the composition of time changes. Many celebrated models are covered as special cases. From this model, we derive an explicit form of the characteristic function for the asset price (SPX) and the pricing formula for European options as well as VIX options. The empirical results indicate great competence of the proposed model in the problem of joint calibration of the SPX/VIX Markets.
    
    ~
\end{abstract}

\tableofcontents
\newpage
\section{Introduction}
\subsection{Consistent Modeling Problem}
By the definition from CBOE, Volatility Index (VIX), as an indicator of implied volatility of $X_t: = \ln \frac{e^{-rt}S_t}{S_0}$ in the following $30$ days, is given as 
$$
\label{VIX}
     \text{VIX}_t^2  = - \frac{2}{\bar \tau} E^\mathbb{P}\left[ \ln \frac{e^{-r \bar \tau}S_{t+\bar \tau}}{S_t} \,\middle\vert\,\mathcal{F}_t^X \right] \times 100^2,
$$
where $\bar \tau = 30/365$, $\mathbb P$ is the risk neutral measure of the equity market, and $X$ is adapted to $\{\mathcal{F}_t^X\}_{t\ge 0}$. Likewise, we have the formula of VVIX, the volatility of volatility computed from VIX market:
$$
\begin{aligned}
     \text{VVIX}_t^2  &=  - \frac{2}{\bar \tau} E^\mathbb{P}\left[\ln \frac{e^{-r \bar \tau}\text{VIX}_{t+\bar \tau}}{\text{VIX}_t} \,\middle\vert\, \mathcal{F}_t^X\right] \times 100 ^2.
\end{aligned}
$$
In the following, the multiplier $100^2$ is neglected for brevity. It is important to note that both definitions hold if we believe that measure $\mathbb P$ is risk-neutral in both SPX option market and the VIX market. That is, the two markets can be consistently modeled. The problem of joint calibration to SPX market and VIX market is closely related to the calibration of the dynamics of $\text{VIX}$ and $\text{VVIX}$, which can be considered as volatility and volatility of volatility respectively. 

To understand the essence of the joint calibration problem, we assume that the log return process $X_t$ is a continuous semimartingale such that $\mathrm{d}\langle X\rangle_t= v_t \mathrm{d}t$ and that $v$ is positive and Markovian.
Thus, we have
\begin{equation}
    \label{arg1}
    \operatorname{VIX}_t = \sqrt{\frac{1}{\bar \tau}E^{\mathbb{P}}\left[\langle X\rangle_{t+\bar \tau}-\langle X\rangle_t\,\middle\vert\, \mathcal{F}_t^X\right]}
    = :f_1(v_t),\end{equation}
where $f_1(\cdot)$ is a function determined by the model. The assumption incorporates classic SVMs including Heston and 3/2 volatility model. Furthermore, VVIX takes the form
$$\operatorname{VVIX}_t = \sqrt{-\frac{2}{\bar \tau}E^{\mathbb{P}}\left[\ln \frac{e^{-r\bar\tau} f_1(v_{t+\bar \tau})}{f_1(v_t)}\,\middle\vert\,\mathcal{F}_t^X\right]} =: g_1(v_t),$$
for some function $g_1(\cdot)$. This implies that VVIX and VIX are coupled through the volatility factor $v$ and such a restrictive expression is expected to underperform in joint calibration, since the market does not conform to a simple relationship, as shown in Figure \ref{VIX VVIX} where the dynamic relationship of VIX and VVIX varies frequently.

\begin{figure}
 \begin{minipage}{0.49\linewidth}
 \centering
 \includegraphics[width=\linewidth]{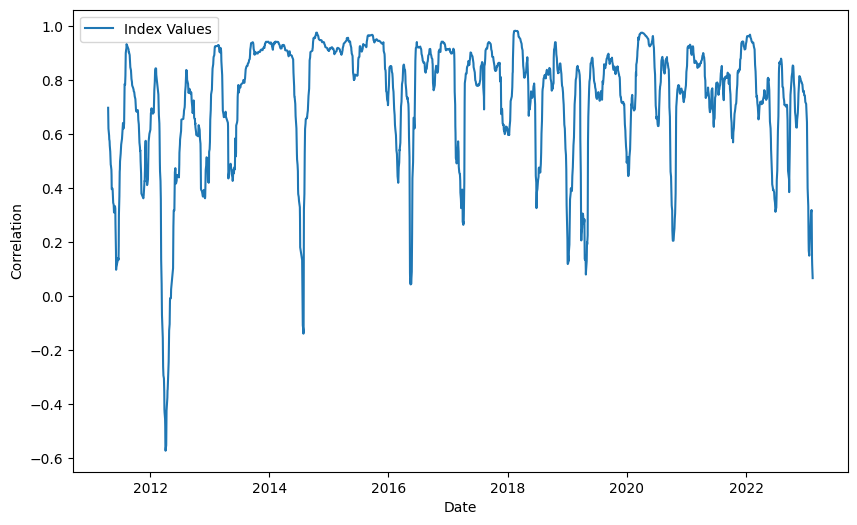}
 \end{minipage}
  \begin{minipage}{0.49\linewidth}
 \centering
 \includegraphics[width=\linewidth]{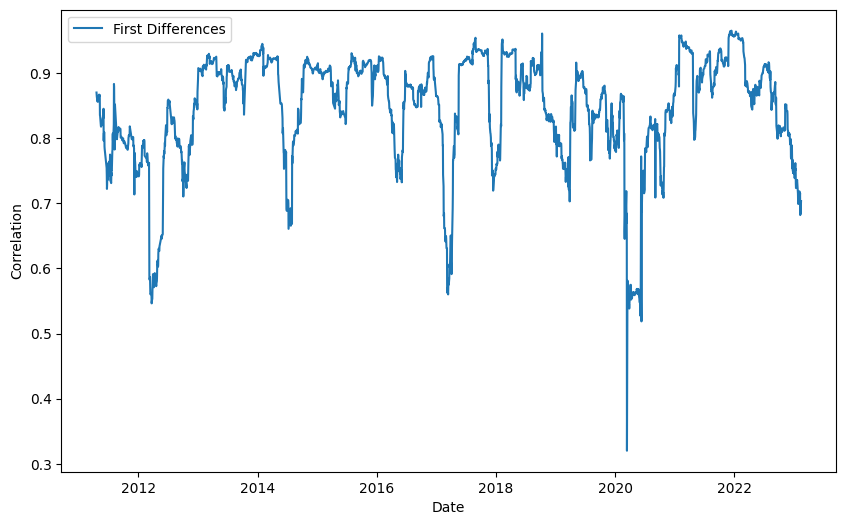}
\end{minipage}
\caption{Correlation of the VVIX and VIX from January 2011 to February 2023\label{VIX VVIX}}{The Pearson correlation coefficients of index values (left) and their first differences (right) are computed in 60-day rolling windows with means 0.7343 (left) and 0.8442 (right) and standard deviation 0.23 (left) and 0.09 (right).}
\end{figure}

There has been a lot of research on the joint calibration problem, where various explorations of decoupling volatility and vol-of-vol are proposed. One line of work is aimed at extending the widely used stochastic volatility models (SVMs) to have more realistic and flexible vol-of-vol functionals. One such example is the 3/2 model applied in \cite{drimus2012options} and in \cite{baldeaux2014consistent}. The former shows that a pure-diffusion 3/2 model reproduces upward-sloping implied volatilities of options on realized variance, and the latter produces the sameupwward-sloping characteristics on VIX options. \cite{fouque2018heston} proposed a Heston vol-of-vol model, where the vol-of-vol consists of exogenous stochastic factors. The model is able to produce skew characteristics in both SPX and VIX option markets. \cite{jeon2021consistent} developed a multiscale stochastic volatility model, where the double mean-reverting volatility factors greatly reduce the joint calibration error. Furthermore, \cite{jacquier2018VIX}, \cite{gatheral2020quadratic}, \cite{abi2025joint} and \cite{bondi2024rough} proposed rough volatility models for VIX futures and the joint calibration of SPX and VIX options. In \cite{gatheral2020quadratic} specifically, the proposed quadratic rough Heston model is shown to produce the price-feedback (Zumbach) effect and fits well to the whole term structures of SPX and VIX options. \cite{abi2025joint} developed Gaussian polynomial volatility models and achieved strong performance in a one-factor Markovian setting. And \cite{bondi2024rough} demonstrated the effectiveness of rough volatility with self-exciting jumps in the joint calibration. Moreover, the role of jumps in stock price is also highlighted in the consistent modeling of the two markets, including \cite{lin2010consistent}, \cite{cont2013consistent}, and \cite{kokholm2015joint}. \cite{lin2010consistent} considered an affine volatility specification where the jumps in stock price and in stock price volatility are correlated. In \cite{cont2013consistent}, instead of modeling the variance process, the variance swap rate is consisdered and modeled as an affine jump-diffusion having co-jumps with the price process. \cite{kokholm2015joint} examined the joint calibration performance under affine volatility model with both jumps in return and in volatility. The paper finds that the relaxation of Feller condition improves the performance considerably, but is still not satisfactory. \cite{zhou2024joint} evaluated the performance of joint calibration for 14 well-known non-affine local stochastic volatility models and found that the non-affine and local volatility structure improve the error but are still far from satisfactory. \cite{papanicolaou2022consistent} established a new consistent modeling framework by directly modeling VIX dynamics, and the volatility function of the log return is inferred from consistent modeling conditions.

To give an example, consider the SVM in \cite{fouque2018heston}:
  $$\begin{array}{l}
\mathrm{d} S_t / S_t=r \mathrm{~d} t+\sqrt{v_t}  \mathrm{~d} W_t^S, \\
\mathrm{~d} v_t=\kappa\left(\theta-v_t\right) \mathrm{d} t+\eta_t \sqrt{v_t} \mathrm{~d} W_t^V, \\

\end{array}$$
where $\eta_t=\eta\left(Y_t^{\varepsilon}, Z_t^\delta\right)$ is a stochastic factor with expression

$$
\begin{array}{l}
\eta_t=\eta\left(Y_t^{\varepsilon}, Z_t^\delta\right) \\
\mathrm{d} Y_t^{\varepsilon}=\frac{V_t}{\varepsilon} \alpha\left(Y_t^{\varepsilon}\right) \mathrm{d} t+\sqrt{\frac{V_t}{\varepsilon}} \beta\left(Y_t^{\varepsilon}\right) \mathrm{d} W_t^Y, \\
\mathrm{~d} Z_t^\delta=V_t \delta c\left(Z_t^\delta\right) \mathrm{d} t+\sqrt{\delta V_t} g\left(Z_t^\delta\right) \mathrm{d} W_t^Z,
\end{array}
$$
with $\left(W_t^S, W_t^V, W_t^Y, W_t^Z\right)$ a correlated Brownian motion with
$$
\mathrm{d} W_t^i \mathrm{~d} W_t^j=\rho_{i j} \mathrm{~d} t,\quad  i, j=S, V, Y, Z.
$$
In this case, we have $\operatorname{VIX}_t = f_2(v_t, Y_t^\epsilon, Z_t^\delta)$ and $\operatorname{VVIX}_t = g_2(v_t, Y_t^\epsilon, Z_t^\delta)$. The three factors jointly calibrate two markets and are shown to significantly improve accuracy. 

However, along with other above-mentioned works, how the factors decouple VIX and VVIX are implicit, which cannot clearly explain why the proposed models work in joint calibration. And the computations in the pricing of European and VIX options are typically more involved, particularly for non-affine volatility models, where simulation techniques are often necessary, or rough volatility models, which are computationally demanding and require solving Riccati-Volterra equations. Moreover, the mentioned works all conducted calibrations on a few selected dates without out-of-sample robustness test except for \cite{lin2010consistent} and \cite{zhou2024joint}.

Another category of models suggest a multi-factor specification of volatility. The original attempts were made by \cite{gatheral2008consistent}
 and \cite{bayer2013fast}. The former proposed a 2-factor double CEV model that improved the short-term fit performance, while the model admits no closed form solution. The modeling approach was followed by the latter, who further considered a 3-factor double mean reverting model with an improved model simulation scheme. Multi-factor affine specifications were also considered. \cite{chen2013consistent} jointly fitted with two different multi-factor affine models and analyze the term structure of VIX futures based on calibrated model parameters. \cite{pacati2018smiling} introduced a deterministic displacement to multi-factor affine models, where the displacement is interpreted as a lower-bound of the forward VIX term structure. \cite{papanicolaou2014regime} and \cite{goutte2017regime} both considered regime-switching Heston models, where sharp volatility regime shifts captures both volatility skews. The former used asymptotic and Fourier methods to efficiently price, and the latter proposed an extension of the EM algorithm to estimate the model. Besides, \cite{guyon2023volatility} proposed a 4-factor path-dependent volatility model that jointly fits very well. The model has continuous stock price path and rough-like volatility path, but relies on simulation methods to solve. \cite{yuan2020new} adopted multi-factor stochastic jump intensity models, where the introduction of a new stochastic factor controls the high-order moments of the risk-neutral distribution. And \cite{bardgett2019inferring} estimated an affine model with a stochastic mean-reverting process. Both reported satisfactory calibration results to the whole joint implied volatility term structure, but the volatility process can become negative because of the negative jumps, which is mathematically flawed as a model assumption.

A multi-factor extension sometimes retains the efficiency of the one-factor version in the pricing of European and VIX derivatives. But the interplay of factors between VIX and VVIX are still implicit. Moreover, only \cite{yuan2020new} and \cite{bardgett2019inferring} did long-term in-sample calibration over consecutive years, and only \cite{yuan2020new} did out-of-sample tests. Due to the inherent complexity of multi-factor models, the high number of parameters often complicates the optimization process. In practice, the best-performing implementations typically require between 10 and 25 parameters. By comparison, our model achieves strong performance competitive with \cite{yuan2020new} while staying relatively parsimonious using only 13 parameters.
 
 Finally, there is some recent research that characterizes the volatility relationship using a non-parametric framework. \cite{guo2022joint} introduced a time-continuous optimal transport formulation of the joint calibration problem, where a semimartingale price process is assumed, the optimization takes the current market states and implied volatility data as inputs and a calibrated diffusion process as an output. Additionally, \cite{guyon2020joint} built a non-parametric discrete-time model and transformed the joint calibration into a dispersion-constrained martingale transport problem which is solved using the Sinkhorn algorithm. The approach claimed to achieved exact joint calibration results. Meanwhile, \cite{guyon2022neural} attacked the joint calibration problem by assuming a local volatility model and modeling volatilities and drifts as neural networks. Furthermore, \cite{dong2025joint} introduced a nonparametric discrete-time method to the joint calibration problem called the joint implied willow tree, where they extract the risk-neutral process directly from market data by solving the transition probabilities.  Finally, \cite{cuchiero2025joint} considers a stochastic volatility model where the dynamics of the volatility are described by a signature model and obtains highly accurate calibration results for both markets.

In this line of approaches, the relationship between VIX and VVIX are even more implicit than parametric models, and calibrations are often time-consuming. In addition, all of the mentioned works did calibrations restricted on a few dates, and none did out-of-sample tests. And the number of parameters in the optimization problem are much larger than parametric models. In the method of \cite{guo2022joint} and \cite{guyon2020joint}, the number is comparable to the aggregate number of options in both markets; in \cite{dong2025joint}, the number scales with $T$ (number of maturities) $\times$ $m$ (number of discrete values, around $50$); and in \cite{cuchiero2025joint}, the signature model is implemented by calibrating a coefficient of dimension $85$.

In response to the limitations of the existing works, we propose composite time change models.
\subsection{Decoupling VIX and VVIX via Composite Time Change}

To study the volatility relationship in composite time change (CTC) models, we first introduce some background knowledge on time change approaches. Time changes can be interpreted as the intensity of business activities that drives the variation in volatility of an asset. The original clock ${t, t \ge 0}$ is referred to as calendar time, and the time change process $T$ is called business time. The log-price process of an asset is originally thought to be stationary and of independent increment, i.e., a Lévy process $L$. Every time a market event occurs and drives the variation in volatility, the change is reflected in the time change, either by accelerating or by slowing the business clock. The real market price of the asset is then updated under the business clock, namely $L_T$, where $L$ is often called a base process.

Originally, \cite{clark1973subordinated} and \cite{geman2001time} proposed a subordinated Brownian motion model for the log price. The time change introduces jumps in volatility. \cite{carr2003stochastic} and \cite{carr2004time} introduced time-changed Lévy models, where the time change is absolutely continuous, through which stochastic volatility is introduced. \cite{luciano2006multivariate} and \cite{eberlein2009correlating} modeled the dependence of multi-assets via correlation of subordinated Brownian motions. \cite{mendoza2010time}
 extended the time-changed L\'{e}vy model by considering the combination of time
 changes and a certain type of composite time change. But the time changes considered in the composite structure are both L\'{e}vy subordinators, and are independent of the L\'{e}vy process. A multi-asset TCLM is developed in \cite{ballotta2016multivariate}. And \cite{cui2019general} extended the theory of time-changed L\'{e}vy process to a general time-changed Markov process. Recently, \cite{ballotta2022smiles} established a unified TCLM structure that allows leverage via diffusion as well as jumps. Applications of time change models include electricity price modeling (cf. \cite{borovkova2017electricity}), commodity derivative pricing (cf. \cite{li2014time}), swap derivatives (cf. \cite{itkin2010pricing}), among many others.

Regarding the literature on composite time change models, as far as we know, it is only mentioned in \cite{mendoza2010time}, where $U_t$ and $V_t$ are L\'{e}vy subordinators independent of each other as well as independent of the base process. It can be seen that the composite structure is theoretically redundant in the sense that $T_t = U_{V_t}$ is again a L\'{e}vy subordinator. The composite structure there was used to exhibit flexible moments of the distribution in practice. In this paper, however, we will consider a different class of composite time changes, where $U_t$ and $V_t$ are absolutely continuous and can correlate with the base L\'{e}vy process.

A time change is a non-decreasing and right-continuous process $\{T_t: t\geq 0\}$ that satisfy $T_0 = 0$, $T_t \uparrow \infty$ as $t\uparrow \infty$ and $T_t$ is a stopping time for every $t\ge 0.$ We define a composite time change (with respect to $\{\mathcal{F}_t\}_{t\ge 0}$) as $T_t = U_{V_t},$
    where $U_t$ is an $\mathcal{F}$-time change and $V_t$ is an $\mathcal{F}_U$-time change. Composite time change models considered in this paper belong to the class of time-changed L\'{e}vy models introduced in \cite{carr2004time}. However, the introduction of a general composite structure is novel in literature. 

In a time change asset price model, $X_t = L_{T_t}$, where L\'{e}vy process $L$ satisfies that $e^L$ is a  martingale with mean $1$ and time change $T$ has $E[T_t] < \infty$ for every $t\ge 0$. Next, we show how a composite change of time helps decouple VIX and VVIX. Under the assumptions specified in Appendix \ref{re_VIX}, the argument in Eq. \eqref{arg1} does not apply to composite time change models. Instead, it is shown by Proposition \ref{prop:VIX} that
$$\operatorname{VIX}_t^2 =-2E[L_1] u_{V_t}v_t + O(\bar{\tau})$$
and by Proposition \ref{prop:VVIX} that
\begin{equation}
    \label{eq:contribution}
    \operatorname{VVIX}_t^2 =2r- \mathcal{A}^v \ln (v_t) - \mathcal{A}^u \ln (u_{V_t}) v_t + o(1)
\end{equation}
for small $\bar \tau$. These expressions of VIX and VVIX hold for the CTC models in Eq. \eqref{CTC} for sufficiently regular $u$ and $v$. Thus, we observe that $\operatorname{VIX}^2$ is a multiplication of two volatility factors and that $\operatorname{VVIX}^2$ is a summation of functions of the two factors. The two factors can thus jointly adjust their values to match both markets based on their explicit relationship to VIX and VVIX. We expect that the calibration performance of composite time change models will improve significantly from their ordinary one-factor version.

As a special case, in Heston model, we have $u_t \equiv 1$ and 
$$\operatorname{VIX}_t^2 = v_t + O(\bar\tau), \quad \operatorname{VVIX}_t^2 =2r + \kappa_v + \frac{\eta_v^2 - 2\kappa_v\theta_v}{2 v_t} + o(1),$$
which admits an unfavorable inverse relationship between VIX and VVIX. 
Moreover, for Composite Heston model introduced in Eq. \eqref{Composite Heston},
$$\operatorname{VIX}_t^2 = u_{V_t}v_t + O(\bar{\tau}), \quad \operatorname{VVIX}_t^2 =2r + \kappa_v + \kappa_u v_t + \frac{\eta_v^2 - 2\kappa_v\theta_v}{2 v_t} +\frac{\eta_u^2 - 2\kappa_u\theta_u}{2 u_{V_t}}v_t + o(1).$$
In this case, factor $v$ does not necessarily contribute negatively to VVIX and factor $u$ can account for the non-linear relationship. This provides flexibility in joint calibration while preserving the mathematical tractability of CIR process. The exact expressions for VIX in affine specifications of $U$ and $V$ are also given in Eq. \eqref{VIX for CTCM}.

\subsection{Contribution of the Paper}
This article explored the theoretical development and practical applications of composite time-changed Lévy models, with a particular emphasis on their use in derivative pricing in the joint market of SPX options and VIX options.

Firstly, we summarized the theory behind time-changed Lévy models, building upon the pioneering work of \cite{carr2004time}. This section emphasized the technique of leverage-neutral measure change that facilitates the computation of characteristic functions of the model's log return.

Then we proposed CTC models, a novel approach to modeling financial variables by considering composite time change processes characterized as \(X_t = L_{U_{V_t}}\). In this framework, \(U\) and \(V\) are increasing continuous processes with activity rates \(u\) and \(v\). The models are defined under the risk-neutral measure \(\mathbb{P}\) and filtration \(\mathcal{F}^X_t := (\mathcal{F}_{U})_{V_t}\), where filtration $\mathcal{F}_t$ satisfies that $U$ is an $\mathcal{F}$-time change and $V$ is an $\mathcal{F}_U$-time change.

A novelty of this model is the flexibility to incorporate different types of stochastic processes for time changes, allowing specific adaptations to a variety of financial phenomena observed in markets. For example, the model degenerates to ordinary time-changed L\'{e}vy models for \(U_t \equiv t\), which incorporates most of the specifications in \cite{carr2004time} and \cite{ballotta2022smiles}. The proposed Composite Heston model assumes \(V\) to be an integrated Cox-Ingersoll-Ross (CIR) process, which makes \(V\) naturally increasing and has many computational conveniences, allowing for Fourier-based numerical methods as well as exact simulation techniques as shown in the derivative pricing section. Additionally, models like Composite JH incorporate jumps via infinite-activity Lévy processes and generate the leverage effects via co-jumps. The randomness of the second time change \(V\) further adds flexibility to the vol-of-vol structure and significantly improves the performance in the joint calibration problem compared with the JH model proposed in \cite{ballotta2022smiles}.

Moreover, we also derived the characteristic function of the underlying's log return for these composite time-changed processes, which is pivotal for theoretical and practical applications. Firstly, we defined a leverage neutral measure on the original filtration. We then obtained the characteristic function of the log return based on a composite time change to the Radon-Nicodym derivative. Likewise, the dynamics of activity rates \(u\) and \(v\) are also derived by considering appropriate time changes to the Radon-Nicodym derivative. The result allows for better understanding and integration of time change models and broadens the traditional concept of affine volatility models through the analytical expression of the characteristic function under the affine specification of the composite time change. By introducing and integrating time changes, the CTC models provide more adaptive and efficient tools for financial modeling.

In the context of derivative pricing under CTC models, we developed an efficient cosine expansion method for European options that maintains a computational framework similar to that in ordinary time-changed Lévy models. This application allows for the separation of strikes and the underlying asset, ensuring that the number of strike prices does not add to the computational complexity when pricing the entire volatility surface.

The complexity analysis for the developed cosine expansion method reveals that the overall computation cost is \(O(ND)\) under affine specifications, where \(N\) represents the number of terms in the cosine series, and \(D\) denotes the degree of discretization of numerical integration, and is of the same order of computational expense as their ordinary time change counterparts, such as the Heston model and Jump Heston models. Such an efficiency remains as long as the Laplace transforms of the time changes are explicit.

We also developed an efficient exact simulation scheme for VIX pricing. In contrast, the pricing becomes more complex for VIX derivatives. Traditional Fourier-based numerical methods used for ordinary time-changed Lévy models do not extend to CTC models due to the nonlinear relationship between the spot VIX and the activity rates. To address this challenge, a novel exact simulation algorithm was proposed. By assuming time change \(V\) to be an integrated CIR process, the conditional distribution \(\left(V_{\bar{T}} \mid v_0, v_{\bar{T}}\right)\) can be efficiently simulated according to known algorithms. This fast simulation technique is applied to the simulation of the activity rates in CTC models, where the algorithm directly simulates the terminal distributions and avoids simulating the whole path.

In the empirical analysis, we carried out a joint calibration utilizing real-market data from SPX and VIX options over two periods: an in-sample period from April 4, 2007, to April 1, 2015, and an out-of-sample period from April 8, 2015, to December 27, 2017. We found that the CTC models successfully captured the joint volatility smiles in both SPX and VIX markets. We observed an improvement of 32.62\% in-sample error for Heston and 23.00\% for JH when adding a composite time change. The pairwise \(t\) hypothesis testing also validated the superiority of CTC models over their ordinary counterparts. We compared the CTC models with those proposed in \cite{yuan2020new} and saw a better performance of Composite JH (RMSRE: 0.0482, 13 parameters) over the best performed model (RMSRE: 0.0581, 21 parameters) in the referenced work. 

In addition, the out-of-sample results also demonstrated the superiority of CTC models. We perform the out-of-sample test in three periods: 2015/3/31 - 2016/3/13, 2016/3/14 - 2017/2/5, 2017/2/6 - 2017/12/31. In each period, we estimate the structural parameters from the estimates of the last eight weeks before the period. Then, daily calibrations are performed by fixing the structural parameters. We found that the composite time change significantly improves the out-of-sample performance, with the RMSRE of Composite JH: 0.0837, 0.0894, 0.1000, respectively, compared to those of JH: 0.1113, 0.1451, 0.1723.

In addition, we displayed the performance of the joint calibration within a series of moneyness and maturity ranges. We showed that the composite time change improves not only the models' aggregate performance, but also the performance in most ranges of options, especially for VIX options, where the improvements are more substantial. 

We also evaluated the calibrated CTC models in the following aspects: implied volatility characteristics, short-term near-the-money performance, contribution of time change to vol-of-vol. Firstly, we showed that composite structures improve the fit of ATM IV and volatility skew, which is defined as the slope of the two implied volatilities with moneyness closest to 0.7 (0.8) and 1.2 (1.5) for SPX (VIX) option markets. Especially in VIX option markets, the daily RMSE of short-term ATM IV is reduced from 0.2668 to 0.1589 for Composite Heston and from 0.0889 to 0.0338 for Composite JH, and the daily RMSE of short-term volatility skew is reduced from 0.6888 to 0.5054 for Composite Heston and from 0.2855 to 0.1568 for Composite JH. In addition, the calibration results of Composite JH were found to be the most balanced in the joint markets, with the RMSE of short-term ATM IV to be 0.0203 (SPX) versus 0.0338 (VIX) and short-term skew 0.1098 (SPX) versus 0.1568 (VIX), in sharp contrast to those of JH (0.0234 vs. 0.0899 for short-term ATM IV,
 0.1335 vs. 0.2855 for short-term volatility skew).

Then we compared the in-sample fit of short-term near-the-money implied volatilities in the joint markets. SPX (VIX) options with maturities smaller than 90 (45) days and moneyness spanning $[0.9, 1.1]$ ($[0.8, 1.2]$) are selected. The experiment is dedicated to comparing the fit of decay rates in short-term ATM skew. Considering numerical stability, the near-the-money RMSRE was computed instead of the ATM skew. We found an improvement in performance in the SPX markets of 21. 61\% for Composite Heston and 18. 90\% for Composite JH. More predominantly in VIX markets, the in-sample error improves by 50.42\% for Composite Heston and by 56.90\% for Composite JH. Moreover, we plotted the daily near-the-money fitting error in the SPX and VIX option market, respectively. In both markets, the composite model was found to prevail on most calibrated dates. 

Finally, we found contribution of time change \(V\) to VVIX in CTC models to be steady and sustained. We evaluated the contribution from the linear decomposition of vol-of-vol using calibrated parameters. In Composite JH (Composite Heston), the average contribution is 19.2\% (22.7\%) and a correlation of 0.371 (0.142) is found between the contributions and VVIX values.

The article is organized as follows. In section 2, we summarize the theory of time-changed L\'{e}vy models, and we show how the technique of leverage neutral measure change helps form the characteristic function; Section \ref{CTC model} develops the theory of composite time change models, where a general form is considered and characteristic fuctions derived. We also discuss some useful specifications of CTC models. In Section 4, we introduce the application of the model in derivative pricing, including European options and VIX options. We show that the European option pricing in CTC models can be conducted quite efficiently. In section 5, we perform real-market joint calibration and discuss the results. The last section concludes.

\section{Preliminaries: time-changed L\'{e}vy processes}
This section establishes the theoretical framework for time-changed Lévy processes. We first review Lévy processes and their characteristic representations, including the special case of subordinators. Subsequently, we examine time-changed Lévy processes under three progressively stronger regularity conditions and introduce the leverage-neutral measure technique to handle dependent time changes. The characteristic function of the time-changed process is derived, and key properties under measure changes are discussed.
\subsection{L\'{e}vy processes}
\begin{definition}
    Given a filtered probability space $\left(\Omega, \mathcal{F},\left\{\mathcal{F}_t\right\}_{t \geq 0}, \mathbb{P}\right)$, a L\'{e}vy process, $L = \{L_t, \, t\ge 0\}$, on $\mathbb{R}^d$ is a continuous-time process with independent and stationary increments. The characteristic function of $L_t$ is denoted as $\phi_L(m ; t):= E^\mathbb{P}[e^{\mathrm{i}m\cdot L_t}]=e^{t \Psi_L(m)}, m \in \mathbb{R}^d$, where $\Psi_L(m)$ is the  characteristic exponent of $L$:
$$
\Psi_L(m)=\mathrm{i} \alpha \cdot m-\frac{1}{2}m \cdot \Sigma m +\int_{\mathbb{R}^d}\left(e^{\mathrm{i} m\cdot x}-1-\mathrm{i} m\cdot x 1_{\{|x| \leq 1\}}\right) \nu(\mathrm{d} x),
$$
where $\alpha \in \mathbb{R}^d$, $\Sigma$ is a symmetric positive semi-definite $d \times d$ matrix, and $\nu$ is a positive measure on $\mathbb{R}^d$ such that $\nu(\{0\})=0$, $\int_{\mathbb{R}^d}\left(|x|^2 \wedge 1\right) \nu(\mathrm{d} x)<\infty$.
\end{definition}
The triplet $\left(\alpha, \sigma^2, \nu(\mathrm{d} x)\right)$ determines the L\'{e}vy process and is referred to as differential characteristics. In particular, a subordinator is a L\'{e}vy process such that $t \mapsto L_t$ is non-decreasing.

\begin{definition}
    Given a L\'{e}vy process $L$, we denote by $L \in \mathcal{H}^2(\mathbb{R}^d)$ if $L$ is integrable and $L_t - E[L_1]t$ is a square-integrable martingale, i.e., $$\sup_{t\in \mathbb{R}_+} E(L_t - E[L_1]t)^2 < \infty.$$
    We denote by $L \in \mathcal{H}^2_{\text{loc}}(\mathbb{R}^d)$ if $L$ is integrable and $L_t - E[L_1] t$ is a locally square integrable martingale, i.e., there exists a sequence of stopping times $\{T_n\}$ increasing to $\infty$ such that, for every $n$, $L_{\cdot \wedge T_n} \in \mathcal{H}^2(\mathbb{R}^d).$
\end{definition}

\subsection{Time-changed L\'{e}vy Processes}
Leverage neutral measure, first introduced in \cite{carr2004time}, is a complex-valued measure change technique that enables the explicit computation of the characteristic function of time-changed L\'{e}vy process $L_T$, especially when $L$ and $T$ are not independent.

\begin{definition}
Given L\'{e}vy process $L$ on $\mathbb{R}$, there are three versions of conditions for time change $T$ in the order of restrictiveness.
\begin{enumerate}
    \item (Type 1) $T$ is integrable: $E[T_t] < \infty$ for all $t\ge 0$.

\item (Type 2) The condition of type 1 holds and $T$ is in synchronization with $L$, i.e. $L$ is constant a.s. on the interval $[T_{t-}, T_t]$ for all $t > 0$ (cf. \cite{jacod1979calcul}.)

\item (Type 3) The condition of type 1 holds and the time change $T_t = \int_0^t v_s \mathrm{d}s$ is absolutely continuous with $v_t > 0$ a.s. for all $t \ge 0$.
\end{enumerate}
\end{definition}

From the definition, it can be directly verified that type 3 implies type 2. 

For now, we assume time change $T$ is of type 1. If $T$ is independent of $L$, then the characteristic function of $X_t = L_{T_t}$ is
    $$\phi_{X}(m;t) := \mathrm{E}^\mathbb{P}\left[\exp{(\mathrm{i} m X_t)}\right] = \phi_{T}(-\mathrm{i}\Psi_L(m); t), \quad m\in \mathbb{R}.$$
Otherwise, we define a class of complex-valued measures $\mathbb Q(m)$, which we denote by $\mathbb{Q}$ for brevity if there is no confusion, is absolutely continuous with respect to $\mathbb P$ and is defined by
\begin{equation}
\label{Q}
\left.\frac{\mathrm{d} \mathbb Q(m)}{\mathrm{d}\mathbb P}\right|_{\mathcal{F}_{t}} = M_t^m,
\end{equation}
with
$$
M_t^m := \exp \left(\mathrm{i} m L_t-t \Psi_L(m)\right).
$$
Then it follows that, under $\left.\mathbb{Q}(m)\right.|_{\mathcal{F}_{T_t}}$,
    \begin{equation}
        \label{Re 1}
\phi_{X}(m; t) = E^\mathbb{Q}[\exp{(\mathrm{i}mX_t)}(M^m(T_t))^{-1}] = \phi_{T}^\mathbb{Q}\left(-\mathrm{i}\Psi_L(m); t\right),\quad m\in \mathbb{R},
\end{equation}
where $\phi_T^\mathbb{Q}(u; t):= E^\mathbb{Q}[\exp{(\mathrm{i}uT_t)}]$ is the (generalized) characteristic function of $T_t$ under measure $\mathbb{Q}(m)$. 

\begin{remark}
Readers can refer to \cite{carr2004time}, \cite{huang2004specification} and \cite{ballotta2022smiles} for more details on the definition of leverage-neutral measures. In the works above, $\mathbb{Q}(m)$ is defined on filtration $\{\mathcal{F}_{T_t}\}_{t\ge 0}$, which coincides with our defined leverage neutral measure on $\mathcal{F}_{T}$.
\end{remark}

Stronger conditions for $T$ are needed to explicitly characterize the time-changed process $X$. Specifically, if $T$ is of type 2 with $\mathrm{d}T_{t-} = v_t \mathrm{d}t$, then semimartingale $X$ has local characteristics \begin{equation}
\label{TC characteristics}(\alpha v_{t}, \sigma^2 v_{t}, v_t \nu(\mathrm{d}x)),\end{equation}
given that $(\alpha, \sigma^2, \nu(\mathrm{d}x))$ is the differential characteristics of $L$ (cf. \cite{kuchler2006exponential}).

\section{Composite Time Change Models}
\label{CTC model}
In this section, we assume $U$ and $V$ are independent increasing continuous processes with activity rates $u_t$ and $v_t$, respectively, i.e., $\mathrm{d}U_t = u_t \mathrm{d}t$, $\mathrm{d}V_t= v_t \mathrm{d}t$ with $u_t >0$, $v_t > 0$ a.s. We further require $E[U_t] < \infty$, $E[V_t] < \infty$, and $E[U_{V_t}] < \infty$ for every $t\ge 0$. Then we have:
\begin{proposition}
\label{prop:adapt}
    If there exists a filtration $\{\mathcal{F}_t\}_{t\ge 0}$ satisfying usual conditions (i.e., right-continuous and complete) such that $U$ is a $\mathcal{F}$-time change and $V$ is a $\mathcal{F}_U$-time change, then $u_{V_t}$ and $v_t$ are adapted to $\{(\mathcal{F}_{U})_{V_t}\}_{t\ge 0}$, i.e., the filtration of $\mathcal{F}_{U}$ time-changed by $V$, and $U_V$ is a type-3 $\mathcal{F}$-time change.
\end{proposition}
\begin{proof}{Proof}
    Since $V$ is a $\mathcal{F}_U$-time change and $U$ is $\mathcal{F}_U$-adapted, we have that $V_t$ and $U_{V_t}$ are adapted to $(\mathcal{F}_U)_{V_t}$ for any $t\geq 0$. Since adaptedness is preserved under limits, $u_{V_t}v_t$ and $v_t$ are $(\mathcal{F}_U)_{V}$-adapted, which leads to the adaptedness of $u_{V_t}$ to $(\mathcal{F}_U)_{V_t}$.

    Next, we show that $U_{V}$ is a type-3 $\mathcal{F}$-time change. Since $V$ is a $\mathcal{F}_U$-time change, for any $r, s, t\ge 0$, the set $\{V_r \le s\}\cap \{U_s \le t\} \in \mathcal{F}_t.$ Thus,
    $$\begin{aligned}
        \{U_{V_r} < t\}
        & =  \{V_r < U^{-1}(t)\}\\
        &= \bigcup_{s \in \mathbb{Q}_+, s < U^{-1}(t)}\{V_r \le s\}\\
        &= \bigcup_{s \in \mathbb{Q}_+}\{V_r \le s\} \cap \{U_s < t\}\\
        & = \bigcup_{s \in \mathbb{Q}_+} \bigcup_{n \ge 1}\{V_r \le s\} \cap \{U_s \le  t -\frac{1}{n}\}\in \mathcal{F}_t,
    \end{aligned} $$
    where the first equality holds since $U_t$ is continuous and strictly increasing. Since $\mathcal{F}_t$ is right-continuous, $\{U_{V_r} \le t\} \in \mathcal{F}_t$. So, combined with $u_{V_t}v_t > 0$ and the assumption $E[U_{V_t}] < \infty$, $U_{V}$ is a 
    type-3 $\mathcal{F}$-time change. 
\end{proof}

If there exists a filtration $\mathcal{F}$ such that the assumption of Proposition \ref{prop:adapt} holds, then $U$ and $V$ are both type-3 $\mathcal{F}$-time changes. Let $L_t$ be an $\mathcal{F}$-adapted L\'{e}vy process such that $\exp{(L_t)}$ is a unit-mean $\mathcal{F}_t$-martingale. We denote $\mathcal{F}^X_t := (\mathcal{F}_{U})_{V_t}$ and let $X_t := L_{U_{V_t}}$ be our composite time change (CTC) model. We also denote $E_t\left[\cdot\right]$ the expectation conditional on $\mathcal{F}^X_t$. 

In Section \ref{model theory}, we derive the characteristic function of the time change model with a particular focus on affine specifications. And in Section \ref{specifications}, we propose two tractable implementations: a Composite Heston model and a Composite Jump Heston model (an extension of a Jump Heston model with co-jumps, cf. \cite{ballotta2022smiles}).

\subsection{Model Theory}
\label{model theory}
In the specifications of \cite{carr2004time}, the activity rate of time change process is assumed to be affine with possibly general jump specifications, which can be summarized (with the volatility factor of dimension $d = 1$) in the following form: given filtration $\mathcal{F}$ and $\mathcal{F}^X$ defined as above with $V_t \equiv t$, the time change $T_t=U_t$. The activity rate $u_t$ of $U_t$ is assumed to follow
\begin{equation*}
             \begin{array}{lr}
             \mathrm{d} u_t = \kappa_u(\theta_u - u_t)\mathrm{d}t + \sqrt{\sigma_1 + \sigma_2 u_t} \mathrm{d}\tilde B_t +  \mathrm{d}\tilde J^u_t, &\\ 
             \end{array}
\end{equation*}
where $\mathrm{d}U_t = u_t \mathrm{d}t$, $\tilde B$ is a $(\mathbb P, \mathcal{F}^X)$-Brownian motion, $\tilde J^u$ is a $(\mathbb P, \mathcal{F}^X)$-adapted Poisson jump component with jump intensity $I(u_t) = \eta_1 + \eta_2 u_t$ and random jump magnitude $q$ or a time-changed L\'{e}vy process with activity rate $I(u_t)$, and $B$ is independent of $J^u$, with $\theta_u, \kappa_u \in \mathbb{R}$, $\sigma_1, \sigma_2, \eta_1, \eta_2 \in \mathbb{R}_+$. In a time-changed form, the SDE of $u_t$ can be equivalently expressed as
$$\mathrm{d}u_t = \kappa_u(\theta_u - u_t)\mathrm{d}t + \mathrm{d} B(\sigma_1 t + \sigma_2 U_t) + \mathrm{d} J^u(\eta_1 t + \eta_2 U_t),$$
where $B$ is a $(\mathbb{P}, \mathcal{F})$-standard Brownian motion and $J^u$ is a $(\mathbb{P}, \mathcal{F})$-L\'{e}vy process. A similar model setup can be found in \cite{ballotta2022smiles}, where a miscellaneous activity rate form is assumed: $u_t = u_0 + Y_0(t) + Y_1(U_t)$ with L\'{e}vy processes $Y_0$, $Y_1$.

In the following, we introduce the CTC model, which is more adapted to the VIX-VVIX relationship, for the joint pricing problem. As shown in the introduction, the above-mentioned ordinary time change models result in restrictive relationship of VIX and VVIX. However, the two volatility factors in a CTC model can jointly adjust their values in an explicit way and match joint market prices.

As an extension to ordinary time changed L\'{e}vy models, we adopt a jump-diffusion stochastic volatility setup of activity rates $u_t$ and $v_t$, which is general enough in application. And combined with the setup in the beginning of this section, a risk-neutral CTC model $X_t = L_{U_{V_t}}$ is specified with
\begin{equation}
\label{CTC}
\begin{aligned}
& \mathrm{d} u_t = \alpha^u(u_{t})\mathrm{d}t + \beta^u(u_{t})\mathrm{d} Z(U_t) + \gamma^u(u_{t-}) \mathrm{d} J^u(U_t), \\
& \mathrm{d} v_t = \alpha^v(v_{t})\mathrm{d}t + \beta^v(v_{t})\mathrm{d}\tilde B(V_t) + \gamma^v(v_{t-}) \mathrm{d}\tilde{J}^v(V_t)\\
& \mathrm{d} L_t = -\Psi(-\mathrm{i}) \mathrm{d}t + \sigma \mathrm{d}W_t + \mathrm{d}J_t,
\end{aligned}
\end{equation}
with constants $u_0 > 0$, $v_0 > 0$. Here $Z$ is a $(\mathbb P, \mathcal{F})$-standard Brownian motion, $\tilde B$ is a $(\mathbb P, \mathcal{F}_U)$-standard Brownian motion, $\tilde{J}^v \in \mathcal{H}^2_{\text{loc}}(\mathbb{R})$ is a  $(\mathbb P, \mathcal{F}_U)$-subordinator, and $\{Z, \tilde B\}$ are independent of $\tilde{J}^v$. Meanwhile, functions $\alpha^i(\cdot)$, $\beta^i(\cdot)$, and $\gamma^i(\cdot) \ge 0$ ($i = u, v$) satisfy regularity conditions for the existence of strong positive solutions of $u_t$ and $v_t$, with $u_t$ adapted to $\mathcal{F}_{U_t}$ and $v_t$ adapted to $\mathcal{F}^X_t$. We define $B = \{B_t, t\ge 0\}$ with $B_t := \int_0^{\hat{U}_t} \sqrt{u_s}\mathrm{d} \tilde B(s)$ with $\hat{U}_t := \inf\{s: U_s > t\}$, which is an $\mathcal{F}$-standard Brownian motion by L\'{e}vy's characterization theorem. We assume $\langle Z, B\rangle_t = 0$ and $\tilde{J}^v$ is independent of $J^u$. Thus, $u$ and $v$ are independent. Proposition \ref{prop:TC SDE} in Appendix \ref{CTC argument3} justifies the existence of $u$ and $v$ in model \eqref{CTC} as well as the assumption in Proposition \ref{prop:adapt}. In particular, Corollary \ref{coro:TC SDE} gives explicit conditions on $\alpha^i(\cdot)$, $\beta^i(\cdot)$, and $\gamma^{(\cdot)}$ ($i = u, v$) for the existence results to hold.

Moreover, for the $\mathcal{F}$-adapted L\'{e}vy process $L$, we assume $(J, J^u) \in \mathcal{H}^2_{\text{loc}}(\mathbb{R}^2)$ is a 2-dimensional pure-jump $(\mathbb{P}, \mathcal{F})$-L\'{e}vy process with subordinator $J^u$ and is independent of $\{Z,  B, \tilde{J}^v\}$. We construct the $(\mathbb{P}, \mathcal{F})$-Brownian motion $W$ in Eq. \eqref{CTC} through the following procedure. We assume $\rho_u, \rho_v \in [-1, 1]$ with $\rho_u^2 + \rho_v^2 \le 1$. Define an $\mathcal{F}$-Brownian motion $W$ via
\begin{equation}
    \label{W}
\mathrm{d} W_t = \rho_u  \mathrm{d}Z_t + \rho_v  \mathrm{d} B_t + \sqrt{1 -\rho_u^2 -\rho_v^2}  \mathrm{d} W^\perp_t,
\end{equation}
where the $\mathcal{F}$-standard Brownian motion $W^\perp$ is independent of $\{Z, B, J, J^u, \tilde{J}^v\}$. Thus, it holds that $W$ is $\mathcal{F}$-adapted, $\langle W, Z\rangle_t = \rho_u t$, $\langle W, B\rangle_t = \rho_v t$, and $W$ is independent of $\{J, J^u, \tilde{J}^v\}$. Finally, we require $\exp{(J_t)}$ to be integrable and $\Psi(m) :=-\frac{m^2}{2} \sigma^2+\Psi_{J}\left(m\right)$ is the characteristic exponent of $\sigma W_t + J_t$, which renders $\exp{(L_t)}$ a unit-mean martingale. 

To summarize, in CTC model \eqref{CTC}, the Brownian motion components $\{W, Z, B\}$ and the pure-jump components $\{J, J^u, \tilde{J}^v\}$ are independent. $W$ and $Z$ ($B$) can be dependent through $\mathrm{d}\langle  W,  Z\rangle_t = \rho_u\mathrm{d}t$ ($\mathrm{d}\langle W, B\rangle_t = \rho_v\mathrm{d}t$), but $Z$ and $B$ are independent. $J$ and $\tilde{J}^v$, $J^u$ and $\tilde{J}^v$ are independent, but $J$ and $J^u$ can be dependent. Consequently, the leverage effect between stochastic volatility $\sqrt{u_{V_t}v_t}$ (i.e., $\sqrt{\mathrm{d}T_t / \mathrm{d}t}$) and the base process $L$ is incorporated through three aspects: the correlation between $L$ and $u$ via $\rho_u$ and the dependence in $(J, J^u)$, and the correlation between $L$ and $v$ via $\rho_v$.

Since we allow the coefficients to be a function of $u$ and $v$, the setup covers the affine specifications in \cite{carr2004time} (with $\sigma_1 = \eta_1 = 0$) and \cite{ballotta2022smiles} by taking $V_t \equiv t$, $\beta^u(\cdot)\equiv \beta$ and $\gamma^u(\cdot) \equiv \gamma$.

Under the leverage-neutral measure $\mathbb{Q}(m)$ defined by Eq. \eqref{Q}, we have the following result of the characteristic function of $X_t$.
\begin{theorem}
\label{Th1}
$\mathrm{(i)}$ Under the leverage neutral measure $\mathbb Q(m)$, $m \in \mathbb{R}$, the SDE of $u_t$ is given by
\begin{equation*}
\begin{aligned}
\mathrm{d}u_t &= \left(\alpha^u(u_{t}) + \mathrm{i}\rho_u \sigma mu_{t}\beta^u(u_{t})\right)\mathrm{d}t + \beta^u(u_{t})\mathrm{d}Z^{\mathbb{Q}}(U_t) + \gamma^u(u_{t-}) \mathrm{d}J^{u}(U_t),
\end{aligned}
\end{equation*}
where $Z^{\mathbb Q}$ is a $(\mathbb Q,\mathcal{F})$-standard Brownian motion. The characteristic exponent of $J^{u}$ under $\mathbb{Q}(m)$ is given by
\begin{equation}
\label{characteristic exponent}
\Psi_{J^u}^{\mathbb Q}(x) = \Psi_{J, J^u}(m, x) - \Psi_J(m), 
\end{equation}
where $\Psi_{J, J^u}(m, x)$ is the characteristic exponent of L\'{e}vy process $(J, J^u)$. And the SDE of $v$ is given by
$$\mathrm{d}v_t = \left(\alpha^v(v_{t}) + \mathrm{i}\rho_v\sigma m \sqrt{u_{V_t}}v_t\beta^v(v_t)\right)\mathrm{d}t + \beta^v(v_{t})\mathrm{d}\tilde B^\mathbb{Q}(V_t) + \gamma^v(v_{t-})\mathrm{d}\tilde J^{v}(V_t),$$
with $\tilde B^\mathbb{Q}$ a $(\mathbb{Q}, \mathcal{F}_U)$-standard Brownian motion, and the law of $\tilde J^v$ remains unchanged under $\mathbb{Q}$.

$\mathrm{(ii)}$ If $\rho_v = 0$, then it further holds that $U$ and $V$ are independent under $\mathbb{Q}$, and the characteristic function of the composite time-changed process $X$ is given by
$$
    \phi_X(m; t) := E\left[\exp{(\mathrm{i}mX_t)}\right] = E^{\mathbb{Q}}\left[\phi_U^{\mathbb Q}(-\mathrm{i} \Psi_L(m); V_t)\right].
$$
\end{theorem}

\begin{proof}{Proof}
(i) We recall that $\Psi_L(m) = -\mathrm{i}m\Psi(-\mathrm{i}) + \Psi_W(\sigma m) + \Psi_J(m)$ and
$$\begin{aligned}
    M_t^m & = \exp{(\mathrm{i}mL_t - t\Psi_L(m))}\\
    & = \exp{(\mathrm{i}m ( -\Psi(-\mathrm{i})t + \sigma W_t + J_t) - t(-\mathrm{i}m\Psi(-\mathrm{i}) + \Psi_W(\sigma m) + \Psi_J(m))}\\
    & = \exp{\left(\mathrm{i}m (\sigma W_t + J_t) - t(\Psi_W(\sigma m) + \Psi_J(m))\right)}.
\end{aligned}$$
By the independence of Brownian motions and jump processes, the characteristic function of $W_t$ under measure $\mathbb Q$ is 
$$
\begin{aligned}
\phi_W^{\mathbb{Q}}(x;t) &= E[\exp{(\mathrm{i}xW_t )}M^m_t]\\
&= E[\exp{\left(\mathrm{i}x W_t +\mathrm{i}m (\sigma W_t + J_t) - t(\Psi_W(\sigma m) + \Psi_J(m))\right)}]\\
&= E[\exp{\left(\mathrm{i}x W_t + \mathrm{i}\sigma m W_t -t\Psi_W(\sigma m)\right)}] \\
& = \exp{\left(t(\Psi_{W}(x + \sigma m) - \Psi_W(\sigma m))\right)}\\
&= \exp{\left(-(\frac{1}{2}x^2 + \sigma m x)t\right)}.
\end{aligned}
$$
Therefore $W^\mathbb{Q}_t := W_t - \mathrm{i}\sigma mt$ is a $(\mathbb{Q}, \mathcal{F})$-standard Brownian motion. 

Next, we consider the dynamics of $u$ under $\mathbb{Q}$. By the correlation $\langle W, Z\rangle_t = \rho_u t$, decompose $W_t = \rho_u Z_t + \sqrt{1 - \rho_u^2}Z_t^\perp$, 
$$
\begin{aligned}
\phi_Z^{\mathbb{Q}}(x;t) &= E[\exp{(\mathrm{i}x Z_t )}M^m_t]\\
&= E[\exp{\left(\mathrm{i}x Z_t + \mathrm{i}\sigma m W_t -t\Psi_W(\sigma m)\right)}] \\
&= E\left[\exp{\left(\mathrm{i}x Z_t + \mathrm{i} \sigma m(\rho_uZ_t + \sqrt{1 - \rho_u^2} Z_t^\perp) -t(\Psi_Z(\rho_u \sigma m) + \Psi_{Z^\perp} (\sqrt{1 -\rho_u^2} \sigma m))\right)}\right] \\
& = \exp{\left(t(\Psi_{Z}(x + \rho_u \sigma m) - \Psi_Z(\rho_u \sigma m))\right)}\cdot \\
&= \exp{\left(-(\frac{1}{2}x^2 + \rho_u \sigma m x)t\right)},
\end{aligned}
$$
and $Z^{\mathbb Q}_t := Z_t - \mathrm{i}\rho_u \sigma m t$ is a $(\mathbb{Q}, \mathcal{F})$-standard Brownian motion.


We then compute the characteristic function of $J^u$ under $\mathbb{Q}(m)$. Under the leverage neutral measure, the characteristic function of $J^u$ becomes
$$
\begin{aligned}
\phi^\mathbb{Q}_{J^u}(x;t) & = E[\exp{\left(\mathrm{i}xJ^u_t\right)}M_{t}^m]\\
&= E[\exp{\left(\mathrm{i}x J^u_t +\mathrm{i}m (\sigma W_t + J_t) - t(\Psi_W(\sigma m) + \Psi_J( m))\right)}]\\
&= E[\exp{\left(\mathrm{i}x J^u_t + \mathrm{i}mJ_t -t\Psi_J(m)\right)}] \\
& = \exp{\left(t(\Psi_{J, J^u}(m,x) - \Psi_J(m))\right) },
\end{aligned}$$
where $\Psi_{J, J^u}(m,x)$ is the characteristic exponent of $(J, J^u)$. It follows that 
\begin{equation}
\label{che JU}\Psi_{J^u}^{\mathbb Q}(x) = \Psi_{J, J^u}(m, x) - \Psi_J(m).
\end{equation}
Thus, the $\mathbb{Q}$-dynamic of $u$ is obtained by Eq. \eqref{che JU} and by substituting the expressions for $Z^\mathbb{Q}$ into the original SDE of $u$.

Next, we consider the dynamics of $v$ and its correlation with $u$ under $\mathbb{Q}$. First, we shall compute the characteristic function of $(Z, B)$ based on relationship Eq. \eqref{W}:

\begin{equation}
\label{Z, B}
\begin{aligned}
& E^{\mathbb{Q}}[\exp{\left(\mathrm{i}xZ_t + \mathrm{i}yB_t\right)}] \\
= & E[\exp{\left(\mathrm{i}xZ_t + \mathrm{i}yB_t\right)}M_{t}^m]\\
= &E[\exp{\left(\mathrm{i}xZ_t + \mathrm{i}yB_t +\mathrm{i}m (\sigma W_t + J_t) - t(\Psi_W(\sigma m) + \Psi_J(m))\right)}]\\
= & E\left[\exp{(\mathrm{i}xZ_t + \mathrm{i}yB_t + \mathrm{i}\sigma m(\rho_u Z_t + \rho_v B_t + \sqrt{1- \rho_u^2 - \rho_v^2} W^\perp_t) - t\Psi_W(\sigma m) )}\right]\\
= &E[\exp{(\mathrm{i}(x +\rho_u \sigma m) Z_t - t\Psi_Z(\rho_u \sigma m) + \mathrm{i}(y + \rho_v\sigma m)B_t -t\Psi_B(\rho_v \sigma m) )}]\\
= & E[\exp{\left(\mathrm{i}xZ_t - t\Psi_Z(\rho_u \sigma m))\right]}E[\exp{(\mathrm{i}yB_t - t\Psi_B(\rho_v \sigma m))}]\\
= & E^\mathbb{Q}[\exp{\left(\mathrm{i}xZ_t\right)]}\exp{\left(-(\frac{1}{2}y^2 + \rho_v \sigma m y)t\right)},
\end{aligned}
\end{equation}
where the fourth equality holds because the $W^\perp$ terms cancels out by the independence of $W^\perp$ and $(Z, B)$. From the expression, we have that $B^\mathbb{Q}_t := B_t - \mathrm{i}\rho_v \sigma mt$ is a $(\mathbb{Q}, \mathcal{F})$-Brownian motion and that $Z$ and $B$ are independent under $\mathbb{Q}$. Furthermore, 
$$\begin{aligned}
\tilde B_t &= \int_0^t u_s^{-\frac{1}{2}} \mathrm{d}B(U_s)\\
&= \int_0^t u_s^{-\frac{1}{2}} \mathrm{d}(B^\mathbb{Q}(U_s) + \mathrm{i}\rho_v \sigma m U_t)\\
& =: \tilde B^\mathbb{Q}_t + \mathrm{i}\rho_v\sigma m \int_0^t \sqrt{u_s} \mathrm{d}s.
\end{aligned}$$
We claim that $\tilde B^\mathbb{Q}_t = \int_0^t u_s^{-\frac{1}{2}}\mathrm{d}B(U_s) - \mathrm{i}\rho_v \sigma m\int_0^t \sqrt{u_s}\mathrm{d}s$ is a $(\mathbb{Q}, \mathcal{F}_U)$-Brownian motion independent of $U$. Decompose $W_t =\rho_v B_t + \sqrt{1 -\rho_v^2}B_t^\perp$. Under $\left.\mathbb{Q}(m)\right.|_{\mathcal{F}_{U_t}}$,
there is
$$\begin{aligned}
    &E^{\mathbb{Q}}[\exp{(\mathrm{i}x\tilde B^{\mathbb{Q}}_t + \mathrm{i}yU_t})]\\
    =&  E[\exp{(\mathrm{i}x\tilde B^{\mathbb{Q}}_t + \mathrm{i}yU_t})M^m(U_t)]\\
    = &E[\exp{(\mathrm{i}x\tilde B^{\mathbb{Q}}_t + \mathrm{i}yU_t} +\mathrm{i}m (\sigma W(U_t) + J(U_t)) - U_t(\Psi_W(\sigma m) + \Psi_J(m))]\\
    = & E\left[\exp{(\mathrm{i}\int_0^t (x u_s^{-\frac{1}{2}} + \rho_v\sigma m)\mathrm{d}B(U_s) +x\rho_v \sigma m\int_0^t \sqrt{u_s}\mathrm{d}s )} \right.\\
    & \quad \quad \left.\cdot \exp{\left(\mathrm{i}yU_t +\mathrm{i}m (\sigma \sqrt{1 - \rho_v^2}B^\perp(U_t) + J(U_t)) - U_t(\Psi_W(\sigma m) + \Psi_J(m))\right)}\right].
\end{aligned}$$
We then proceed with an iterated conditioning argument. By conditioning on $\{B^\perp, U, J\}$ and using the independence of $B$ and $\{B^\perp, U, J\}$,
$$
\begin{aligned}
& E\left[\exp{\left(\mathrm{i}\int_0^t (xu_s^{-\frac{1}{2}} + \rho_v\sigma m)\mathrm{d}B(U_s)\right)}\mid \{B^\perp, U, J\}\right] \\
=& E\left[\exp{\left(-\frac{1}{2}\int_0^t(xu_s^{-\frac{1}{2}} + \rho_v \sigma m)^2 u_s \mathrm{d}s\right)}\right]\\
=& E\left[\exp{\left(-\frac{1}{2}\int_0^t(x^2 + 2x\rho_v \sigma m\sqrt{u_s} + \rho_v^2\sigma^2m^2u_s)\mathrm{d}s\right)}\right]\\
= & \exp{(-\frac{1}{2}x^2t)}E\left[\exp{\left(-x\rho_v\sigma m \int_0^t\sqrt{u_s}\mathrm{d}s+U_t\Psi_B(\rho_v\sigma m) \right)}\right].
\end{aligned}$$
Thus, the characteristic function is reduced to
$$\begin{aligned}
&E^{\mathbb{Q}}[\exp{(\mathrm{i}x\tilde B^{\mathbb{Q}}_t + \mathrm{i}yU_t})]\\
= & E\left[E[\exp{(\mathrm{i}\int_0^t (x u_s^{-\frac{1}{2}} + \rho_v\sigma m)\mathrm{d}B(U_s))}] \mid \{B^\perp, U, J\}\right.\\
& \quad \quad \left.\cdot \exp{\left(x\rho_v \sigma m\int_0^t \sqrt{u_s}\mathrm{d}s + \mathrm{i}yU_t +\mathrm{i}m (\sigma \sqrt{1 - \rho_v^2}B^\perp(U_t) + J(U_t)) - U_t(\Psi_W(\sigma m) + \Psi_J(m))\right)}\right]\\
= & \exp{(-\frac{1}{2}x^2t)}E\left[\exp{(\mathrm{i}yU_t +\mathrm{i}m (\sigma \sqrt{1-\rho_v^2}B^\perp(U_t) + J(U_t)) - U_t(\Psi_W(\sqrt{1-\rho_v^2}\sigma m) + \Psi_J(m))}\right]\\
=& \exp{(-\frac{1}{2}x^2t)}E[E\left[ \exp{(\mathrm{i}\rho_v\sigma m B(U_t) - U_t\Psi_B(\rho_v\sigma m))}\mid \{B^\perp, U, J\}\right]\\
& \quad\quad \cdot\exp{(\mathrm{i}yU_t +\mathrm{i}m (\sigma \sqrt{1-\rho_v^2}B^\perp(U_t) + J(U_t)) - U_t(\Psi_W(\sqrt{1-\rho_v^2}\sigma m) + \Psi_J(m))}]\\
= & \exp{(-\frac{1}{2}x^2t)}E\left[\exp{(\mathrm{i}yU_t + +\mathrm{i}m (\sigma W(U_t) + J(U_t)) - U_t(\Psi_W(\sigma m) + \Psi_J(m)))}\right]\\
= & \exp{(-\frac{1}{2}x^2t)}E\left[\exp{(\mathrm{i}yU_t)}M^m(U_t)\right]\\
= & \exp{(-\frac{1}{2}x^2t)}E^{\mathbb{Q}}[\exp{(\mathrm{i}y U_t)}],
\end{aligned}$$
where the second equation comes from the independence of $B$ and $\{B^\perp, U, J\}$ under $\mathbb{P}$. Thus, $\tilde B^\mathbb{Q}$ is a $(\mathbb{Q}, \mathcal{F}_U)$-Brownian motion independent of $U$. It also follows that $\tilde B$ and $U$ are independent under $\mathbb{Q}$ if and only if $\rho_v = 0$.

Moreover, by considering the characteristic function of $(J^u(U_t), \tilde{J}^v_t)$ under $\left.\mathbb{Q}(m)\right.|_{\mathcal{F}_{U_t}}$,
\begin{align}
& E^{\mathbb{Q}}[\exp{\left(\mathrm{i}xJ^u(U_t) + \mathrm{i}y\tilde{J}_t^v\right)}] \notag\\
=& E^{\mathbb{Q}}[\exp{\left(\mathrm{i}xJ^u(U_t) + \mathrm{i}y \tilde{J}_t^v\right)} M_{U_t}^m]] \notag\\
= &E[\exp{\left(\mathrm{i}xJ^u(U_t) + \mathrm{i}y \tilde{J}_t^v +\mathrm{i}m L(U_t) - U_t \Phi_L(m)\right)}]\notag\\
= & E[\exp{(\mathrm{i}xJ^u(U_t) + \mathrm{i}mL(U_t) - U_t\Psi_L(m))}] E[\exp{(\mathrm{i}y \tilde{J}_t^v)}] \notag\\
= & E^{\mathbb{Q}}[\exp{(\mathrm{i}xJ^u(U_t))}]E[\exp{(\mathrm{i}y\tilde{J}_t^v)}] \label{law JV},
\end{align}
where the third equality follows from the independence of $\tilde{J}^v$ and $\{L, U, J^u\}$ under $\mathbb{P}$. Thus, $J^u$ and $\tilde{J}^v$ are independent under $\mathbb{Q}$. By letting $x = 0$ in Eq. \eqref{law JV}, the law of $\tilde{J}^v$ under $\mathbb{Q}$ is unchanged. Following a similar argument, $\tilde{J}^v$ is independent of $Z$ under $\mathbb{Q}$. We also have the independence of $B$ and $J^u$ under $\mathbb{Q}(m)$ from
\begin{align*}
& E^{\mathbb{Q}}[\exp{\left(\mathrm{i}xB_t + \mathrm{i}yJ^u_t\right)}] \\
= & E[\exp{\left(\mathrm{i}xB_t + \mathrm{i}yJ^u_t\right)}M_{t}^m] \\
= & E[\exp{\left(\mathrm{i}xB_t + \mathrm{i}m\sigma W_t - t\Psi_W(\sigma m)  + \mathrm{i}yJ^u_t + \mathrm{i}m J_t - t\Psi_J(m)\right)}] \\
= & E[\exp{(\mathrm{i}xB_t + \mathrm{i}m\sigma W_t - t\Psi_W(\sigma m))}]\cdot E[\exp{(\mathrm{i}yJ^u_t + \mathrm{i}m J_t - t\Psi_J(m))}] \\
= & E[\exp{(\mathrm{i}xB_t + \mathrm{i}m L_t - t\Psi_L(m))}]\cdot E[\exp{(\mathrm{i}yJ^u_t + \mathrm{i}m L_t - t\Psi_L(m))}] \\
= & E^{\mathbb{Q}}[\exp{(\mathrm{i}xB_t)}]E^{\mathbb{Q}}[\exp{(\mathrm{i}yJ^u_t)}].
\end{align*}
 Finally, the $\mathbb{Q}$-dynamic of $v$ is obtained by substituting the expressions for $\tilde B^\mathbb{Q}$ into the corresponding SDE.

~\\(ii) Based on the results in the proof of (i), $u$ and $v$ are independent under $\mathbb{Q}$ if $\rho_v = 0$. Thus, under $\left.\mathbb{Q}(m)\right.|_{\mathcal{F}_t^X}$, we have
$$
\begin{aligned}
    \phi_X(m; t) &= E^{\mathbb Q}[\exp{(\mathrm{i}mX_t)}\cdot (M^m(U_{V_t}))^{-1}]\\
    &= E^{\mathbb Q}[\exp{(U_{V_t}\Psi_L(m))}]\\
    &= E^{\mathbb Q}[E^{\mathbb Q}[\exp{(U_{V_t}\Psi_L(m))}\mid V_t]]\\
    & = E^{\mathbb{Q}}[\phi_U^{\mathbb Q}(-\mathrm{i} \Psi_L(m); V_t)].
\end{aligned}
$$

\end{proof}

To formulate a mathematically tractable example, we impose an affine structure on time changes, under which the CTC model \eqref{CTC} typically admits a unique, strong solution followed by Corollary \ref{coro:TC SDE}.
\begin{corollary}
\label{affine spec}
When a joint affine specification $(X_t, u_t, v_t)$ is assumed, i.e., $\alpha^u(u_t) = \kappa_u(\theta_u - u_t)$, $\alpha^v(v_t) = \kappa_v(\theta_v - v_t)$ with $\kappa_i\theta_i \ge 0$, $\beta^i(\cdot) \equiv \sigma_i \ge 0$, $\gamma^i(\cdot) \equiv \eta_i \ge 0$, $i = u,v$, and $\rho_v = 0$, the characteristic function of $X$ is explicitly expressed as follows:
$$
    \phi_X(m; t) = \frac{1}{\pi}\int_0^{+\infty} \int_0^{+\infty} \phi_U^{\mathbb Q}(-\mathrm{i}\Psi_L(m);s)\operatorname{Re}\left[e^{-z s} \phi_V(-\mathrm{i}z;t)\right] \mathrm{d} z_I\,\mathrm{d} s, \quad m\in \mathbb{R}.
$$
Here, $z = z_R + \mathrm{i}z_I$ with $z_R \le 0$ and 
\begin{equation}
\label{ODE for V}
\phi_V(m;t):= E[e^{\mathrm{i}mV_t}] = \exp{({b^V(t)v_0 + c^V(t)})},
\end{equation}
with the affine exponents $b^V(t), c^V(t)$ solutions to the system of Riccati-type ODE system
$$
    \begin{aligned}
        &b^V(t)^\prime = \mathrm{i}m - \kappa_v b^V(t) + \frac{\sigma_v^2}{2}b^V(t)^2 + \Psi_{\tilde J^v}(\mathrm{i}\eta_v b^V(t)),\\
        & c^V(t)^\prime = \kappa_v \theta_v b^V(t),
    \end{aligned}
$$
with $b^V(0) = c^V(0) = 0$. The function $\phi_U^\mathbb{Q}(-\mathrm{i}\Psi_L(m); t)$ is given by
\begin{equation}
\label{ODE for T}
\phi_U^{\mathbb Q}(-\mathrm{i}\Psi_L(m);t) = \exp{({b^U(t)u_0 + c^U(t)})},
\end{equation}
with coefficients the solutions to 
$$
\begin{aligned}
    b^U(t)^\prime &= \Psi_L(m) - \kappa^{\mathbb{Q}}_u b^U(t) + \frac{\sigma_u^2}{2}b^U(t)^2 + \Psi_{J^u}^\mathbb{Q}(\mathrm{i}\eta_u b^U(t)),\\
    c^U(t)^\prime & = \kappa_u \theta_u b^U(t),
\end{aligned}
$$
with $b^U(0) = c^U(0) = 0,$ $\kappa^{\mathbb Q}_u = \kappa_u - \mathrm{i}\rho_u\sigma_u\sigma m$, and $\Psi_{J^u}^\mathbb{Q}(\cdot)$ is the characteristic exponent of $J^u$ given by Eq. \eqref{characteristic exponent}.
\end{corollary}

\begin{proof}{Proof}
By the independence between $B$, $J^v$ and $(L, U)$ under measure $\mathbb{P}$, the distribution of $V_t$ is unchanged under $\mathbb{Q}$ by Theorem \ref{Th1}. Combined with the result of Theorem \ref{Th1} and the property of inverse generalized Fourier transform:
$$f_{V_t}(s) = \frac{1}{\pi}\int_0^\infty \operatorname{Re}[e^{-zs}\phi_V(-\mathrm{i}z; t)]\mathrm{d}z_I,$$
where $f_{V_t}$ is the density function of $V_t$ and $z = z_R + \mathrm{i}z_I$ with $z_R \le 0$, it follows that
\begin{equation}
\label{ind chf}
\begin{aligned}
    \phi_X(m; t) & = E^{\mathbb{Q}}[\phi_U^{\mathbb Q}(-\mathrm{i} \Psi_L(m); V_t)]\\
    &= E[\phi_U^\mathbb{Q}(-\mathrm{i}\Psi_L(m); V_t)]\\
    &= \frac{1}{\pi}\int_0^\infty \phi_U^{\mathbb Q}(-\mathrm{i}\Psi_L(m);s)f_{V_t}(s) \mathrm{d} s\\
    &= \frac{1}{\pi}\int_0^\infty \int_0^{\infty} \phi_U^{\mathbb Q}(-\mathrm{i}\Psi_L(m);s)\operatorname{Re}\left[e^{-z s} \phi_V(-\mathrm{i}z;t)\right] \mathrm{d} z_I \mathrm{d} s.
\end{aligned}
\end{equation}
If an affine structure is imposed for $U$ and $V$, we see from Theorem \ref{Th1} that, under $\mathbb{Q}(m)$, $u$ is still affine with $\kappa_u^\mathbb{Q} = \kappa_u - \mathrm{i}\rho_u\sigma_u\sigma m$, $\kappa^\mathbb{Q}_u\theta^\mathbb{Q}_u = \kappa_u\theta_u$, and the characteristic exponent of $J^u$ given by Eq. \eqref{che JU}. Thus, from the results of \cite{filipovic2001general} for generalized affine volatility, the generalized Fourier transform of $U_t$ is given by
$$\phi_U^{\mathbb Q}(-\mathrm{i}\Psi_L(m);t) = \exp{({b^U(t)u_0+ c^U(t)})},$$
where the coefficients $b^U(t)$ and $c^U(t)$ are given by Eq. \eqref{ODE for T}. Likewise, combined with the fact that the law of $V$ is unchanged under $\mathbb{Q}$, the characteristic function of $\phi_V$ has coefficients $b^V(t)$ and $c^V(t)$ given by Eq. \eqref{ODE for V}. Then the result follows.
\end{proof}

\subsection{Specifications}
\label{specifications}
In this section, we introduce some composite time change models that will be used in empirical research. All specifications assume that $V$ is independent of $L$ and $U$, so Theorem \ref{Th1}(ii) can be applied to compute the characteristic function of the log return.
\begin{example}[Composite Heston Model]
\label{eg comp heston}
    \begin{equation}
    \label{Composite Heston}
             \begin{array}{lr}
            \mathrm{d} L_t = -\frac{1}{2} \mathrm{d}t + \mathrm{d}W(t), &   \\
             \mathrm{d} u_t = \kappa_u (\theta_u - u_t)\mathrm{d}t + \sigma_u \mathrm{d}Z(U_t), &\\
             \mathrm{d} v_t = \kappa_v (\theta_v - v_t)\mathrm{d}t + \sigma_v \mathrm{d}\tilde B(V_t),
             \end{array}
    \end{equation}
with $\langle W, Z\rangle_t = \rho_u t$, $\langle W, B\rangle_t = 0$, and $\langle Z, B\rangle_t = 0$.
\end{example}

\begin{example}[Composite Jump Heston]
\label{eg comp jh}
A Jump Heston (JH) model follows:
 \begin{equation}
             \label{JH}
             \begin{array}{lr}
             \mathrm{d}L_t = -\Psi(-\mathrm{i}) \mathrm{d}t + \mathrm{d}J(t), &   \\
             \mathrm{d} u_t = \kappa_u (\theta_u -u_t)\mathrm{d}t + \eta_u \mathrm{d}J^u(U_t),
             \end{array}
            \end{equation}
where $J^u$ is a subordinator such that 
\begin{equation}
\label{dJ^u}\mathrm{d} J^u_t = \begin{cases}
    0, & \quad \mathrm{d} J_t \ge 0,\\
    -\mathrm{d} J_t, & \quad \mathrm{d} J_t < 0,
\end{cases}
\end{equation}
and $J$ is a CGMY process whose characteristic exponent is given by
$$\Psi_J(m) = C\Gamma(- Y) \left(G + \mathrm{i}m)^Y - G^Y + (M - \mathrm{i}m)^Y - M^Y\right),$$
for $C > 0$, $G \ge 0$, $M \ge 0$, $Y \in (0, 2)$. The model exhibits leverage effect purely by co-jumps in the base process and activity rate process. \cite{ballotta2022smiles} demonstrate a superior performance of an ordinary time change JH model over other classic models, e.g., Heston, Barndorff-Nielsen-Shephard model (cf. \cite{barndorff2003}). 

To extend the Jump Heston model to CTC models, we consider the following Composite JH model
      \begin{equation}
      \label{Composite JH}
             \begin{array}{lr}
             \mathrm{d}L_t = -\Psi(-\mathrm{i})\mathrm{d}t + \mathrm{d}J(t), &   \\
             \mathrm{d} u_t = \kappa_u(\theta_u - u_t)\mathrm{d}t + \eta_u \mathrm{d}J^u(U_t), &\\
             \mathrm{d} v_t = \kappa_v(\theta_v - v_t)\mathrm{d}t + \sigma_v \mathrm{d}\tilde B(V_t),
             \end{array}
      \end{equation}
where $J^u$ is as defined above, and $\tilde B$ is a Brownian motion. In this specification, leverage is introduced purely by simultaneous jumps of return and volatility. 

In both JH and Composite JH models, we obtain from Eq. \eqref{characteristic exponent} and Eq. \eqref{dJ^u} that
$$\begin{aligned}
    \Psi_{J^u}^{\mathrm{Q}}(x) & = \Psi_{J, J^u}(m,x) - \Psi_J(m) \\
    & =C\Gamma(-Y) \left((G + \mathrm{i}(x + m))^Y - (G + \mathrm{i}m)^Y\right).
\end{aligned}$$
\end{example}

\section{Derivative Pricing Under the CTC Model}
This section develops pricing methods for derivatives under CTC models. In Section \ref{CTC-COS method}, we present the CTC-COS method for efficient European-style option pricing via characteristic function-based methods. And in Section \ref{VIX derivatives}, we analyze VIX derivatives pricing through explicit VIX-spot variance relationships for affine models and Laplace transforms for general cases. For implementation, we provide exact simulation schemes that exploit the models' structure to avoid path-dependent discretizations.
\subsection{CTC-COS Method}
\label{CTC-COS method}
Since we have obtained the expression for the characteristic function of log price process $X$, the pricing of European options follows from, for example, FFT (\cite{carr1999option}) or the cosine method (\cite{fang2009novel}), that efficiently compute option prices. However, evaluating the characteristic function of a CTC model usually involves repeated integration as shown in Corollary \ref{affine spec}. In the following, we propose a CTC-COS method that decouples such double integration in option pricing and achieves lower computational complexity.
\begin{theorem}[CTC-COS]
For a CTC model with $\rho_v = 0$, given expiry date $\bar{T}$, the price of a European-style option with payoff $G(Y, \bar{T})$, $Y:= \ln S_{\bar T}$, is numerically approximated by
\begin{equation} 
\label{call price}
\begin{aligned}
V\left(S_0, \bar{T}\right) & \approx \frac{2e^{-r \bar T}}{c}\left.\int_0^c \left(\sum_{k=0}^{N-1}{}^\prime \operatorname{Re}\left\{\phi_U^{\mathbb Q}\left(-\mathrm{i}\Psi_L(\frac{k\pi}{b-a});v\right) A_k\right\} F_k\right)\right. \\
& \quad \quad \left.\left(\sum_{l=0}^{M-1}{}^\prime \operatorname{Re}\left\{\phi_V(\frac{l\pi}{c}; \bar{T})\right\} \cos \left(\frac{l\pi v}{c}\right)\right)\mathrm{d}v.\right.
\end{aligned}
\end{equation}
Here $\sum {}^\prime$ means the first term in the summation is weighted by one-half, $a, b, c \in \mathbb{R}$ are integration range parameters with $a < b$ and $c > 0$,
\begin{equation}
\label{Ak}
A_k = \exp{\left\{-\mathrm{i} k \pi \frac{a}{b-a}+\frac{\mathrm{i}k\pi(\ln S_0 + r\bar T)}{b-a}\right\}},
\end{equation}
and
\begin{equation}
\label{Vk}
F_k = \frac{2}{b-a}\int_a^b G(y,\bar{T})\cos \left(k \pi \frac{y-a}{b-a}\right)\mathrm{d}y.
\end{equation}
\end{theorem}

\begin{proof}{Proof}

To apply the cosine expansion to the pricing of European-style options, we denote by $f_Y$ the density function of $Y = \ln S_{\bar T}$ and by $\phi_Y$ its characteristic function. Since the density function decays to zero at $\pm \infty$, we choose an interval $[a, b] \subset \mathbb{R}$ large enough such that $\hat f_Y := f_Y 1_{[a, b]}$ approximates $f_Y$ well.

The cosine expansion of $\hat f_Y$ reads
$$\hat f_Y(x) = \sum_{k=0}^{\infty}{}^\prime U_k \cos{\left(\frac{k\pi (x-a)}{b-a}\right)},$$
where  
\begin{equation}
\begin{aligned}
U_k &= \frac{2}{b-a} \int_a^b \hat f_Y(x) \cos{\left(\frac{k\pi (x-a)}{b-a}\right)} \mathrm{d}x\\
&= \frac{2}{b-a} \operatorname{Re}\left\{\hat \phi_Y\left(\frac{k\pi}{b-a}\right)\exp{\left(\frac{-\mathrm{i}k\pi a}{b-a}\right)}\right\}.
\end{aligned}
\end{equation}
with $\hat \phi_Y(w):= \int_{a}^b e^{\mathrm{i}wx} \hat f_Y(x)\mathrm{d}x$. By our assumption, $\phi_Y(w):= \int_{\mathbb{R}}e^{\mathrm{i}wx}f_Y(x)\mathrm{d}x$ is well approximated by $\hat \phi_Y(w)$. Thus, we approximate $U_k$ by 
$$\hat U_k := \frac{2}{b-a} \operatorname{Re}\left\{\phi_Y\left(\frac{k\pi}{b-a}\right)\exp{\left(\frac{-\mathrm{i}k\pi a}{b-a}\right)}\right\}.$$
Then we obtain a truncated approximation of $\hat f_Y$ as $$\begin{aligned}\bar f_Y(x) &= \sum_{k=0}^{N-1}{}^\prime \hat U_k \cos{\left(\frac{k\pi (x-a)}{b-a}\right)}\\
&= \frac{2}{b-a}\sum_{k=0}^{N-1}{}^\prime \operatorname{Re}\left\{\phi_Y\left(\frac{k\pi}{b-a}\right)\exp{\left(\frac{-\mathrm{i}k\pi a}{b-a}\right)}\right\}\cos{\left(\frac{k\pi (x-a)}{b-a}\right)}.\end{aligned}$$

Note that $X_{\bar{T}} = Y - r\bar T - \ln S_0$, which implies $\phi_Y(x)= \phi_X(x; \bar T)\exp{\left(\mathrm{i}x(\ln S_0 + r\bar T)\right)}$. We thus compute the option price as
$$
\begin{aligned}
V\left(S_0, \bar{T}\right) & = e^{-r\bar T} E[G(Y,\bar T)] \\
& = e^{-r\bar T} \int_{-\infty}^\infty G(y, \bar{T}) f_Y(y) \mathrm{d}y\\
& \approx e^{-r\bar T} \int_{a}^b G(y, \bar{T}) \bar f_Y(y) \mathrm{d}y\\
& = e^{-r\bar T} \int_a^b \frac{2}{b-a}G(y, \bar{T})\left(\sum_{k=0}^{N-1}{}^\prime \operatorname{Re}\left\{\phi_Y(\frac{k\pi}{b-a})\exp{\left(\frac{-\mathrm{i}k\pi a}{b-a}\right)}\right\}\cos{\left(\frac{k\pi (y-a)}{b-a}\right)}\right) \mathrm{d}y \\
& = \frac{2e^{-r\bar T}}{b-a}\sum_{k=0}^{N-1}{}^\prime \left(\operatorname{Re}\left\{\phi_X\left(\frac{k\pi}{b-a}; \bar T\right)\exp{\left(-\mathrm{i} k \pi \frac{a}{b-a}+\frac{\mathrm{i}k\pi(\ln S_0 + r\bar T)}{b-a}\right)}\right\} \right.\\
& \quad \quad \cdot\left. \int_a^b G(y, \bar{T})\cos{\left(\frac{k\pi (y-a)}{b-a}\right)}\right) \mathrm{d}y.
\end{aligned}
$$
We denote $A_k$ and $F_k$ as in Eq. \eqref{Ak} and Eq. \eqref{Vk}, respectively. Then the formula above is simplified as
$$V\left(S_0, \bar{T}\right) \approx e^{-r \bar T} \sum_{k=0}^{N-1}{}^\prime\operatorname{Re}\left\{\phi_X\left(\frac{k \pi}{b-a}; \bar{T}\right) A_k\right\}F_{k}.$$

According to Theorem \ref{Th1}, under the condition $\rho_v = 0$, 
$$
\phi_X\left(\frac{k \pi}{b-a}; \bar{T}\right) = E^\mathbb{Q}\left[\phi_U^{\mathbb Q}\left(-\mathrm{i}\Psi_L(\frac{k\pi}{b-a}); V_{\bar{T}})\right)\right].
$$
and the distribution of $V_{\bar{T}}$ remains invariant under $\mathbb{Q}$. Then the density of $V_{\bar T}$, denoted by $f_{V_{\bar T}}$, is supported on $[0, \infty)$. As a result, 
$$\phi_X\left(\frac{k \pi}{b-a}; \bar{T}\right) = \int_0^\infty \phi_U^{\mathbb Q}\left(-\mathrm{i}\Psi_L(\frac{k\pi}{b-a}); v\right)f_{V_{\bar T}}(v)\mathrm{d}v.$$
Similar to the treatment of $f_Y$, we apply the cosine expansion to $f_{V_{\bar T}}$ and truncate the integration above by the upper limit $c$:
$$f_{V_{\bar T}}(x) \approx  \frac{2}{c}\sum_{k=0}^{M-1}{}^\prime \operatorname{Re}\left\{\phi_V\left(\frac{k\pi}{c}; \bar T\right)\right\} \cos{\left(\frac{k\pi x}{c}\right)}.$$
Thus, the option price is further computed as
$$
\begin{aligned}
&\quad V\left(S_0, \bar{T}\right) \\
& \approx e^{-r \bar T} \sum_{k=0}^{N-1}{}^\prime \operatorname{Re}\left\{E^\mathbb{Q}\left[\phi_U^{\mathbb Q}(-\mathrm{i}\Psi_L(\frac{k\pi}{b-a}); V_{\bar{T}})\right]A_k\right\} F_{k}\\
& \approx e^{-r \bar T} \sum_{k=0}^{N-1}{}^\prime \operatorname{Re}\left\{ \frac{2A_k}{c} \int_0^c \phi_U^{\mathbb Q}\left(-\mathrm{i}\Psi_L(\frac{k\pi}{b-a});v\right)\left( \sum_{l=0}^{M-1}{}^\prime\operatorname{Re} \left\{\phi_V\left(\frac{l\pi}{c}; \bar T\right)\right\}\cos{\left(\frac{l\pi v}{c}\right)}\right)\mathrm{d}v \right\}F_k\\
& = \frac{2e^{-r \bar T}}{c}\left.\int_0^c \left(\sum_{k=0}^{N-1}{}^\prime \operatorname{Re}\left\{\phi_U^{\mathbb Q}(-\mathrm{i}\Psi_L(\frac{k\pi}{b-a});v) A_k\right\} F_k\right)\left(\sum_{l=0}^{M-1}{}^\prime \operatorname{Re}\left\{\phi_V(\frac{l\pi}{c}; \bar{T})\right\} \cos \left(\frac{l\pi v}{c}\right)\right)\mathrm{d}v\right.
\end{aligned}
$$
The last equation interchanges the order of integral and the summation, and is rearranged such that the summations with respect to time change $U$ and $V$ are separate.
\end{proof}

\begin{remark}
The payoff for plain vanilla options with strike price $K$ reads
$$G(y, \bar{T}) := \left[\alpha \cdot \left(e^y-K\right)\right]_{+} \quad \text{with} \quad \alpha=\left\{\begin{array}{cl}
1 & \text{for a call} \\
-1 & \text{for a put.}
\end{array}\right.$$
In this case, we note that $F_k$ admits an explicit form (cf. \cite{fang2009novel})
$$F_k = \left\{\begin{array}{cl}
\frac{2}{b-a}\left(\chi_k(0, b) - K\psi_k(0, b)\right) & \text{for a call,} \\
\frac{2}{b-a}\left(-\chi_k(a, 0) + K\psi_k(a, 0)\right) & \text{for a put,} 
\end{array}\right.$$
where
$$\begin{aligned}
\chi_k(c, d):= & \frac{1}{1+\left(\frac{k \pi}{b-a}\right)^2}\left[\cos \left(k \pi \frac{d-a}{b-a}\right) e^d-\cos \left(k \pi \frac{c-a}{b-a}\right) e^c\right. \\
& \left.\quad+\frac{k \pi}{b-a} \sin \left(k \pi \frac{d-a}{b-a}\right) e^d-\frac{k \pi}{b-a} \sin \left(k \pi \frac{c-a}{b-a}\right) e^c\right]
\end{aligned}$$
and
$$\psi_k(c, d):= \begin{cases}{\left[\sin \left(k \pi \frac{d-a}{b-a}\right)-\sin \left(k \pi \frac{c-a}{b-a}\right)\right] \frac{b-a}{k \pi}} & k \neq 0, \\ d-c & k=0 .\end{cases}$$
\end{remark}
\begin{remark}
The approximation error contains three parts: the integration range truncation error $\varepsilon_1$ induced by truncating the integration range from $(-\infty,\infty)$ to $[a, b]$ and from $[0, \infty)$ to $[0, c]$, the series truncation error $\varepsilon_2$ induced by truncating the cosine expansion series, and the error $\varepsilon_3$ induced by replacing $U_k$ by $\hat{U}_k$ in each term of the cosine expansion series. Error $\varepsilon_1, \epsilon_3$ are controlled by $a, b, c$ and converge to zero as $a\to -\infty$ and $b, c\to \infty$. The series truncation error $\varepsilon_2$ is controlled by $M$ and $N$. If the first derivatives of the densities $f_Y$ and $f_V$ are integrable, then $\epsilon_2$ converges to zero as $M, N\to \infty$ and $a\to -\infty$, $b, c\to +\infty$. Readers may refer to \cite{fang2009novel} for further error analysis. In Appendix \ref{accuracy of COS}, we further conduct numerical analysis of the accuracy of CTC-COS method, exemplified by Composite Heston model.
\end{remark}

In summary, the approximation above is deduced by performing double cosine series expansions and then reordering the summation and integration. A direct numerical computation can be applied if the Fourier transform of time changes $U$ and $V$ are available. 

Since the summations within the integral are separate, the overall complexity is $O(ND)$ under affine specifications, where $D$ is the number of time steps in numerical integration. Thus, the complexity of call price computation for affine CTC models is the same as their ordinary time-change counterparts, including Composite Heston and Composite JH in Example \ref{JH}.

In practice, when pricing the whole volatility surface, the number of strike prices does not add computational complexity since Fourier-based methods like COS allow for a separation of strikes and the underlying asset. Specifically, the strike price is only incorporated in $F_k$ in the numerical approximation \eqref{call price} and can be computed in advance.

Meanwhile, for affine models, we can simultaneously compute $\phi_U^{\mathbb{Q}}(\cdot; v)$ and $\phi_V^{\mathbb{Q}}(\cdot; \bar{T})$ across all maturities $\bar T$. This is because the characteristic functions are computed by solving ODE systems \eqref{ODE for V} and \eqref{ODE for T} through temporal discretization. And by discretizing for the largest maturities, the corresponding coefficients for smaller maturities are solved along the discretization process.

\begin{algorithm}
\label{alg0}
\renewcommand{\algorithmicrequire}{\textbf{Input:}}
\renewcommand{\algorithmicensure}{\textbf{Output:}}
\caption{Simulation Method for Call Option Price Under Composite Heston Model (Eq. \eqref{Composite Heston})} 
\label{simu_alg} 
\begin{algorithmic}
\Require Maturity $\bar{T}$, strike $K$, number of simultion sample path $N$, temporal discretization step size $\delta$, model parameters $\Theta$ (with $\rho_u =\rho \in [-1, 0]$ and $\rho_v = 0$).
\\
\vspace{-1.5em}
{\noindent}\rule[-10pt]{14.3cm}{0.05em}
\State Define $M_v = \bar{T} / \delta$. Obtain a $N \times M_v$ array $(B_{n,m})$, where the elements $B_{n,m}$ are i.i.d. samples of standard normal distribution.
\For{$n = 1, 2, \dots, N$}
\State Obtain $\{v_{0}, v_{\delta}, \cdots, v_{M_v\delta}\}$ from a proper discretization scheme of $v$ in Eq. \eqref{Composite Heston} using $\{B_{n,m}, m = 1,\cdots, M_v\}$.
\EndFor
\State Sum up $v$ to obtain $\{V_0, V_{\delta}, \cdots, V_{M_v \delta}\}$.
\State Obtain two $N \times M_v$ arrays, where each term $Z_{n, m}$ (of the first array) is sampled independently from std normal $\mathcal{N}_{m,n}^1$ and each term $W_{n, m}$ (of the second array) is sampled independently from $\rho \mathcal{N}_{m,n}^1 + \sqrt{1-\rho^2} \mathcal{N}_{m,n}^2$. $\mathcal{N}_{m,n}^1$ and $\mathcal{N}_{m,n}^2$ are independent.

\For{$n = 1, 2, \dots, N$}
\State Obtain $\{u_{0}, u_{V_\delta}, \cdots, u_{V_{M_v\delta}}\}$ from a proper discretization scheme of $u$ in Eq. \eqref{Composite Heston} using $\{Z_{n,m}, m=1,\cdots, M_u\}$. That is, discretizing
$$\begin{aligned}
    \mathrm{d}u_{V_t} &= \kappa_u(\theta_u- u_{V_t})\mathrm{d}V_t + \sigma_u\mathrm{d}Z(V_t)\\
    &= \kappa_u (\theta_u - u_{v_t}) v_t \mathrm{d}t + \sigma_u \sqrt{v_t} \mathrm{d}\bar Z_t.
\end{aligned}$$
\EndFor

\For{$n = 1, 2, \dots, N$}
\State Obtain $\{L_{0}, L_{U_{V_\delta}}, \cdots, L_{U_{V_{M_v\delta}}}\}$ from a proper discretization scheme of $L$ in Eq. \eqref{Composite Heston} using $\{W_{n,m}, m=1,\cdots, M_u\}$. That is,
discretizing 
$$\begin{aligned}\mathrm{d}L_{U_{V_t}} &= -\frac{1}{2}\mathrm{d}U_{V_t} + \mathrm{d}W(U_{V_t})
\\&= -\frac{1}{2}u_{V_t}v_t\mathrm{d}t + \sqrt{u_{V_t}v_t}\mathrm{d}\bar W_t.\end{aligned}$$
\EndFor

\State Sum up and obtain $L_{U_{V_{M_v\delta}}}$
\Ensure $V\left(S_0, K, \bar{T}\right) = \operatorname{Mean}(S_0\exp{(L_{U_{V_{M_v\delta}}}}) - K)^+.$
\end{algorithmic} 
\end{algorithm}
\subsection{VIX Derivatives}
\label{VIX derivatives}



To price VIX derivatives, it is common practice to first derive the relationship between the VIX and spot variance at maturity since the distribution of VIX is not directly available. It's known that $\text{VIX}^2$ is linearly dependent on the spot variance for the Heston model. However, for most other models, such relationship is implicit and requires simulation or a numerical procedure to derive. 

For example, as shown in \cite{jacquier2018VIX} for rough Bergomi models, the $\text{VIX}_T^2$ is expressed as an integral form and a computationally costly simulation is needed to obtain samples of $\text{VIX}_T$. Moreover, the VIX-Spot relationship is also implicit for a 3/2 model, where a numerical differentiation is needed that leads to additional computational cost.

In the following, we first establish explicit VIX-spot variance relationships for both ordinary and composite time change models (in Section \ref{ordinary time change} and \ref{CTC models:affine}, respectively). In Section \ref{exact simulation}, we derive an exact simulation scheme for certain CTC models. The approach avoids path simulation by directly sampling terminal distributions. And in Section \ref{CTC models:nonaffine}, we consider non-affine cases using numerical differentiation.

\subsubsection{Ordinary Time Change Models}
\label{ordinary time change}
\begin{proposition}
\label{VIX formula}
Consider a time change model $X_t = L_{T_t}$ in Eq. \eqref{CTC} with $V_t \equiv t$ and an affine process $u$ as considered in Corollary \ref{affine spec}, i.e., 
$$\mathrm{d}u_t = \kappa_u (\theta_u - u_t)\mathrm{d}t + \sigma_u \mathrm{d}Z(U_t) + \eta_u \mathrm{d}J^u(U_t).$$
Define $\mathcal{M}(x,t) := \frac{e^{xt} - 1}{x t}$ and $\bar{\tau} = 30 / 365$. 
Then, it holds that
\begin{equation*}\textbf{}
\text{VIX}_t^2 = au_t + b,
\end{equation*} where
\begin{align*}
& a = 2 (\Psi(-\mathrm{i}) - E[J_1])\mathcal{M}(m_u,\bar\tau),\\
& b = 2 (\Psi(-\mathrm{i}) - E [J_1])\frac{\kappa_u\theta_u}{m_u} (\mathcal{M}(m_u, \bar\tau) - 1),\\
& m_u = \eta_u E[J^u_1] - \kappa_u.
\end{align*}
\end{proposition}

\begin{proof}{Proof}
Note that the local martingale property is preserved under a type-2 time change (cf. Theorem (10.17) in \cite{jacod1979calcul}). Then according to Proposition 4.22 (Statement 3) of \cite{eberlein2019mathematical} and the local characteristics of a time-changed L\'{e}vy process given by Eq. \eqref{TC characteristics}, for L\'{e}vy process $L \in \mathcal{H}^2_{\text{loc}}(\mathbb{R})$ and a type-3 time change $T$, $L(T_t) - E[L_1] T_t$ is a locally square-integrable martingale if and only if $E[T_t] < \infty$. Consequently, $B(U_t)$ and $J(U_t) - E[J_1]U_t$ are $\mathcal{F}^X$-martingales, and we have
$$
\begin{aligned}
\text{VIX}_t^2 & = - \frac{2}{\bar{\tau}} E_t\left[\ln \frac{e^{-r \bar \tau}S_{t+\bar{\tau}}}{S_t}\right]\\
& = - \frac{2}{\bar{\tau}} E_t\left[X_{t+\bar{\tau}} - X_t \right] \\
& = \frac{2}{\bar \tau}\left( \Psi(-i)E_t[ U_{t+\bar\tau} - U_t] - E_t[\sigma(B(U_{t+\bar\tau}) - B(U_t))+ \eta(J(U_{t+\bar\tau})- J(U_t))]\right)\\
& = \frac{2}{\bar \tau}\left(\Psi(-\mathrm{i}) E_t\left[ U_{t+\bar{\tau}} - U_{t} \right] - E[J_1]E_t\left[ \left(U_{t+\bar{\tau}} - U_{t})\right)\right]\right) \quad \text{(martingale property)}\\
& = 2 (\Psi(-\mathrm{i}) - E[J_1])g(u_t, \bar{\tau}),
\end{aligned}
$$
where $g(u_t, h) := \frac{1}{h}E_t\left[ U_{t+h} - U_{t} \right]$. Due to, for $s > t$,
$$u_s = u_t + \int_t^s \kappa_u (\theta_u - u_r)\mathrm{d}r + \sigma_u \left(B(U_s) - B(U_t)\right) + \eta_u \left(J^u(U_s) - J^u(U_t)\right),
$$
and the martingale property of $B(U_t)$ and $J^u(U_t) - E[J_1^u] U_t$, we have
$$
\begin{aligned}
    E_t[u_s] & = u_t + \kappa_u \theta_u (s-t) -\kappa_u\int_t^s E_t[u_r]\mathrm{d}r + \eta_u E[J^u_1]E_t[U_s - U_t]\\
    & = u_t + \kappa_u \theta_u (s-t) + m_u \int_t^s E_t[u_r]\mathrm{d}r,
\end{aligned}
$$
with $m_u = \eta_u E[J^u_1] - \kappa_u$. Thus, we solve $E_t[u_s]$ by
$$
E_t[u_s] = \left(u_t + \frac{\kappa_u\theta_u}{m_u}\right) e^{m_u (s-t)} - \frac{\kappa_u \theta_u}{m_u}
$$
and obtain
\begin{equation}
\label{g}g(u_t,h) = \frac{1}{h}\int_t^{t+ h} E_t[u_s]\mathrm{d}s= u_t \mathcal{M}(m_u, h) + \frac{\kappa_u \theta_u}{m_u} (\mathcal{M}(m_u, h) - 1).
\end{equation}
Then the results follow by substituting $g(u_t,  \bar\tau)$ back to $\operatorname{VIX}_t^2.$
\end{proof}

Since the characteristic function of $u$ is easily obtained by solving an ODE system as given in \cite{filipovic2001general}, the pricing of European-style derivatives is obtained accordingly via Fourier-based methods. For a VIX call option with strike price $K$, the option price is given by (e.g., according to Proposition 2 of \cite{lian2013pricing}):
$$
V^{\operatorname{VIX}}\left(S_0, \bar{T}\right) = \frac{e^{-r\bar T}}{2 a \sqrt{\pi}} \int_0^{\infty} \operatorname{Re}\left[e^{z b / a} \phi_u(-\mathrm{i}z; s) \frac{1- \operatorname{erf}(K \sqrt{z / a})}{(\sqrt{z / a})^3}\right] \mathrm{d}z_I
$$
where $z=z_R+\mathrm{i} z_I $ is a complex variable with any $z_R \le 0$, $\phi_u(z; s)$ is the (generalized) characteristic function of $u_s$, constants $a, b$ are given in Proposition \ref{VIX formula}, and 
$$\operatorname{erf}(x) = \frac{2}{\sqrt{\pi}}\int_0^x e^{-s^2}\mathrm{d}s.$$

\subsubsection{CTC Models: Affine Cases}
\label{CTC models:affine}
For the case of affine CTC models, the VIX-spot relationship is still explicit and easy to obtain.
\begin{proposition}
Consider a CTC model with affine activity rates as considered in Corollary \ref{affine spec}. 
Define $\phi_V(x; t, t+ \bar\tau) := E_t[\exp{\left(\mathrm{i}x(V_{t+\bar\tau} - V_t)\right)}]$,  $\mathcal{M}(x, t) := \frac{e^{xt} - 1}{xt}$, and $\bar{\tau} = 30 / 365$.
Then, it holds that
\begin{equation}
\label{VIX for CTCM}
\operatorname{VIX}_t^2 = A(v_t) u_{V_t} + B v_t + C(v_t),
\end{equation}
where
\begin{align*}
& A(v_t) = \bar M\frac{\phi_V(-\mathrm{i}m_u;t,t+\bar{\tau})-1}{m_u\bar{\tau}},\\
& B = -\bar M \frac{\kappa_u\theta_u\mathcal{M}(m_v,\bar\tau)}{m_u},\\ 
& C(v_t) = \bar M \left(- \frac{\kappa_u\kappa_v\theta_u\theta_v(\mathcal{M}(m_v,\bar\tau) - 1)}{m_u m_v}+ \frac{\kappa_u\theta_u}{m_u^2\bar{\tau}}\left(\phi_V(-\mathrm{i}m_u;t,t+\bar{\tau}) - 1\right)\right),\\
& m_u = \eta_u E[J^u_1] - \kappa_u, \,m_v = \eta_v E[\tilde J^v_1] - \kappa_v,\\
& \bar M = 2 (\Psi(-\mathrm{i}) - E[J_1]).
\end{align*}
\end{proposition}
\begin{proof}{Proof}
Since $E[U_{V_t}] < \infty$ for every $t\ge 0$ by our assumption of composite time change, $X_t - E[L_1] U_{V_t} = X_t + \frac{1}{2}\bar M U_{V_t}$ is a zero-mean martingale (following the same argument in the proof of Proposition \ref{affine spec}), where $\bar M = 2 (\Psi(-\mathrm{i}) - E[J_1]).$ Thus,
\begin{equation}
\label{VIX general}
    \operatorname{VIX}_t^2 = -\frac{2}{\bar{\tau}} E_t[X_{t+ \bar\tau} - X_t]= \frac{\bar M}{\bar\tau}E_t[U_{V_{t+\bar{\tau}}} - U_{V_t}].\end{equation}
By the same approach for deriving Eq. \eqref{g}, we obtain 
$$
    \frac{1}{h} \int_{t}^{t+h} E[u_s\mid \mathcal{F}_{U_t}] \mathrm{d}s = u_{t}\mathcal{M}(m_u, h) + \frac{\kappa_u \theta_u}{m_u} (\mathcal{M}(m_u, h) - 1).
$$
Thus, by the independence of $u$ and $v$,
\begin{equation}
    \label{g for u}
\begin{aligned}&\frac{1}{\Delta_{V_t}(h)} \int_{V_t}^{V_{t+h}}E[u_s \mid \mathcal{F}^X_t, \{V_r\}_{r\ge 0}]\mathrm{d}s\\
=& u_{V_t}\mathcal{M}(m_u, \Delta_{V_t}(h)) + \frac{\kappa_u \theta_u}{m_u} (\mathcal{M}(m_u, \Delta_{V_t}(h)) - 1),
\end{aligned}
\end{equation}
with $\Delta_{V_t}(h):= V_{t+h} - V_t$. 

Next, we apply an iterated conditioning argument to compute the expectation $E_t[U_{V_{t+\bar{\tau}}} - U_{V_t}]$:
\begin{align}
&\quad \frac{1}{\bar\tau}E_t \left[U_{V_{t+\bar{\tau}}} - U_{V_t}\right] \notag\\ &=\frac{1}{\bar\tau} E_t \left[E\left[U_{V_{t+\bar{\tau}}} - U_{V_t}\mid \mathcal{F}_t^X, \{V_{r}\}_{r\ge 0}\right]\right] \notag\\
&=\frac{1}{\bar\tau} E_t \left[\int_{V_t}^{V_{t+\bar\tau} }E\left[u_{s}\mid \mathcal{F}^X_t, \{V_r\}_{r\ge 0}\right] \mathrm{d}s\right] \quad (\text{independence of } u \text{ and } v)\notag\\
&= \frac{1}{\bar\tau}E_t\left[ u_{V_t} \Delta_{V_t}(\bar{\tau})\mathcal{M}(m_u, \Delta_{V_t}(\bar{\tau})) + \frac{\kappa_u \theta_u}{m_u}\Delta_{V_t}(\bar{\tau})\left(\mathcal{M}(m_u, \Delta_{V_t}(\bar{\tau})) - 1\right)\right] \notag \quad (\text{by Eq. \eqref{g for u}})\\
&=\frac{u_{V_t}}{m_u\bar\tau}E_t[\exp{(m_u \Delta_{V_t}(\bar\tau))} - 1] + \frac{\kappa_u \theta_u}{m_u^2 \bar\tau}E_t[\exp{(m_u\Delta_{V_t}(\bar\tau))} - 1] - \frac{\kappa_u \theta_u}{m_u \bar\tau} E_t[\Delta_{V_t}(\bar\tau)] \notag\\
&= \frac{m_u u_{V_t} + \kappa_u \theta_u}{m_u^2 \bar\tau}E_t[\exp{(m_u\Delta_{V_t}(\bar{\tau})} -1] - \frac{\kappa_u\theta_u}{m_u}\left(v_t\mathcal{M}(m_v, \bar\tau) + \frac{\kappa_v\theta_v}{m_v} (\mathcal{M}(m_v, \bar\tau) - 1)\right) \notag\\
&= \frac{\phi_V(-\mathrm{i}m_u;t,t+\bar{\tau})-1}{m_u\bar{\tau}}u_{V_t} - \frac{\kappa_u\theta_u\mathcal{M}(m_v, \bar\tau)}{m_u} v_t + \frac{C(v_t)}{\bar M}, \notag
\end{align}
where $C(v_t)$ is as defined in the proposition. Substituting it back to the expression of $\operatorname{VIX}_t^2$ and the result follows.
\end{proof}


\subsubsection{Exact Simulation of Spot Variances}
\label{exact simulation}
Despite the explicit VIX-Spot relationship, a simulation procedure for $u_{V_t}$ and $v_t$ is still needed in general to price VIX derivatives. While Monte Carlo methods are typically slow, we will show in the following that an efficient and exact simulation scheme exists for certain models, where $v$ is a CIR process and the characteristic function of $u_t$ is explicit. This approach only requires simulating distributions at maturity rather than entire trajectories, which eliminates discretization error and improves efficiency.

Given a maturity $\bar{T}$, we first simulate $v_{\bar{T}}$ and $u_{V_{\bar{T}}}$, and then compute VIX samples using formula \eqref{VIX for CTCM}. The exact simulation proceeds as follows:
\begin{description}
    \item[Step 1] Simulate $v_{\bar T}$ from a non-central chi-square distribution.
    \item[Step 2] Simulate the conditional distribution $(V_{\bar T} \mid v_0, v_{\bar T})$ according to the method of \cite{glasserman2011gamma}.
    
    Specifically, by the independence of $u$ and $v$, we first simulate $V_{\bar{T}}$ and then simulate $u_{V_{\bar{T}}}$ with a noncentral chi-square distribution. Therefore, the conditional distribution of $V_{\bar{T}}$, in the form of $\left(\int_0^{\bar{T}} v_s \mathrm{d}s \mid v_0, v_{\bar{T}}\right)$, needs to be derived. This is exactly the same problem faced in the Heston simulation, as brought up by \cite{broadie2006exact}. We then apply the method of gamma expansion in \cite{glasserman2011gamma}, where the conditional distribution is efficiently simulated as a sum of independent variables.
    \item[Step 3] Simulate $u_{V_{\bar{T}}}$ using density function obtained by inverting the characteristic function of $u$ at $V_{\bar{T}}$. 

    The characteristic function of $u_t$ is explicit for some models, e.g., Heston  and 3/2 models given in \cite{carr2007new}. For Composite JH model \eqref{Composite JH}, the method still applies. We can ease the computation by precomputing the characteristic function of $u_t$ for a range of $t \in [0, M]$ along the numerical discretization of the corresponding ODE, where $M$ is an estimated upper bound of the support of $V_{\bar T}$. We then obtain sample points of $u_{V_{\bar{T}}}$ by matching the precomputed values with the simulated $V_{\bar{T}}$ samples.
    \item[Step 4] Obtain $\text{VIX}_{\bar{T}}$ by formula \eqref{VIX for CTCM} and compute $V^{\operatorname{VIX}}(t, S_t, \bar{T})$ by taking the average of the discounted payoff.
\end{description}


To sum up, we price European-style options and VIX derivatives with payoff function $G(\cdot)$ by algorithm \ref{alg1} and \ref{alg2}, respectively.
\begin{algorithm}
	\renewcommand{\algorithmicrequire}{\textbf{Input:}}
	\renewcommand{\algorithmicensure}{\textbf{Output:}}
	\caption{Calculate European-style Prices} 
	\label{alg1} 
	\begin{algorithmic}
		\Require Maturity $\bar{T}$, payoff $G(\cdot)$, discretization parameters $N, M, Q$, integration range $a,b,c$.
		\For{$v = \frac{c}{Q}, \frac{2c}{Q}, \dots, , c$} 
  \For{$k = 0, 1, \cdots, N-1$}
  \State Compute $A_k, V_k$ according to Eq. \eqref{Ak} and \eqref{Vk}
		\State Compute $\phi_U^{\mathbb Q}(-\mathrm{i}\Psi_L(m); v)$ according to Eq.$\eqref{ODE for T}$
		\EndFor
  \For{$l = 0, 1, \cdots, M-1$}
  \State Compute $\phi_V^{\mathbb{Q}}(\frac{l\pi}{c}; \bar{T})$ and $\cos(\frac{l\pi v}{c})$
  \EndFor
  \State Compute option price $V\left(S_0, \bar{T}\right)$ by summing up according to Eq. \eqref{call price}
  \EndFor
  \Ensure $V\left(S_0, \bar{T}\right)$
	\end{algorithmic} 
\end{algorithm}

\begin{algorithm}
	\renewcommand{\algorithmicrequire}{\textbf{Input:}}
	\renewcommand{\algorithmicensure}{\textbf{Output:}}
	\caption{Calculate VIX Derivative Prices} 
	\label{alg2} 
	\begin{algorithmic}
		\Require Maturity $\bar{T}$, strike $K$, payoff $G(\cdot)$.
	\State Simulate $v_{\bar T}$ from a non-central chi-square distribution
 \State Simulate the conditional distribution $(V_{\bar{T}} \mid v_0, v_{\bar{T}})$ by the method of \cite{glasserman2011gamma}
 \State Given $V_{\bar{T}}$, simulate $u_{V_{\bar{T}}}$ by inverting the characteristic function of $u$
 \State Obtain the sample of $\text{VIX}_{\bar{T}}$ according to Eq. \eqref{VIX for CTCM}
 \State Compute the derivative price $V^{\operatorname{VIX}}\left(S_0, \bar{T}\right)$ by taking the mean of the payoff function $G(\cdot)$
  \Ensure $V^{\operatorname{VIX}}\left(S_0, \bar{T}\right)$
	\end{algorithmic} 
\end{algorithm}

As shown in the results in this section, a CTC model has the following features in derivative pricing. First, the composite structure does not pose additional computational complexity to European option pricing. A CTC-COS method can be developed to price efficiently as in ordinary time-changed L\'{e}vy models. However, the situation is different in VIX derivative pricing. Under an affine specification, Fourier-based numerical methods are available for ordinary time-changed L\'{e}vy models, but are not for CTC models, due to a nonlinear relationship between spot $\text{VIX}$ and the activity rates. But fortunately, an exact simulation is available for certain CTC models, where we can directly simulate the distributions at maturity with the aid of the explicit characteristic functions of the time changes. This semi-numerical simulation method enables the VIX derivatives to be priced efficiently as well.

\subsubsection{CTC Models: Non-affine Cases}
\label{CTC models:nonaffine}
For more general specifications of time changes in Eq. \eqref{CTC}, there generally does not exist an explicit expression for $\text{VIX}_t$. Still, we could recover it numerically from the Laplace transform of $U_{V_t}$. Since $u$, by its definition, is a time-homogeneous Markov process with respect to $\{\mathcal{F}_{U_t}\}_{t\ge 0}$, for any bounded and measurable function $f$ and $ s > 0$,
$$\begin{aligned}
    E_t[f(u_{V_{t+s}})] &= E_t[E[f(u_{V_{t+s}})] \mid \mathcal{F}^X_t, \{V_r\}_{r\ge 0}]\\
    &= E_t[P_{\Delta_{V_t}(s)}f(u_{V_t})]\\
    &=: E[\tilde f(\Delta_{V_{t}}(s); u_{V_t}) \mid \mathcal{F}^X_t],
\end{aligned}$$
where $P_s f(x) := E[f(u_{V_t + s}) \mid u_{V_t} =x, \{V_r\}_{r\ge 0}]$ is bounded and measurable with respect to variable $s$ since $u_{V_t}$ is right-continuous. Thus, the function $\tilde f$ defined from $P_{\Delta_{V_t}(s)}f$ is bounded and measurable. As a result of the Markovian property of $v$, $$E[\tilde f(\Delta_{V_{t}}(s); u_{V_t})\mid \mathcal{F}_t^X] = E[\tilde f(\Delta_{V_{t}}(s); u_{V_t}) \mid u_{V_t}, v_t].$$
Thus, $(u_{V_t}, v_t)$ is Markovian, and we have from Eq. \eqref{VIX general}:
\begin{equation*}
\begin{aligned}
\text{VIX}_t^2 &= \frac{\bar M}{\bar\tau}E[U_{V_{t+\bar\tau}} - U_{V_t} \mid u_{V_t}, v_t]\\
&=\frac{2 (\Psi(-\mathrm{i}) - E[J_1]}{\bar{\tau}}g(v_t, u_{V_t}, \bar{\tau}),
\end{aligned}
\end{equation*}
where $$g(v_t, u_{V_t}, \bar{\tau}) := -\left.\frac{\partial}{\partial l}E\left[\exp{(-l(U_{V_{t+\bar{\tau}}} - U_{V_t}))}\mid u_{V_t}, v_t\right]\right|_{l=0}.$$
In particular, for affine specifications, the Laplace transform can be computed according to Eq. \eqref{ind chf}, with $\mathbb{Q}$ replaced by $\mathbb{P}$ and $\Psi_L(m)$ replaced by $-l.$ 

Once the expression for $\operatorname{VIX}_t$ is obtained, the price of VIX derivatives can be computed by Monte Carlo simulation. That is, for a given maturity $\bar T$, simulate $v_{\bar T}$ and $u_{V_{\bar T}}$ and obtain samples of $\operatorname{VIX}_{\bar T}$ from function $g$. The price is obtained by averaging the payoff function given the samples of $\operatorname{VIX}_{\bar T}.$

\section{Joint Calibration}
This section examines the joint calibration of our models to SPX and VIX options markets. In Section \ref{data} and \ref{calibration procedure}, we describe the dataset characteristics and outline our calibration procedure. In Section \ref{results}, we present in-sample and out-of-sample performance metrics and general calibration results. Finally in Section \ref{evaluations}, we analyze how well the calibrated models capture key volatility characteristics.
\subsection{Data}
\label{data}
We use the S\&P 500 equity-index and VIX options traded on CBOE to test the performance of the proposed models. In order to compare with existing works, we choose the dataset used in \cite{yuan2020new}. The dataset we consider spans ten years, containing the implied volatility surfaces of the SPX and VIX options from April 2, 2007 to December 29, 2017. We use the closing data including implied volatilties, trading volumes, and option contract information. Only Wednesday data are considered to avoid weekday effects, resulting in 557 weeks in total. 

To obtain option moneyness, defined as $K / F$ for strike price $K$ and futures price $F$, we compute the implied futures price using put-call parity from the pair of options with the strike price closest to the index. Then we use the inferred futures price to compute moneyness. The options are selected through the following procedures:
\begin{itemize}
\item Remove options quotes with zero trading volume.
\item Remove options quotes that violates standard arbitrage conditions (see \cite{kokholm2015joint} for example).
    \item Remove SPX (VIX) option quotes with maturity fewer than 7 (7) days or more than 365 (160) days.
    \item Remove SPX (VIX) option quotes whose moneyness does not fall in $[0.5, 1.4]$, $([0.7, 2.5])$.
    \item Remove all ITM options.
\end{itemize}

The filtering conditions are common practice, as introduced in \cite{bakshi1997empirical} and \cite{bardgett2019inferring}. Finally, we obtain a daily average of 352 SPX options and 73 VIX options. In comparison, \cite{yuan2020new} obtained a daily average of 426 SPX options and 58 VIX options.

\subsection{Calibration Procedure}
\label{calibration procedure}
The sample is divided into an in-sample period from April 4, 2007 to April 1, 2015, and an out-of-sample period from April 8, 2015 to December 27, 2017.

In the in-sample period, we calibrate the subsample on date $t$ by solving the optimization problem below.
\begin{equation}
\label{err}
\begin{aligned}
\Theta^*_t & =  \underset{\Theta}{\arg \min }\, \frac{1}{N_S} \sum_{i=1}^{N_S}\left(\frac{\sigma_S^i(\Theta)-\widehat{\sigma}_S^{i, t}}{\widehat{\sigma}_S^{i, t}}\right)^2+\frac{1}{N_V}\sum_{i=1}^{N_V}\left(\frac{\sigma_V^i(\Theta)-\widehat{\sigma}_V^{i, t}}{\widehat{\sigma}_V^{i, t}}\right)^2\\
& = :  \underset{\Theta}{\arg \min } \, F^t(\Theta),
\end{aligned}
\end{equation}
where $\Theta$ is the set of model parameters, $N_S$ and $N_V$ are the number of SPX and VIX options quotes in the sample, $\sigma_X^i(\Theta), X\in \{S, V\}$ is the $i$-th implied volatility computed from the calibrated model, and $\widehat{\sigma}_X^{i, t}, X\in \{S, V\}$ is the $i$-th market implied volatility on date $t$.

To perform an out-of-sample test of pricing models, we divide the model parameters into structural parameters (denoted by $\Theta_0$) and state parameters (denoted by $\mathcal{V}$). For example, in Composite JH model, $\Theta_0 = \{\kappa_u, \theta_u, \eta_u, \kappa_v, \theta_v, \sigma_v, C, G, M, Y\}$ and $\mathcal{V} = \{u_0, v_0\}$. Then we perform a two-step optimization. First, we choose a small sample ($T_0$ number of dates) before the out-of-sample period. In this period, we optimize by allowing the state parameters to vary daily and fixing a common set of structural parameters, which leads to an aggregate optimization problem
\begin{equation}
    \min_{\Theta_0, \mathcal{V}(1), \cdots, \mathcal{V}(T_0)}\sum_{t=1}^{T_0} F^t(\{\Theta_0, \mathcal{V}(t)\}).
\end{equation}
Denote by $\Theta_0^*$ the resulting optimized set structural parameters. Then in step 2, we solve daily out-of-sample optimization problems (with $T$ number of dates) by optimizing the state parameters
\begin{equation}
    \min_{\mathcal{V}(t)} F^t(\Theta_0^*, \mathcal{V}(t)),\quad t = T_0 + 1, \cdots, T_0 + T.
\end{equation}
\subsection{Results}
\label{results}
Below are the results of joint calibration under the four models: Heston, JH, Composite Heston and Composite JH. The performance error is measured by the root mean square relative error (RMSRE) of implied volatility, defined as
\begin{equation}
\label{error}
   \epsilon_t = \frac{1}{2} \sqrt{\frac{1}{N_S} \sum_{i=1}^{N_S}\left(\frac{\sigma_S^{i}(\Theta^*_t)-\widehat{\sigma}_S^{i, t}}{\widehat{\sigma}_S^{i, t}}\right)^2}+\frac{1}{2}\sqrt{\frac{1}{N_V}\sum_{i=1}^{N_V}\left(\frac{\sigma_V^i(\Theta^*_t)-\widehat{\sigma}_V^{i, t}}{\widehat{\sigma}_V^{i, t}}\right)^2}.
\end{equation}

Table \ref{Insample} summarizes the in-sample and out-of-sample calibration performance (via RMSRE) of ordinary models (Heston, JH) and their composite extensions, with the out-of-sample period divided into three sub-periods: 2015/3/31 - 2016/3/13 (P1), 2016/3/14 - 2017/2/5 (P2), and 2017/2/6 - 2017/12/31 (P3). Before the calibration of every out-of-sample period, we fix the structural parameters calibrated from the last period using two-month market data.

Consistently, composite models outperform ordinary ones. For in-sample period, the Composite Heston reduces RMSRE by 32.62\% (vs. Heston), and the Composite JH by 23.00\% (vs. JH). For out-of-sample period, the improvement persists: for the Composite JH, RMSRE is 0.0837 (P1), 0.0894 (P2), 0.1000 (P3)—all lower than the ordinary JH’s 0.1113, 0.1451, 0.1723. The Composite Heston also cuts Heston’s out-of-sample RMSRE, with the largest reduction (45.75\%) in P2. These results confirm that composite extensions enhance both in-sample fit and out-of-sample generalization. Another notable finding is that the RMSRE standard deviations of composite models are consistently smaller, which further confirms their superior generalization capability.

\begin{table}
\caption{The in-sample and out-of-sample calibration RMSRE and its standard deviation (in parathesis)\label{Insample}}
\begin{center}
\scalebox{0.8}{
\begin{tabular}{ccccccc}
\hline & Heston & \makecell[c]{Composite\\ Heston} & \makecell[c]{Heston \\Err. Reduced}  & JH & \makecell[c]{Composite\\ JH} & \makecell[c]{JH \\Err. Reduced}\\
\hline \makecell[c]{In-sample\\ \quad} & \makecell[c]{0.1263\\ (0.0480)}  & \makecell[c]{0.0851 \\(0.0353)} & \textit{-32.62\%}& \makecell[c]{0.0626\\(0.0171)}  & \makecell[c]{0.0482\\ (0.0111)} &  \textit{-23.00\%}\\
\makecell[c]{Out-of-sample\\ $(P1)$} & \makecell[c]{0.1784 \\(0.0375)} & \makecell[c]{0.1297 \\(0.0282)} & \textit{- 27.30\%} & \makecell[c]{0.1113\\ (0.0239)} & \makecell[c]{0.0837 \\(0.0196)} & \textit{- 24.80\%}\\
\makecell[c]{Out-of-sample\\ $(P2)$}& \makecell[c]{0.3685\\
(0.0688)} & \makecell[c]{0.1999\\
(0.0312)} & \textit{- 45.75\%} &  \makecell[c]{0.1451 \\ (0.0445)} &\makecell[c]{0.0894 \\(0.0184)} & \textit{- 38.39\%}  \\
\makecell[c]{Out-of-sample\\ $(P3)$} & \makecell[c]{0.2491\\
(0.0672)} &\makecell[c]{0.1780\\
(0.0317)} & \textit{- 28.54\%} & \makecell[c]{0.1723 \\ (0.0501)} & \makecell[c]{0.1000 \\(0.0307)} & \textit{- 41.96\%} \\
\hline
\end{tabular}}
\end{center}
{\textit{Note.} This table reports the RMSRE and its standard deviation (in parentheses) for four models, covering in-sample and out-of-sample calibration. The models include two ordinary benchmarks (Heston, JH) and their extended versions (Composite Heston, Composite JH). Values in italics (``Heston Err. Reduced", ``JH Err. Reduced") are the relative error reductions of composite models versus their ordinary counterparts. The in-sample period is April 4, 2007–April 1, 2015. The out-of-sample period is split into three sub-periods: 2015/3/31–2016/3/13 (P1), 2016/3/14–2017/2/5 (P2), 2017/2/6–2017/12/31 (P3). Prior to each out-of-sample calibration, structural parameters are fixed using two-month market data from the last period.}
\end{table}

To compare different models, we define a pairwise $t$-statistics between model $i$ and model $j$ as
\begin{equation}
\label{t}
    t\text{-statistics} = \frac{\overline{\text{RMSRE}_{i,t} -\text{RMSRE}_{j,t}}}{\operatorname{SE}(\text{RMSRE}_{i,t} -\text{RMSRE}_{j,t})},
\end{equation}
where $\operatorname{SE}(\cdot)$ denotes the sample standard error. We adjust the standard error calculation for serial dependence based on \cite{newey1987simple}, with the maximum lag estimated as the fourth root of the number of trading days in the sample following \cite{greene2000econometric}. The resulting pairwise $t$-statistics are presented in Table \ref{indicators}.

Table \ref{indicators} evaluates the in-sample joint calibration of ordinary models (Heston, JH) and their composite versions, using multi-error metrics (Panel A) and statistical significance tests (Panel B).

Panel A shows composite models outperform ordinary ones in both error magnitude and stability. For error reduction, the improvement is most notable in VIX-related calibration: the Composite Heston cuts VIX’s RMSE by 41.58\% (from 0.1229 to 0.0718), and the Composite JH reduces it by 30.11\% (from 0.0548 to 0.0383). This advantage extends to other metrics—for example, the Composite Heston’s MAE drops to 0.0359 (vs. Heston’s 0.0586) and the Composite JH’s MAE to 0.0198 (vs. JH’s 0.0277). Importantly, composite models also show better calibration stability (reflected in smaller standard deviations). For RMSRE, the Composite Heston’s standard deviation (0.0353) is lower than Heston’s (0.0480); the Composite JH’s (0.0111) is smaller than JH’s (0.0171). Similarly, the Composite JH’s MAE standard deviation (0.0057) is 29.6\% lower than JH’s (0.0081), confirming more consistent performance.

Panel B’s pairwise $t$-statistics further validate these results. All composite models have negative statistics against ordinary counterparts—for example, the Composite JH shows -22.55 vs. Heston and -14.47 vs. JH—indicating their outperformance is statistically significant at the 99.9\% confidence level. Overall, composite models deliver more accurate and stable in-sample joint calibration, with robustness across metrics.

\begin{table}
\caption{Comparison of in-sample joint calibration performance\label{indicators}}
\begin{center}
    \begin{tabular}{ccccc}
       \hline & Heston & \makecell[c]{Composite\\ Heston} & JH & 
       \makecell[c]{Composite\\ JH}\\\hline
       \multicolumn{5}{c}{Panel A. Summary of Statistics of Performance}\\\hline
        \makecell[c]{RMSRE\\ \quad} & \makecell[c]{0.1263\\ (0.0480)} & \makecell[c]{0.0851\\ (0.0352)} & \makecell[c]{0.0626\\ (0.0171)} & \makecell[c]{0.0482\\
       (0.0111)}\\
       \makecell[c]{RMSE (SPX)\\ \quad} & \makecell[c]{0.0292 \\(0.0144)}& \makecell[c]{0.0260 \\(0.0154)} & \makecell[c]{0.0174 \\(0.0065)} & \makecell[c]{0.0149 \\(0.0055)}\\
       \makecell[c]{RMSE (VIX)\\ \quad} & \makecell[c]{0.1229 \\ (0.0370)} & \makecell[c]{0.0718 \\(0.0381)} & \makecell[c]{0.0548 \\(0.0176)} & \makecell[c]{0.0383 \\(0.0136)}\\
       \makecell[c]{RMSE (Aggregate)\\ \quad} & \makecell[c]{0.0760 \\(0.0230)} & \makecell[c]{0.0489\\ (0.0236)} & \makecell[c]{0.0361\\ (0.0103)} & \makecell[c]{0.0266 \\(0.0077)}\\
       \makecell[c]{MAE\\ \quad} & \makecell[c]{0.0586 \\(0.0185)} & \makecell[c]{0.0359\\ (0.0180)} & \makecell[c]{0.0277\\ (0.0081)} & \makecell[c]{0.0198 \\ (0.0057)}\\\hline
              \multicolumn{5}{c}{Panel B. Pairwise $t$-statistics}\\\hline
        Heston & 0 & 16.05 & 18.69 & 22.55\\
        Composite Heston & -16.05 & 0 & 9.21 & 16.28\\
        JH & -18.69 & -9.21 &  0 & 14.47\\
        Composite JH & -22.55 & -16.28 &-14.47 & 0\\\hline
    \end{tabular}
    \end{center}
    {\textit{Note.} Panel A compares the in-sample joint calibration performance of four models (Heston, Composite Heston, JH, Composite JH) using error metrics including RMSE, RMSRE and MAE. All values are the daily average of errors, with the standard deviation of the errors shown in parentheses. Panel B reports pairwise $t$-statistics (defined in Equation \eqref{t}), where the $i$-th row corresponds to Model $i$ and the $j$-th column to Model $j$. A significantly negative statistic indicates Model $i$ outperforms Model $j$, while a positive statistic indicates the opposite.}
\end{table}

Table \ref{RMSRE} contrasts the in-sample joint calibration performance (via RMSRE) of our proposed models with existing models from \cite{yuan2020new}, while also reporting each model’s parameter count (N). First, models with jumps (JH, Composite JH) show substantial advantages over continuous models (Heston, Composite Heston). For example, the ordinary JH (0.0626) has a 50.4\% lower RMSRE than the ordinary Heston (0.1263), which aligns with expectations: jumps capture the short-maturity skews of SPX and VIX options that continuous models fail to replicate.

Notably, our composite jump model also outperforms existing complex models. The Composite JH has fewer parameters (N=12) than \(\text{3FU-CJ}^{*}\) (N=16) and \(\text{4F-ICJ}^{*}\) (N=21), but its RMSRE (0.0482) is lower than both (0.0664 and 0.0581, respectively). This shows the Composite JH achieves better calibration accuracy with a more parsimonious structure.

\begin{table}
\caption{In-sample comparison of joint calibration performance\label{RMSRE}}
\begin{center}
\begin{tabular}{cccccccccc}
\hline Model & Heston & \makecell[c]{Comp. \\Heston} & JH & \makecell[c]{Comp. \\JH} & $\text{2F}^{*}$ & $\text{3FU-CJ}^{*}$ & $\text{4F-ICJ}^{*}$\\
\hline 
N & 5 & 9 & 8 & 12 & 13 & 16 & 21\\
RMSRE & 0.1263 & 0.0851 & 0.0626 & 0.0482 & 0.1010 & 0.0664 &  0.0581\\
\hline
\end{tabular}
\end{center}
{\textit{Note.} This table compares the in-sample RMSRE (defined in Equation \eqref{error}) between two types of models: the four proposed models (Heston, Composite Heston, JH, Composite JH) and three existing models from \cite{yuan2020new} (marked with\({}^*\): \(\text{2F}^{*}\), \(\text{3FU-CJ}^{*}\), \(\text{4F-ICJ}^{*}\)). The ``N" row represents the number of model parameters.}
\end{table}

Table \ref{Para} presents the mean values and the standard errors of calibrated parameters for the four models, revealing key structural features that explain their calibration performance, consistent with Table \ref{RMSRE} where Composite JH performed best.

First, the leverage parameter $\rho_u$ and fine structure parameter $Y$ show notable stability: the coefficient of vairation for $\rho_u$ is $-0.16$ ($-0.19$) for Heston (Composite Heston) model and $0.057$ ($0.044$) for JH (Composite JH) model. Second, the time change $V$ drives volatility dynamics as expected. The initial value $v_0$ is 1.2070 for Composite Heston and 1.5118 for Composite JH, both greater than 1, confirming that V first accelerates volatility dynamics. Over the long term, $V$ reverts to lower mean levels: $\theta_v$ ($V$’s mean-reversion level) is 1.0909 for Composite Heston and 0.9224 for Composite JH, revealing an ``accelerate-then-revert" pattern that improves short-term volatility fitting.

Third, the incorporation of $V$ slightly softens the dynamics of time change $U$ in the top-performing JH and Composite JH models. For $U$’s mean-reversion level $\theta_u$, it falls from 0.5105 (JH) to 0.4783 (Composite JH); while parameter $C$, the overall level of activity, is reduced more than half, reflecting milder fluctuations in U. This proves that $V$ contributes to both the model’s overall volatility and its vol of vol.

\begin{table}
\caption{In-sample model parameter statistics
\label{Para}}
\begin{center}
\begin{tabular}{ccccccccccccc}
\hline & Heston & Composite Heston & JH & Composite JH  \\
\hline \makecell[c]{$\kappa_u$\\ \quad} & \makecell[c]{14.3761 \\ (0.3193)} & \makecell[c]{10.1081 \\ (0.3674)}  & \makecell[c]{2.7048\\ (0.1036)} & \makecell[c]{3.9423\\ (0.0816)} \\
\makecell[c]{$\theta_u$\\ \quad} & \makecell[c]{0.0750\\(0.0023)} & \makecell[c]{0.1646\\(0.0080)} & \makecell[c]{0.5105 \\(0.0152)} & \makecell[c]{0.4782\\ (0.0081)}\\
\makecell[c]{$\sigma_u/\eta_u$\\ \quad} &\makecell[c]{1.9859 \\(0.0307)} & \makecell[c]{1.9973 \\ (0.0441)} & \makecell[c]{4.1362 \\(0.1635)} & \makecell[c]{7.2706 \\(0.2009)} \\
\makecell[c]{$\kappa_v$\\ \quad} &  & \makecell[c]{2.7114 \\(0.0716)} &  & \makecell[c]{4.4931 \\(0.1358)}\\
\makecell[c]{$\theta_v$\\ \quad} & & \makecell[c]{1.0908 \\(0.0279)} &  & \makecell[c]{0.9224 \\(0.0167)}  \\
\makecell[c]{$\sigma_v$\\ \quad} &  & \makecell[c]{0.4124 \\(0.0069)} &  & \makecell[c]{0.4194 \\(0.0098)} \\
\makecell[c]{$C$\\ \quad} & & & \makecell[c]{0.2213 \\(0.0084)} & \makecell[c]{0.1071 \\(0.0048)}\\
\makecell[c]{$G$\\ \quad} & & & \makecell[c]{2.2288 \\(0.0538)} & \makecell[c]{3.4883 \\(0.0684)} \\
\makecell[c]{$M$\\ \quad} & & & \makecell[c]{22.4491 \\(1.3127)} & \makecell[c]{24.8861 \\(0.8723)} \\
\makecell[c]{$Y$\\ \quad} & & & \makecell[c]{1.6166 \\(0.0046)} & \makecell[c]{1.6975 \\(0.0037)} \\
\makecell[c]{$\rho_u$\\ \quad} & \makecell[c]{-0.7126 \\ (0.0056)} & \makecell[c]{-0.6891 \\(0.0065)}
 & &\\
\makecell[c]{$u_0$\\ \quad} & \makecell[c]{0.0384 \\(0.0041)} & \makecell[c]{0.0484 \\(0.0030)} & \makecell[c]{0.1023 \\(0.0056)} & \makecell[c]{0.0720 \\(0.0026)} \\
\makecell[c]{$v_0$\\ \quad} &  & \makecell[c]{1.2070 \\(0.0134)}  & & \makecell[c]{1.5115 \\(0.0226)} \\
\hline
\end{tabular}
\end{center}
{\textit{Note.} This table reports the mean values of calibrated parameters (with standard errors in parentheses) for four models during the in-sample period (April 4, 2007–April 1, 2015). The model parameters are specified in Example \ref{eg comp heston} and \ref{eg comp jh}.}
\end{table}

Table \ref{sort} evaluates the in-sample calibration performance of the four models across different moneyness and maturity ranges, separately for SPX and VIX options, to identify where composite models add the most value.

For SPX options, composite models consistently outperform their ordinary counterparts in most scenarios. In terms of moneyness (Panel A.1), Composite Heston reduces RMSRE versus the original Heston across all intervals—for example, from 0.1681 to 0.1084 when \(0.80 \leq K/F<0.90\), and from 0.1995 to 0.1276 when \(1.0 \leq K/F<1.1\). For Composite JH, it outperforms JH in both small and large moneyness: when \(K/F<0.90\), its RMSRE $(0.0623, 0.0581)$ is lower than JH’s $(0.0720, 0.0632)$, and, when $K/F > 1$, its RMSRE $(0.0680, 0.0669)$ is lower than JH's $(0.0827, 0.0685)$. By maturity (Panel A.2), Composite Heston outperforms Heston across all maturity ranges, with the largest improvement in short maturity (\(<15\) days: 0.2225 vs. 0.3462). Composite JH excels in short-to-medium maturity (\(<90\) days: e.g., \(<15\) days 0.0828 vs. JH’s 0.1180) but falls slightly behind in long maturity (\(>180\) days: 0.0657 vs. JH’s 0.0565).

For VIX options, the advantage of composite models is even more pronounced (Panel B). In moneyness terms (Panel B.1), both composite models prevail for \(K/F<1\) (out-of-the-money or at-the-money options): Composite Heston cuts RMSRE from 0.3029 to 0.1068 when \(K/F<0.8\), and Composite JH reduces it from 0.1059 to 0.0457 in the same interval—nearly a 60\% drop. By maturity (Panel B.2), composite models achieve ~50\% RMSRE reduction across short-to-long terms: for short maturity (\(<30\) days), Composite JH’s RMSRE (0.0467) is nearly half of JH’s (0.0825); for long maturity (\(>120\) days), it falls from 0.0867 (JH) to 0.0475 (Composite JH). 

These results confirm that the composite time change plays a critical role in calibrating VIX options. By reconstructing the ``vol of vol" (volatility of volatility), it effectively improves fitting accuracy for both moneyness and maturity ranges, especially for VIX options where ordinary models struggle most.

\begin{table}[H]
\caption{Comparison of in-sample RMSRE across different moneyness and maturity ranges
\label{sort}}
\begin{center}
\begin{tabular}{ccccc}
\hline & Heston & Composite Heston & JH & Composite JH\\
\hline \multicolumn{5}{c}{Panel A: SPX Options} \\
 \multicolumn{5}{l}{Panel A.1: Sorting by moneyness} \\
\hline $K/F<0.80$ & 0.1309 & 0.1086 & 0.0720 & 0.0623  \\
$0.80 \leq K/F<0.90$ & 0.1681 & 0.1084 & 0.0632 & 0.0581 \\
 $0.90 \leq K/F<0.95$ & 0.1397 & 0.0910 & 0.0440 & 0.0482 \\
 $0.95 \leq K/F<1.0$ & 0.1425 & 0.0822 & 0.0697 & 0.0498\\
$1.0 \leq K/F<1.1$ & 0.1996 & 0.1276 & 0.0827 & 0.0680\\
 $1.1 \leq K/F$ & 0.1738 & 0.1399 & 0.0685 & 0.0669  \\
\\ \multicolumn{5}{l}{Panel A.2: Sorting by days to maturity} \\
\hline $<15$ & 0.3463 & 0.2225 & 0.1180 & 0.0828 \\
 15-30 & 0.1534 & 0.1365 & 0.0805 & 0.0757 \\
 30-50 & 0.1084 & 0.0955 & 0.0705 & 0.0584  \\
 50-60 & 0.1133 & 0.0741 & 0.0551 & 0.0515\\
 60-90 & 0.1315 & 0.0783 & 0.0595 & 0.0493 \\
 90-180 &  0.1613 & 0.0910 & 0.0573 & 0.0464 \\
 $>180$ & 0.1941 & 0.1095 & 0.0565 & 0.0657 \\
\hline \\\multicolumn{5}{c}{Panel B: VIX Options} \\ \multicolumn{5}{l}{Panel B.1: Sorting by moneyness} \\
\hline $K/F<0.8$ & 0.3029 & 0.1068 & 0.1059 & 0.0457 \\
 $0.8 \leq K/F<0.9$ & 0.1688 & 0.0813 & 0.0738 & 0.0376 \\
 $0.9 \leq K/F<1.0$ & 0.1747 & 0.0765 & 0.0563 & 0.0335\\
$1.0 \leq K/F<1.2$ & 0.1393 & 0.0681 & 0.0461 & 0.0304 \\
 $1.2 \leq K/F<1.5$ & 0.0968 & 0.0677 & 0.0513 & 0.0384 \\
 $1.5 \leq K/F<1.8$ &0.1035 & 0.0909 & 0.0593 & 0.0525\\
 $1.8 \leq K/F$ & 0.1461 & 0.1276 & 0.0665 & 0.0679 \\
\\ \multicolumn{5}{l}{Panel B.2: Sorting by days to maturity} \\
\hline $< 30$ & 0.2062 & 0.1254 & 0.0825 & 0.0467 \\
 30-60 & 0.1531 & 0.0822 & 0.0657 & 0.0397\\
60-90 & 0.1402 & 0.0748 & 0.0533 & 0.0379 \\
90-120 & 0.1403 & 0.0756 & 0.0625 & 0.0399 \\
$>120$ & 0.1531 & 0.0855 & 0.0867 & 0.0475 \\
\hline
\end{tabular}
\end{center}
{\textit{Note.} This table reports the daily average RMSRE of implied volatilities for four models during the in-sample period, grouped by two dimensions: moneyness and days to maturity $\tau$. RMSRE measures the alignment between model-predicted and market-observed implied volatilities. Panel A focuses on SPX options, with Panel A.1 sorting results by 6 moneyness intervals (from \(K/F<0.80\) to \(K/F \geq 1.1\)) and Panel A.2 by 7 maturity intervals (from \(<15\) days to \(>180\) days). Panel B focuses on VIX options, with Panel B.1 sorting results by 7 moneyness intervals (from \(K/F<0.80\) to \(K/F \geq 1.8\)) and Panel B.2 by 5 maturity intervals (from \(<30\) days to \(>120\) days).}
\end{table}
\subsection{Evaluation of Calibrated CTC Models}
\label{evaluations}
Table \ref{characteristics} focuses on fitting the shape characteristics of the implied volatility surface (ATM IV and Skew) in the joint SPX-VIX market, using two RMSE metrics (Daily and Sample RMSE) to test short-term and long-term calibration accuracy.

First, we clarify the key definitions: ATM IV is derived via linear interpolation for log moneyness $k:= \ln (K/F)$ near 0:
\begin{equation}
\label{IV}\text{IV}(\tau) = \frac{k_+}{k_+ - k_-} \sigma(k_-, \tau) - \frac{k_-}{k_+ - k_-} \sigma(k_+, \tau),
\end{equation}
where $k_+$ ($k_-$) is the smallest positive 
(largest negative) log moneyness of the implied volatility in the market. The SPX (VIX) volatility skew is defined as the slope of the two implied volatilities whose moneyness $\exp{(k_1)}$ and $\exp{(k_2)}$ are the closest to $0.7$ ($0.8$) and $1.2$ ($1.5$), respectively. That is,
\begin{equation}
    \label{Skew}\text{Skew}(\tau) = \frac{\sigma(k_1,\tau) - \sigma(k_2, \tau)}{k_1 - k_2}.
\end{equation}
Short-term refers to $< 30$ days to maturity (for both options), and long-term is $> 270$ days (SPX) or $> 90$ days (VIX). Daily RMSE averages across intraday maturities, while Sample RMSE treats all maturities equally.

For SPX options, composite models generally outperform ordinary ones, with minor exceptions. In short-term metrics: Composite JH reduces short-term ATM IV RMSE by 13.6\% (from 0.0235 to 0.0203) and short-term Skew RMSE by 41.2\% (from 0.1867 to 0.1098) versus JH. Composite Heston also improves on Heston, cutting long-term ATM IV RMSE by 45.0\% (from 0.0334 to 0.0189). The only exception is Composite JH’s weaker long-term SPX performance (long-term ATM IV RMSE: 0.0149 vs. JH’s 0.0100), which is partly justified by long-term SPX options accounting for just 4.72\% of the sample (vs. 18.90\% for short-term)—the model prioritizes fitting the more representative short-term data.

For VIX options, the advantage of composite models is prominent. Short-term gains are most striking: Composite Heston cuts short-term ATM IV RMSE by 37.7\% (from 0.2728 to 0.1699) and short-term Skew RMSE by 31.4\% (from 0.7122 to 0.4886) versus Heston. Composite JH performs even better, reducing short-term ATM IV RMSE by 60.9\% (from 0.0893 to 0.0349) and short-term Skew RMSE by 45.1\% (from 0.2876 to 0.1578) versus JH. Long-term VIX metrics show similar improvements—Composite JH’s long-term ATM IV RMSE (0.0204) is half of JH’s (0.0393).

Notably, Composite JH achieves balanced performance across the joint SPX-VIX market. Its short-term ATM IV RMSE is 0.0203 (SPX) vs. 0.0349 (VIX), and short-term Skew RMSE is 0.1098 (SPX) vs. 0.1578 (VIX)—a narrow gap. By contrast, JH shows large disparities (0.0235 vs. 0.0893 for short-term ATM IV; 0.1867 vs. 0.2876 for short-term Skew). Robustness tests in Appendix \ref{robustness} (using more moneyness pairs) confirm this pattern, proving composite models’ superiority is not sensitive to moneyness selection.

\begin{table}
\caption{Short-term and long-term calibration accuracy of implied volatility characteristics
    \label{characteristics}}
    \begin{center}
    \begin{tabular}{ccccc}
         & Heston & Composite Heston & JH & Composite JH \\\hline
  \multicolumn{5}{c}{Panel A: S\&P 500 options}\\
   \multicolumn{5}{l}{Panel A.1: Daily RMSE}\\\hline
       Short-term ATM IV & 0.0264 & 0.0255 & 0.0235 & 0.0203\\
       Long-term ATM IV & 0.0334 & 0.0189 & 0.0100 & 0.0149\\
    Short-term Skew & 0.4384 & 0.3584 &  0.1867 & 0.1098\\
    Long-term Skew & 0.1007 & 0.0777 & 0.0425 & 0.0613\\
     \multicolumn{5}{l}{Panel A.2: Sample RMSE}\\\hline
     Short-term ATM IV & 0.0266 & 0.0247 & 0.0236 & 0.0206\\
     Long-term ATM IV & 0.0336 & 0.0187 & 0.0098 & 0.0147\\
     Short-term Skew & 0.4397 & 0.3601 &  0.1813 & 0.1121\\
     Long-term Skew &  0.0991 &  0.0773 & 0.0420 & 0.0602 \\\hline
    \multicolumn{5}{c}{Panel B: VIX options}
         \\
    \multicolumn{5}{l}{Panel B.1: Daily RMSE}\\\hline
       Short-term ATM IV & 0.2728 & 0.1699 & 0.0893 & 0.0349\\
       Long-term ATM IV & 0.0621 & 0.0510 & 0.0393 & 0.0204\\
       Short-term Skew & 0.7122 & 0.4886 & 0.2876 & 0.1578\\
        Long-term Skew & 0.2905 & 0.1530 & 0.1050 & 0.0647\\
       \multicolumn{5}{l}{Panel B.2: Sample RMSE}\\\hline
        Short-term ATM IV & 0.2728 & 0.1699 & 0.0893 & 0.0349\\
     Long-term ATM IV & 0.0622 & 0.0520 & 0.0425 & 0.0204\\
     Short-term Skew & 0.7122 & 0.4886 & 0.2876 & 0.1578\\
     Long-term Skew & 0.3074 & 0.1530 &  0.1199 & 0.0640\\\hline
    \end{tabular}
    \end{center}
    {\textit{Note.} This table evaluates the calibration accuracy of IV characteristics (ATM IV and volatility skew) for four models, covering short-term and long-term maturities, and separately for SPX and VIX options. Panel A presents results for SPX options (A.1 = Daily RMSE; A.2 = Sample RMSE), and Panel B presents results for VIX options (B.1 = Daily RMSE; B.2 = Sample RMSE). ATM IV is alculated via linear interpolation (Eq. \eqref{IV}), using the two implied volatilities with log moneyness closest to 0. And volatility skew is defined as the slope of two implied volatilities (Eq. \eqref{Skew})—for SPX options, moneyness is closest to 0.7 and 1.2; for VIX options, moneyness is closest to 0.8 and 1.5. Maturities of SPX (VIX) options smaller than 30 (30) days are considered short-term, and larger than 270 (90) days are considered long-term. Daily RMSE is computed by the root mean square of daily mean absolute errors, and sample RMSE takes the root mean square error over all qualified option maturities.}
\end{table}

We can observe the difference in performance more explicitly in the plot of daily fits of implied volatility characteristics. Figure \ref{JH IV SPX} shows the short-term daily fits of ATM implied volatility for JH and Composite JH in SPX option market. The two models both fit ATM IV quite well, but it can be observed in the difference plot that Composite JH model performs consistently better, typically in the peak values of market IV.

\begin{figure}[!ht]
\caption{The short-term ATM implied volatility of JH and Composite JH in SPX markets}
 \begin{minipage}{\linewidth}
 \subfloat[][Short-term ATM IV]{
 \includegraphics[width=\linewidth]{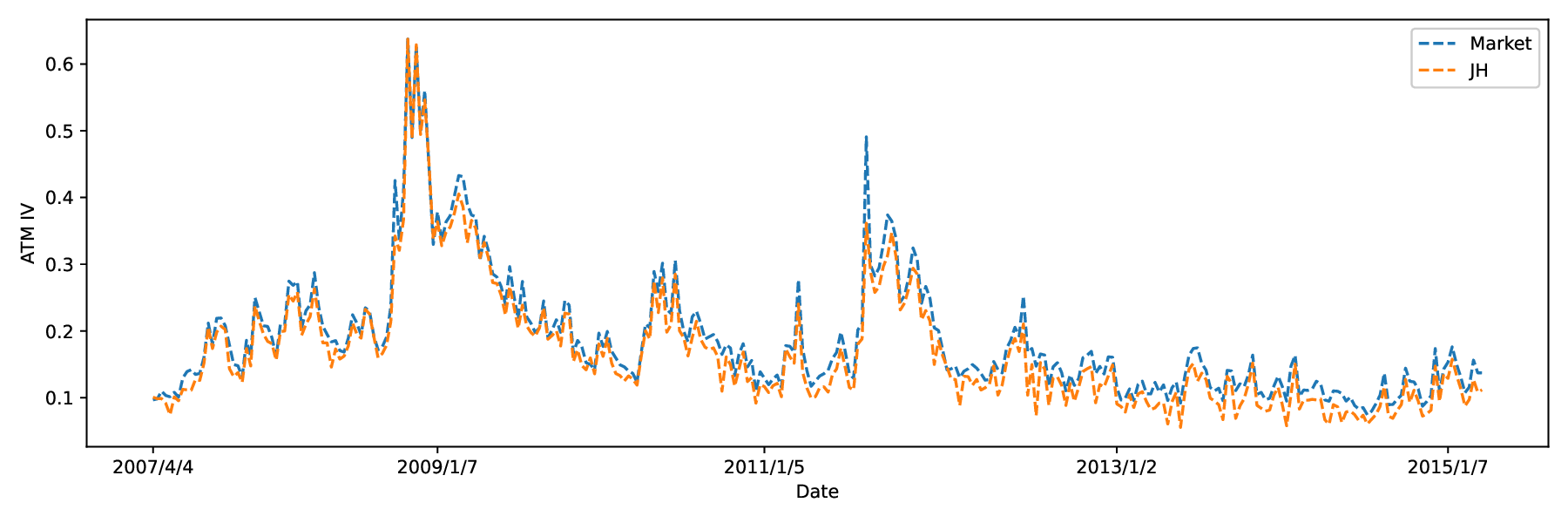}
 }
 \end{minipage} \vspace{-1em}
  \begin{minipage}{\linewidth}
 \subfloat[][Short-term ATM IV]{
 \includegraphics[width=\linewidth]{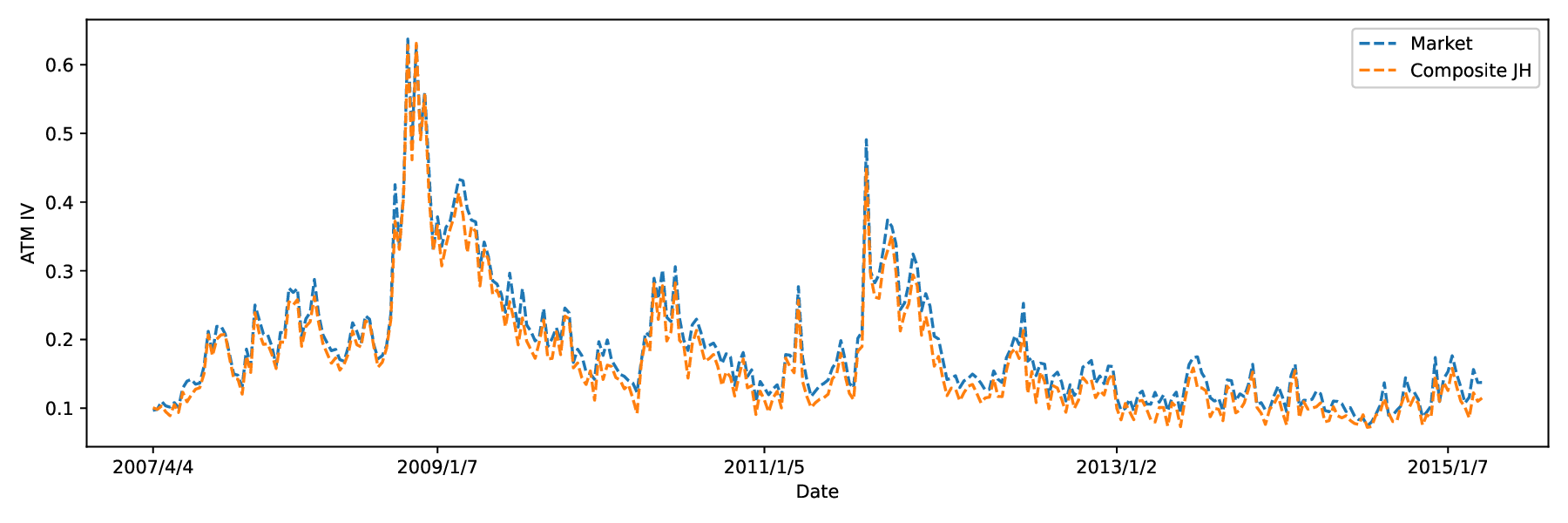}
}
 \end{minipage} \vspace{-1em}
 \begin{minipage}{\textwidth}
\subfloat[][ATM IV difference]{
\includegraphics[width=\textwidth]{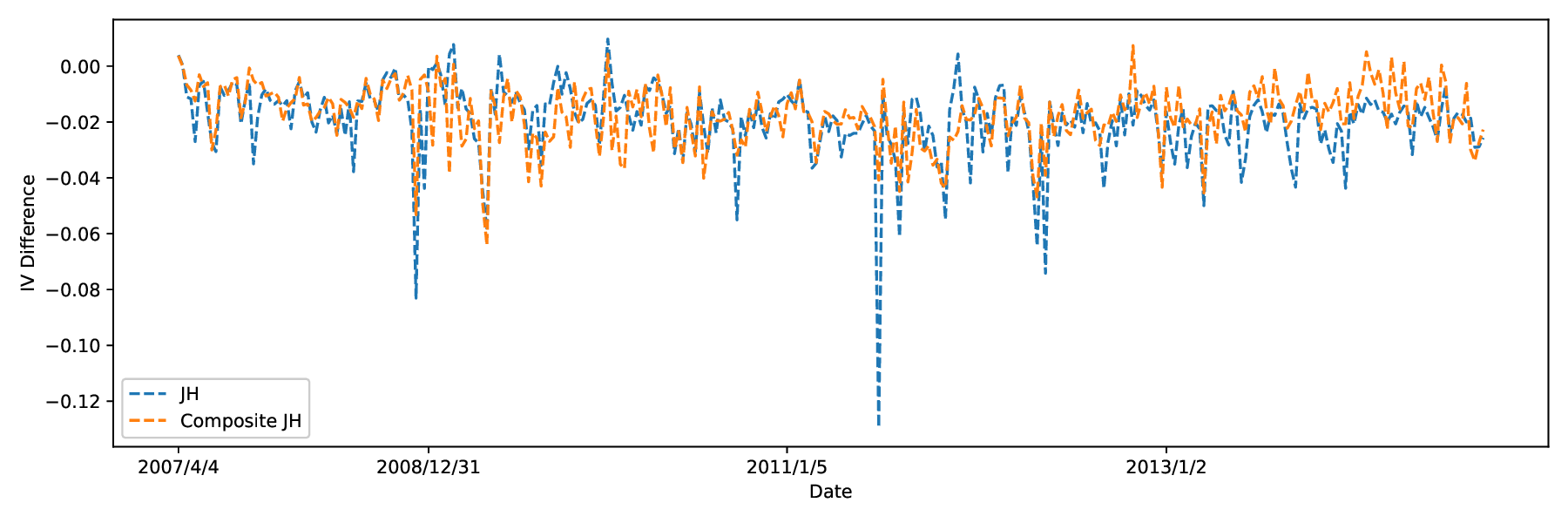}
}
\end{minipage}
\vspace{-2pt}
\footnotesize \textit{Note. } The plots compare the short-term ATM IV, defined by Eq. \eqref{IV}, between the models (JH and Composite JH) and the market. The first row plots the ATM IV of JH model, the second row plots that of Composite JH model, and the last row compares the corresponding difference (model - market) between models and the market. Maturities of SPX options smaller than 30 days are considered short-term.\label{JH IV SPX}
\end{figure}

Figure \ref{JH IV SPX Skew} plots the short-term daily fits of volatility skew for JH and Composite JH in SPX option market. We observe comparable performance between the two models, but JH model presents less calibration robustness since it occasionally generates irregular skews. Both Figure \ref{JH IV SPX} and \ref{JH IV SPX Skew} confirm the better performance of Composite JH in short-term SPX option market as shown in Table \ref{characteristics}.

\begin{figure}[!ht]
\caption{The short-term volatility skew of JH and Composite JH in SPX markets}
 \begin{minipage}{\linewidth}
 \subfloat[][Short-term volatility skew]{
 \includegraphics[width=\linewidth]{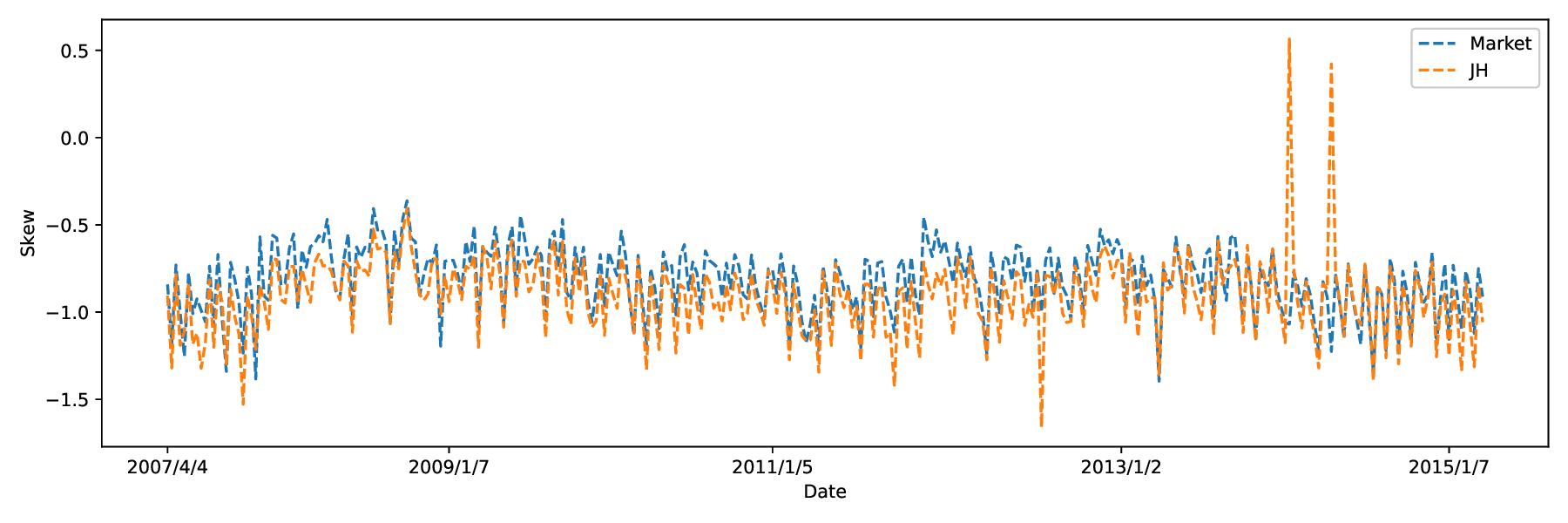}
 }
 \end{minipage} \vspace{-1em}
  \begin{minipage}{\linewidth}
 \subfloat[][Short-term volatility skew]{
 \includegraphics[width=\linewidth]{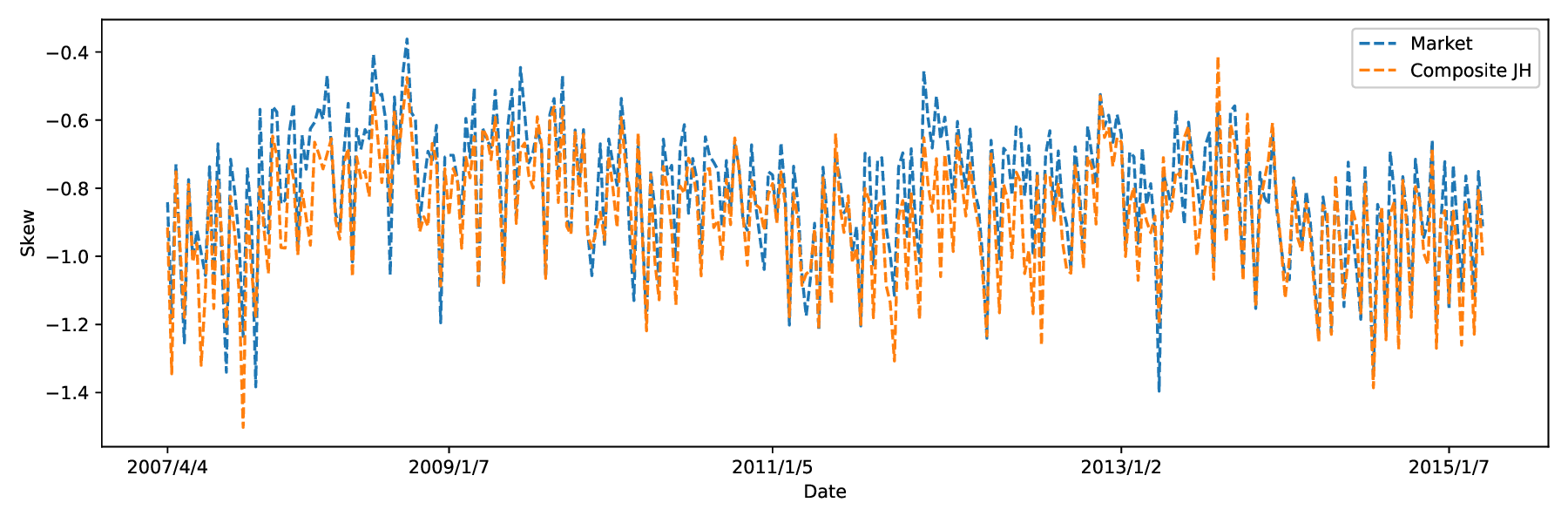}
}
 \end{minipage} \vspace{-1em}
 \begin{minipage}{\textwidth}
\subfloat[][Volatility skew difference]{
\includegraphics[width=\textwidth]{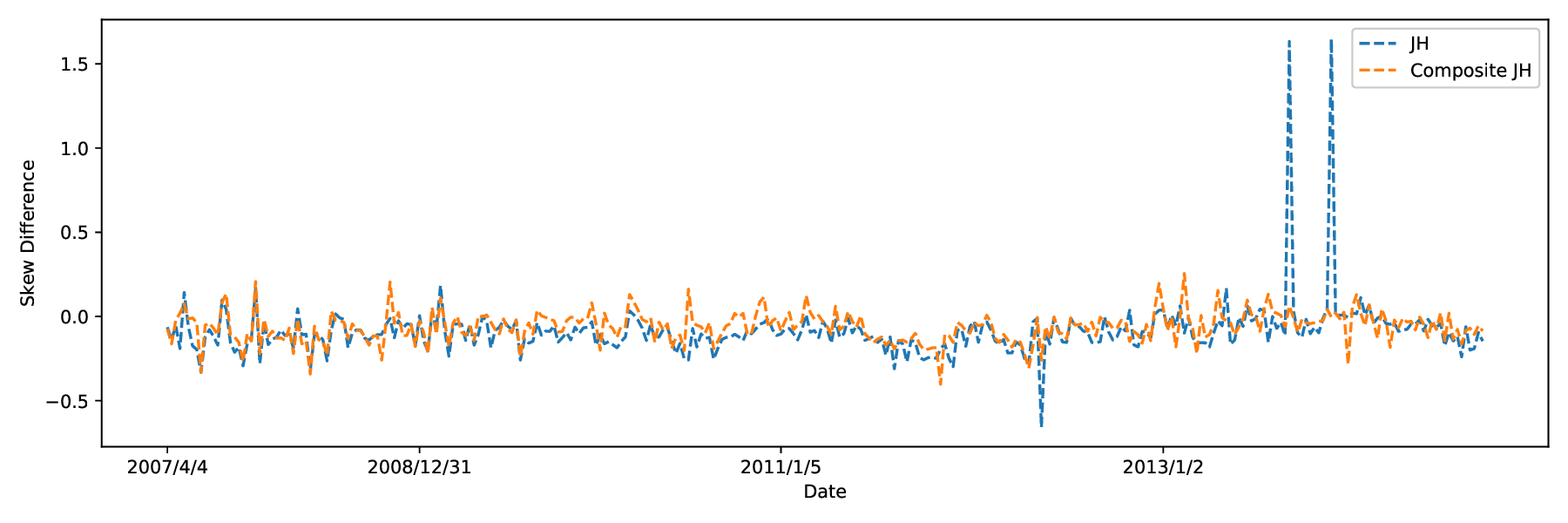}
}
\end{minipage}
\vspace{-2pt}
\footnotesize \textit{Note. } The plots compare the short-term volatility skew, defined by Eq. \eqref{Skew}, between the models (JH and Composite JH) and the market. The first row plots the volatility skew of JH model, the second row plots that of Composite JH model, and the last row compares the corresponding difference (model - market) between models and the market. Maturities of SPX options smaller than 30 days are considered short-term.\label{JH IV SPX Skew}
\end{figure}

Figure \ref{JH IV SPX long term} shows the long-term daily fitting performance of ATM IV for JH and Composite JH model in SPX option market. And the difference plot witnesses a consistently better performance of JH over Composite JH. 

\begin{figure}[!ht]
\caption{The long-term ATM implied volatility of JH and Composite JH in SPX markets}
 \begin{minipage}{\linewidth}
 \subfloat[][Long-term ATM IV]{
 \includegraphics[width=\linewidth]{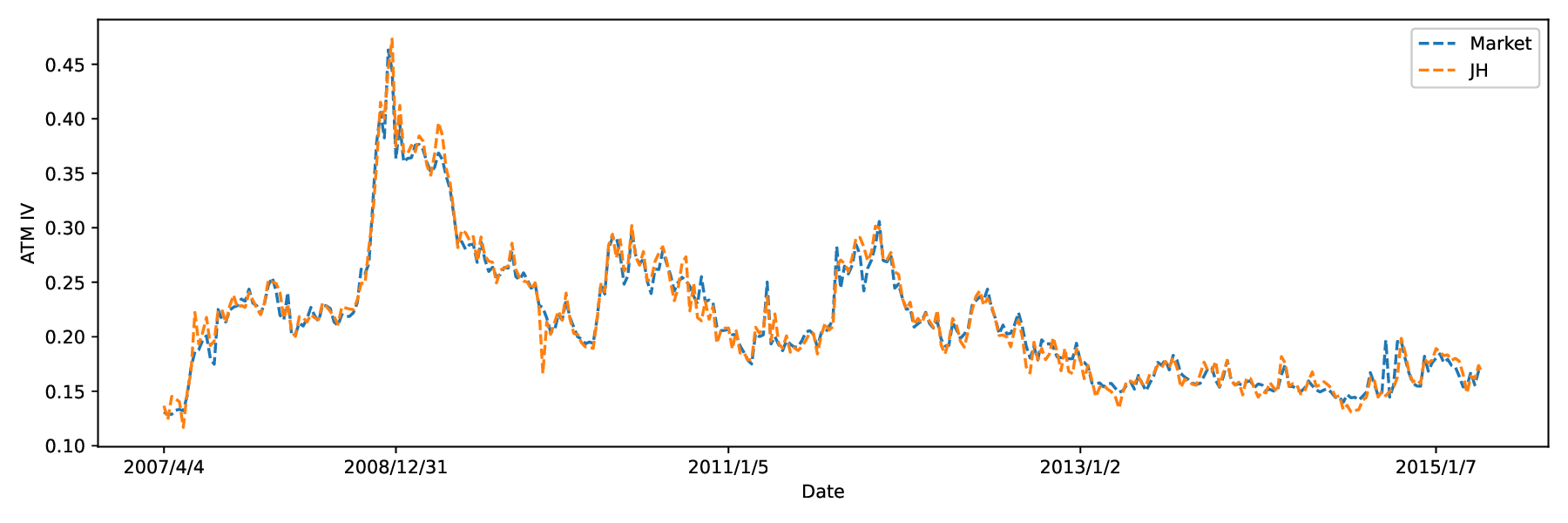}
 }
 \end{minipage} \vspace{-1em}
  \begin{minipage}{\linewidth}
 \subfloat[][Long-term ATM IV]{
 \includegraphics[width=\linewidth]{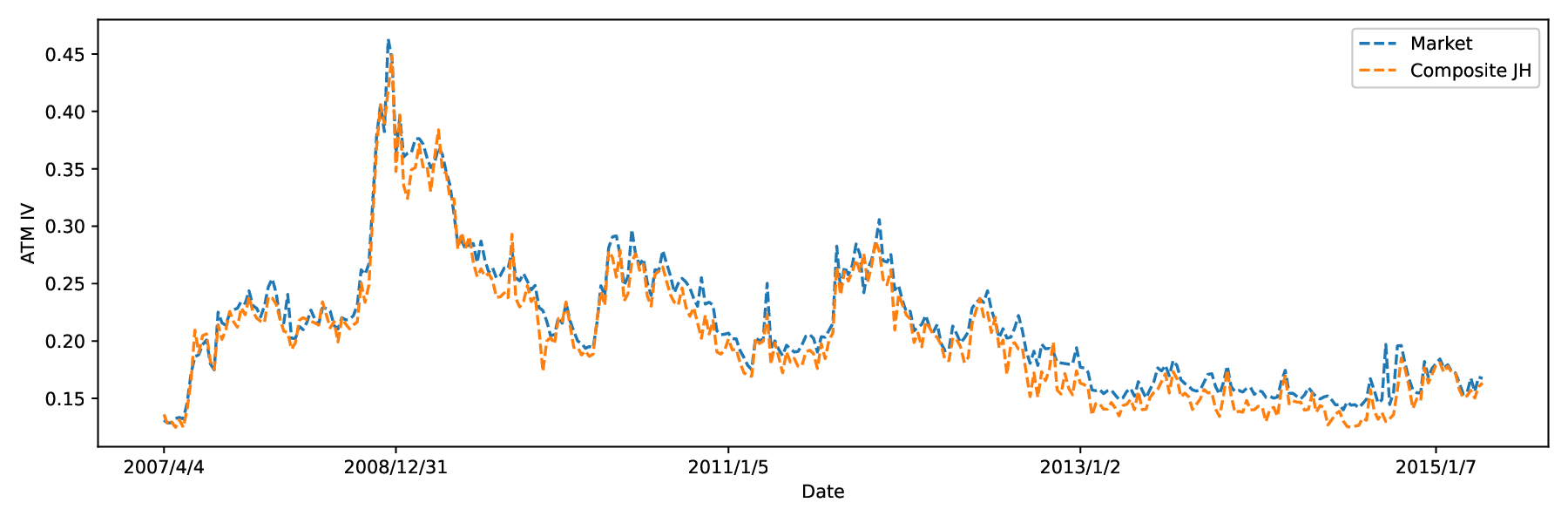}
}
 \end{minipage} \vspace{-1em}
 \begin{minipage}{\textwidth}
\subfloat[][ATM IV difference]{
\includegraphics[width=\textwidth]{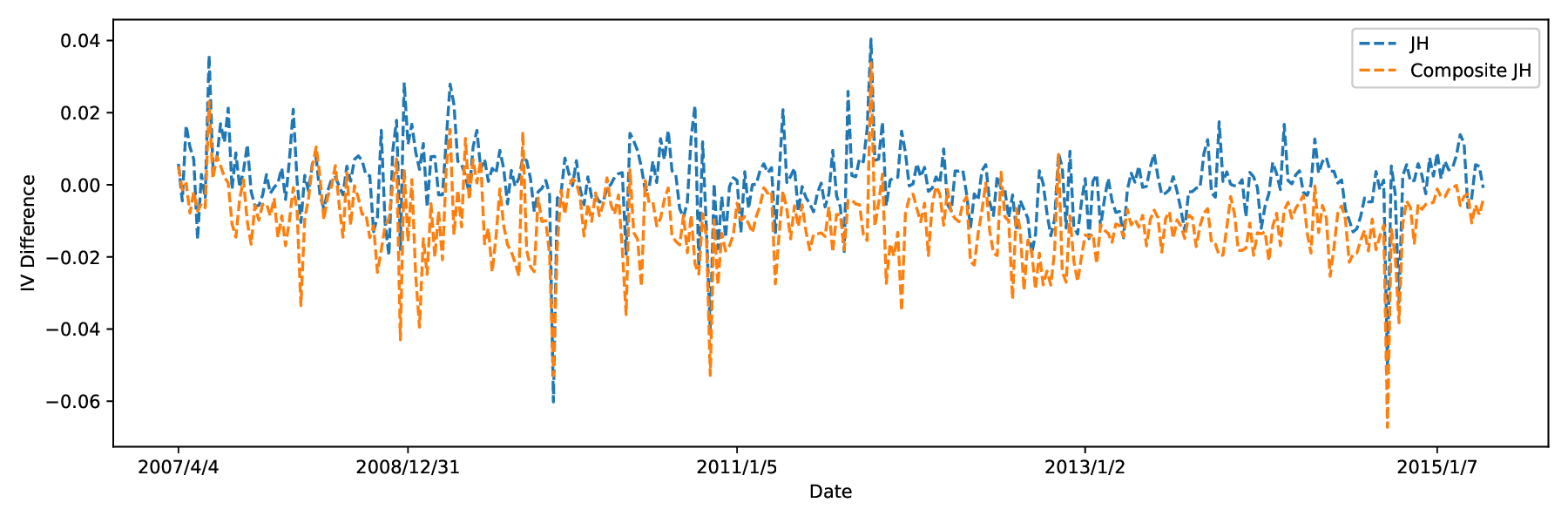}
}
\end{minipage}
\vspace{-2pt}
\footnotesize \textit{Note. } The plots compare the long-term ATM IV, defined by Eq. \eqref{IV}, between the models (JH and Composite JH) and the market. The first row plots the ATM IV of JH model, the second row plots that of Composite JH model, and the last row compares the corresponding difference (model - market) between models and the market. Maturities of SPX options smaller than 270 days are considered short-term.\label{JH IV SPX long term}
\end{figure}

Figure \ref{JH IV SPX skew long term} further confirms the better performance of JH model by comparing the daily fits of long-term volatility skew, where JH presents consistently closer volatility skew to the market values. This advantage of JH model in the long-term SPX option calibration performance is also generally presented in Table \ref{characteristics} by RMSE metric. We also note that both models have relatively inferior performance in long-term skew, which is partly justified by the low proportion of option data in the long term.

\begin{figure}[!ht]
\caption{The long-term volatility skew of JH and Composite JH in SPX markets}
 \begin{minipage}{\linewidth}
 \subfloat[][Long-term volatility skew]{
 \includegraphics[width=\linewidth]{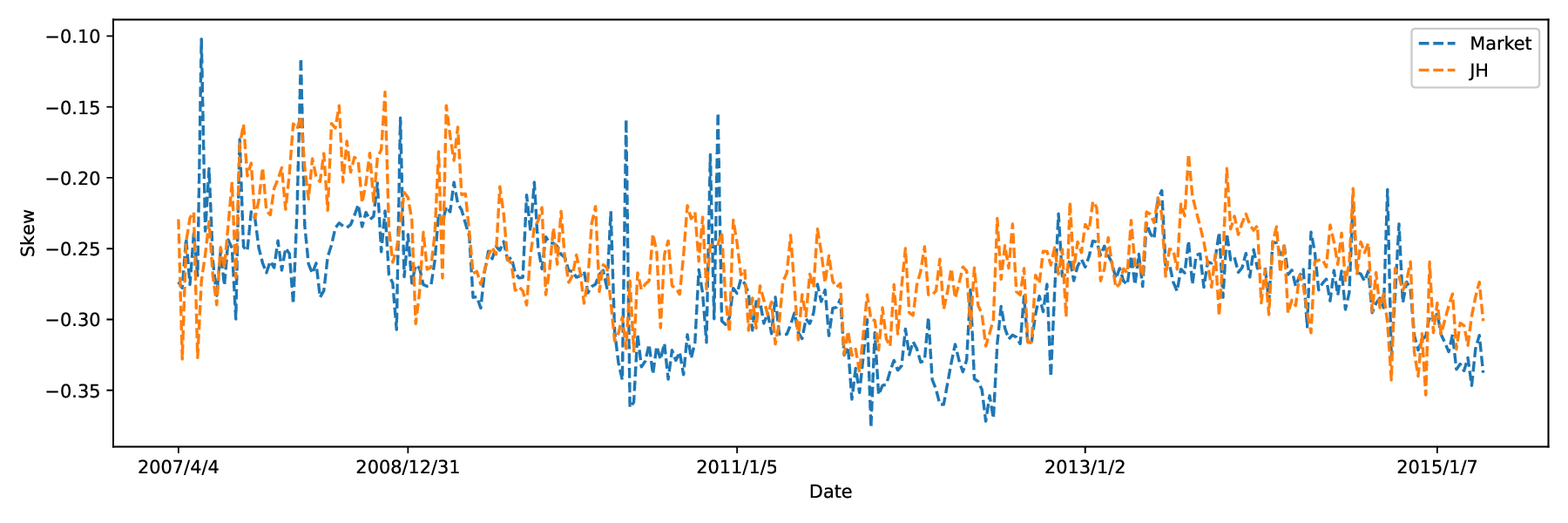}
 }
 \end{minipage} \vspace{-1em}
  \begin{minipage}{\linewidth}
 \subfloat[][Long-term volatility skew]{
 \includegraphics[width=\linewidth]{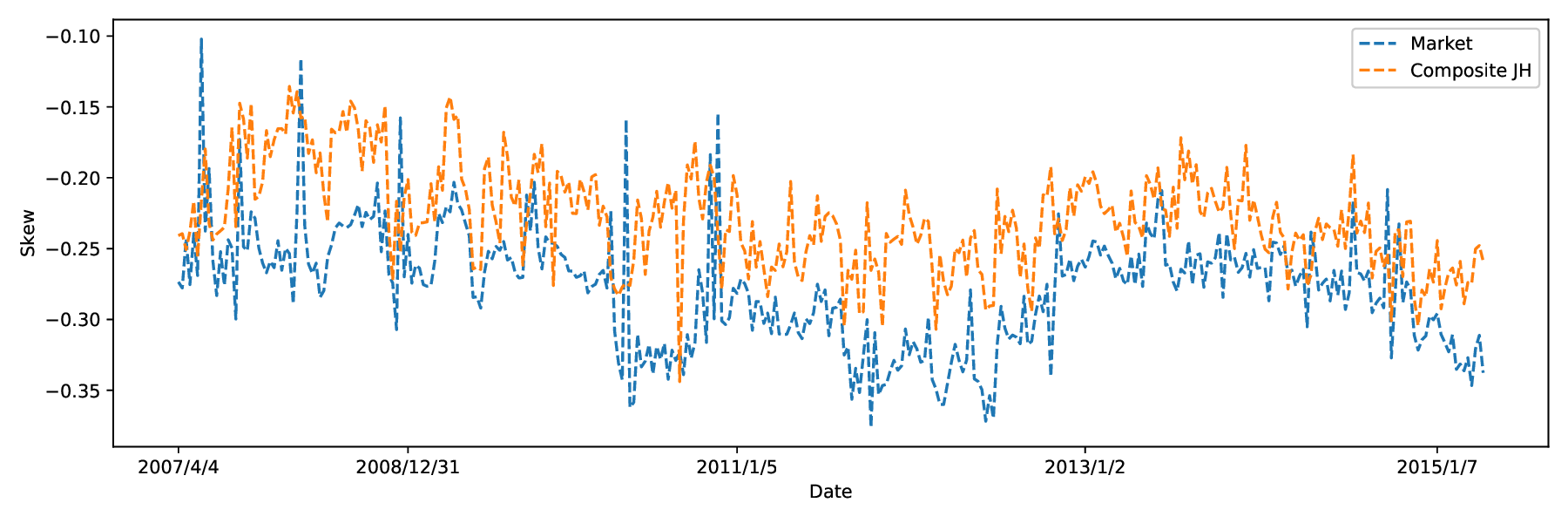}
}
 \end{minipage} \vspace{-1em}
 \begin{minipage}{\textwidth}
\subfloat[][Volatility skew difference]{
\includegraphics[width=\textwidth]{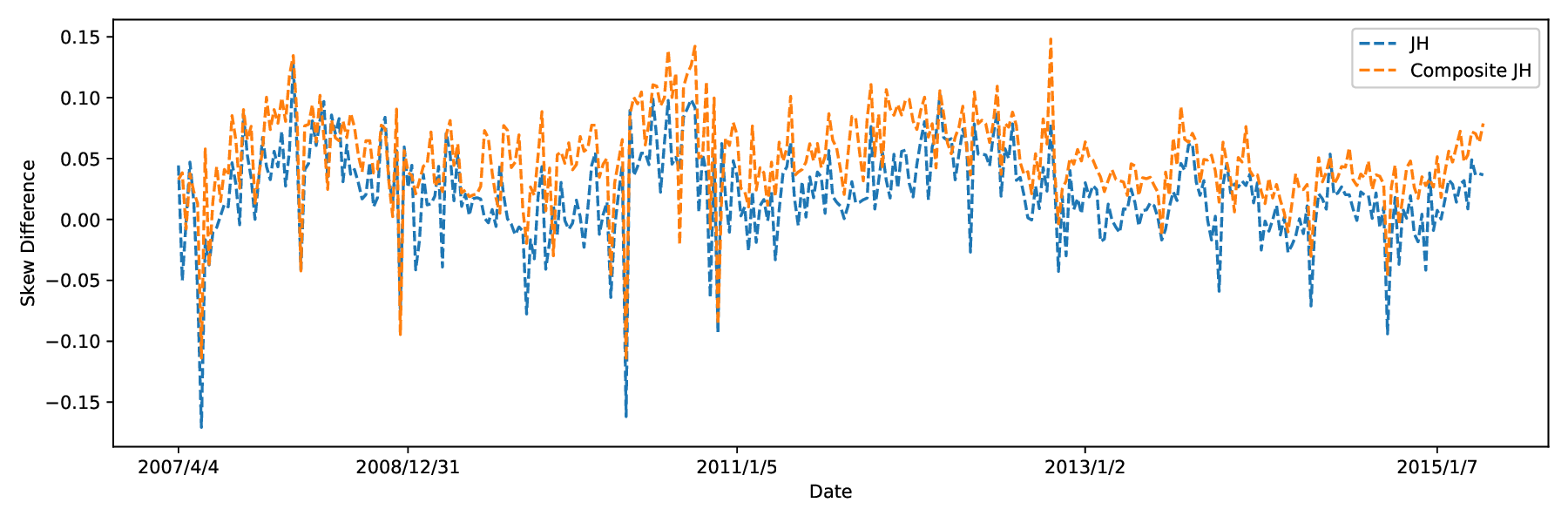}
}
\end{minipage}
\vspace{-2pt}
\footnotesize \textit{Note. } The plots compare the long-term volatility skew, defined by Eq. \eqref{Skew}, between the models (JH and Composite JH) and the market. The first row plots the volatility skew of JH model, the second row plots that of Composite JH model, and the last row compares the corresponding difference (model - market) between models and the market. Maturities of SPX options smaller than 270 days are considered short-term.\label{JH IV SPX skew long term}
\end{figure}

Despite JH model's advantage in the long-term SPX option market, we will observe that Composite JH excels consistly in VIX option market. Figure \ref{JH IV VIX} plots the short-term ATM IV fits in VIX option market for JH and Composite JH. We see that Composite JH accurately fits the market values while JH model fits relatively worse, especially in the latter half period and in extrema of market values. The difference plot also confirms this. 

\begin{figure}[!ht]
\caption{The short-term ATM implied volatility of JH and Composite JH in VIX markets}
 \begin{minipage}{\linewidth}
 \subfloat[][Short-term ATM IV]{
 \includegraphics[width=\linewidth]{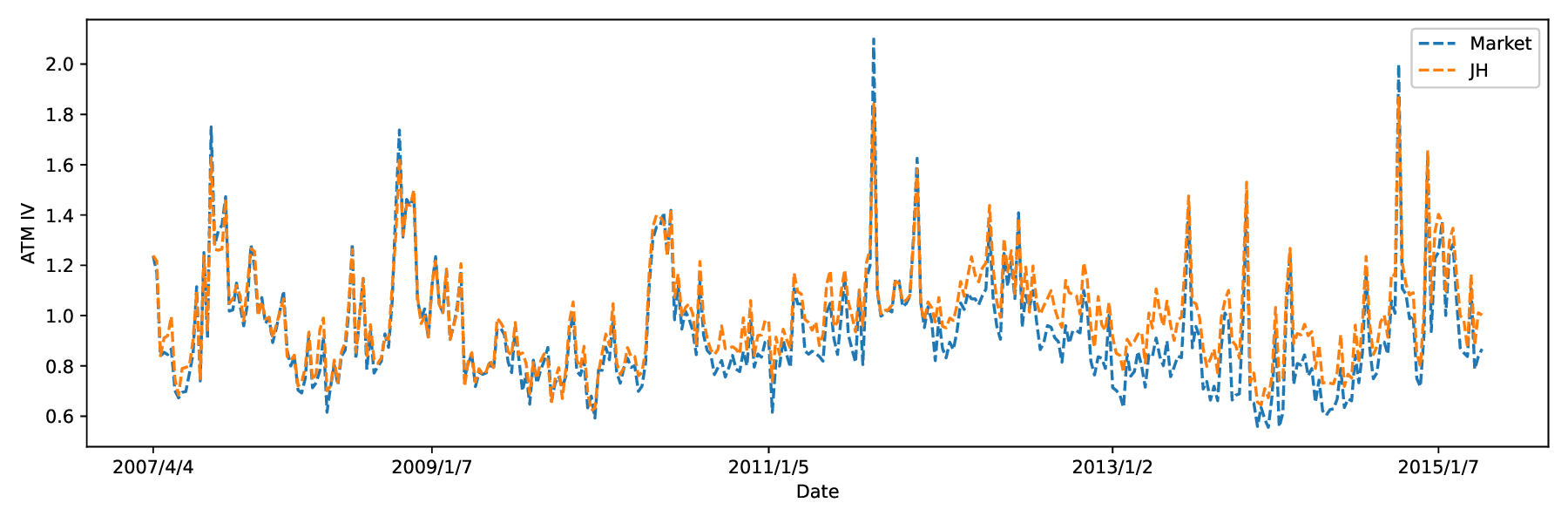}
 }
 \end{minipage} \vspace{-1em}
  \begin{minipage}{\linewidth}
 \subfloat[][Short-term ATM IV]{
 \includegraphics[width=\linewidth]{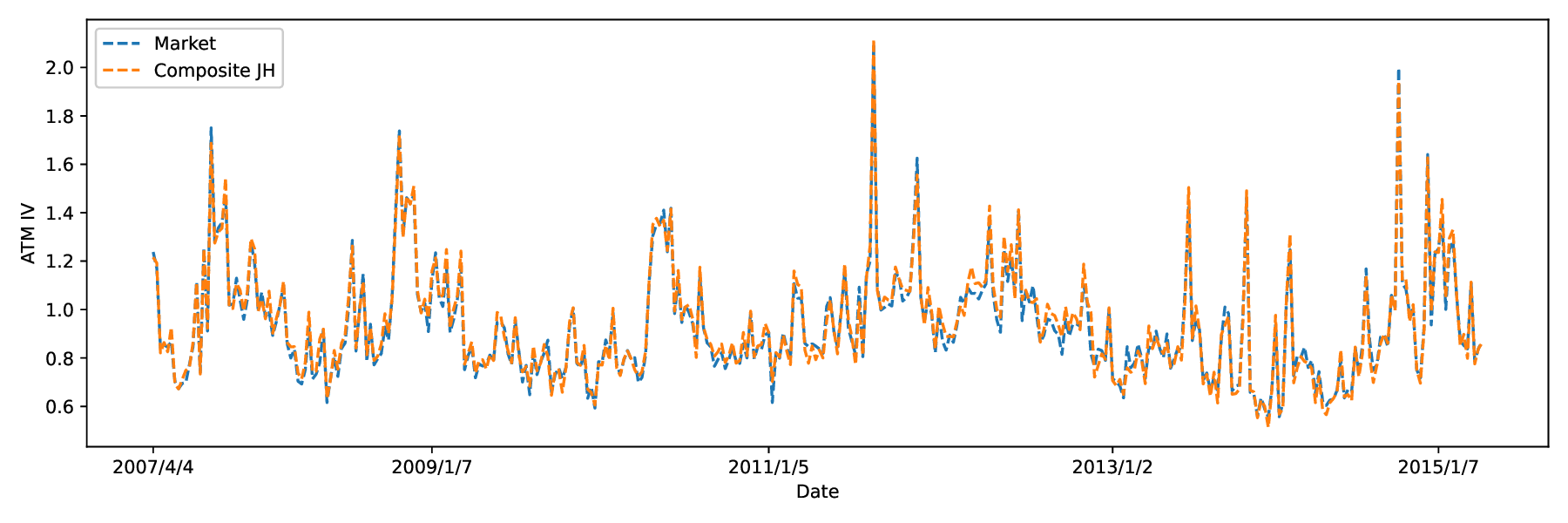}
}
 \end{minipage} \vspace{-1em}
 \begin{minipage}{\textwidth}
\subfloat[][ATM IV difference]{
\includegraphics[width=\textwidth]{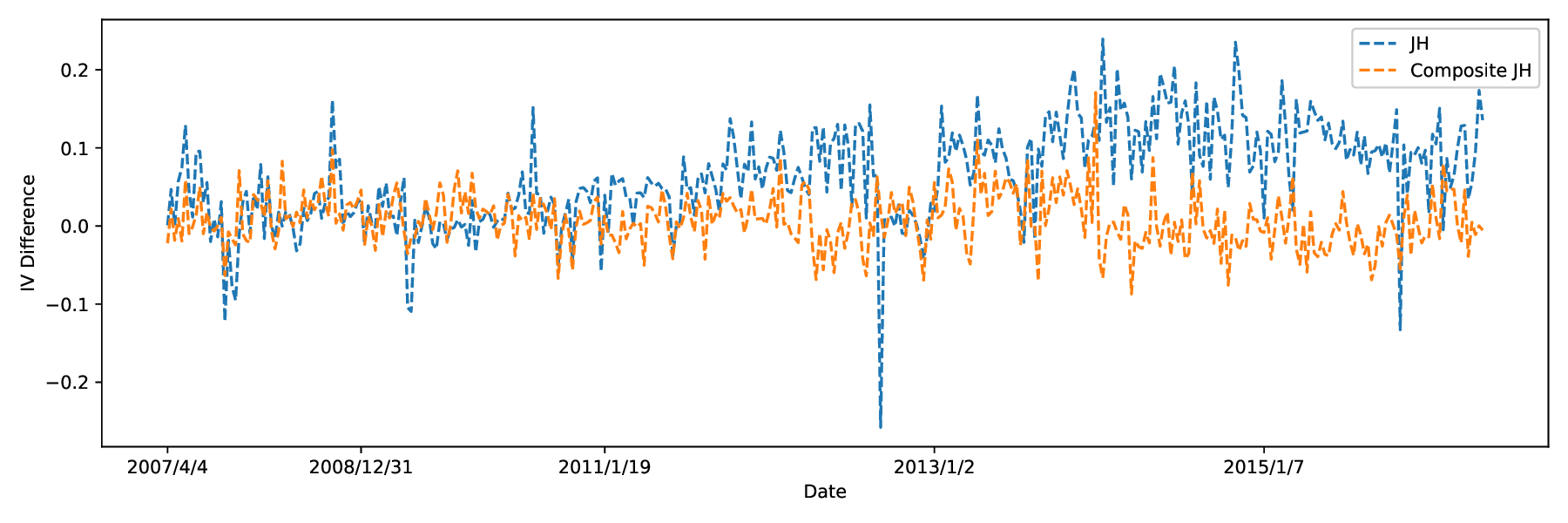}
}
\end{minipage}
\vspace{-2pt}
\footnotesize \textit{Note. } The plots compare the short-term ATM IV, defined by Eq. \eqref{IV}, between the models (JH and Composite JH) and the market. The first row plots the ATM IV of JH model, the second row plots that of Composite JH model, and the last row compares the corresponding difference (model - market) between models and the market. Maturities of SPX options smaller than 30 days are considered short-term.\label{JH IV VIX}
\end{figure}

Figure \ref{JH IV VIX Skew} further plots the short-term volatility skew fits in VIX option market for the two models. We also witness a consistently and strongly better performance of Composite JH. In addition, while JH model occasionally generates irregular calibrated models (e.g., the peak values in the difference plot), Composite JH calibrates much more robustly.

\begin{figure}[!ht]
\caption{The short-term volatility skew of JH and Composite JH in VIX markets}
 \begin{minipage}{\linewidth}
 \subfloat[][Short-term volatility skew]{
 \includegraphics[width=\linewidth]{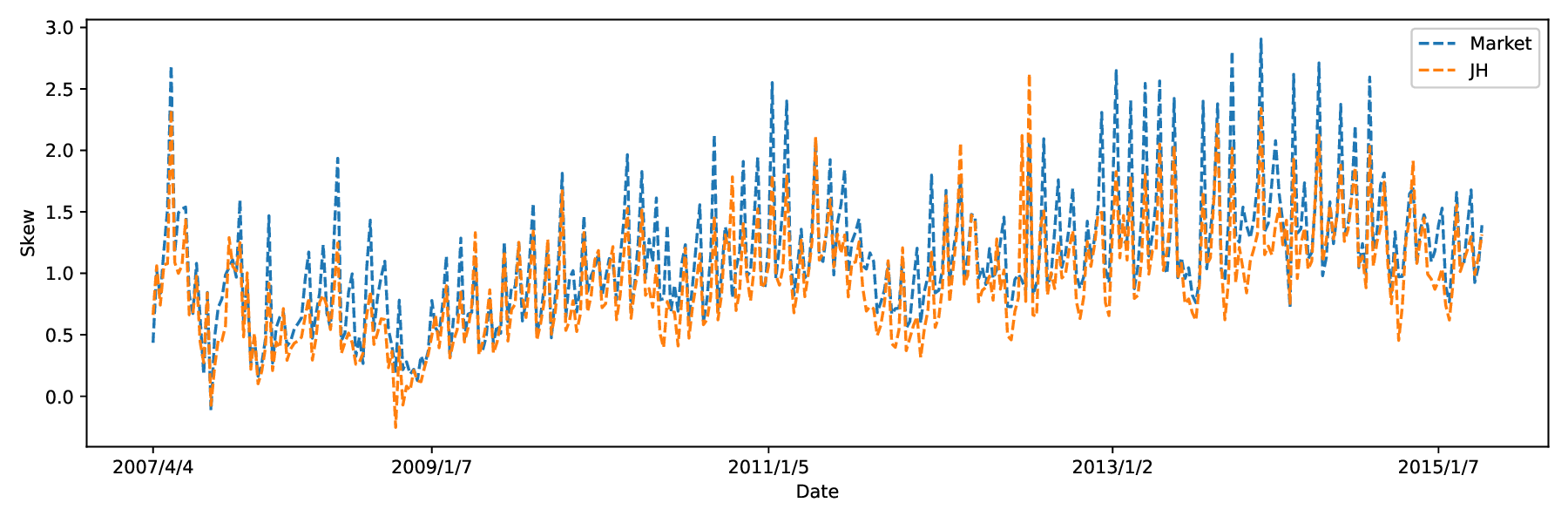}
 }
 \end{minipage} \vspace{-1em}
  \begin{minipage}{\linewidth}
 \subfloat[][Short-term volatility skew]{
 \includegraphics[width=\linewidth]{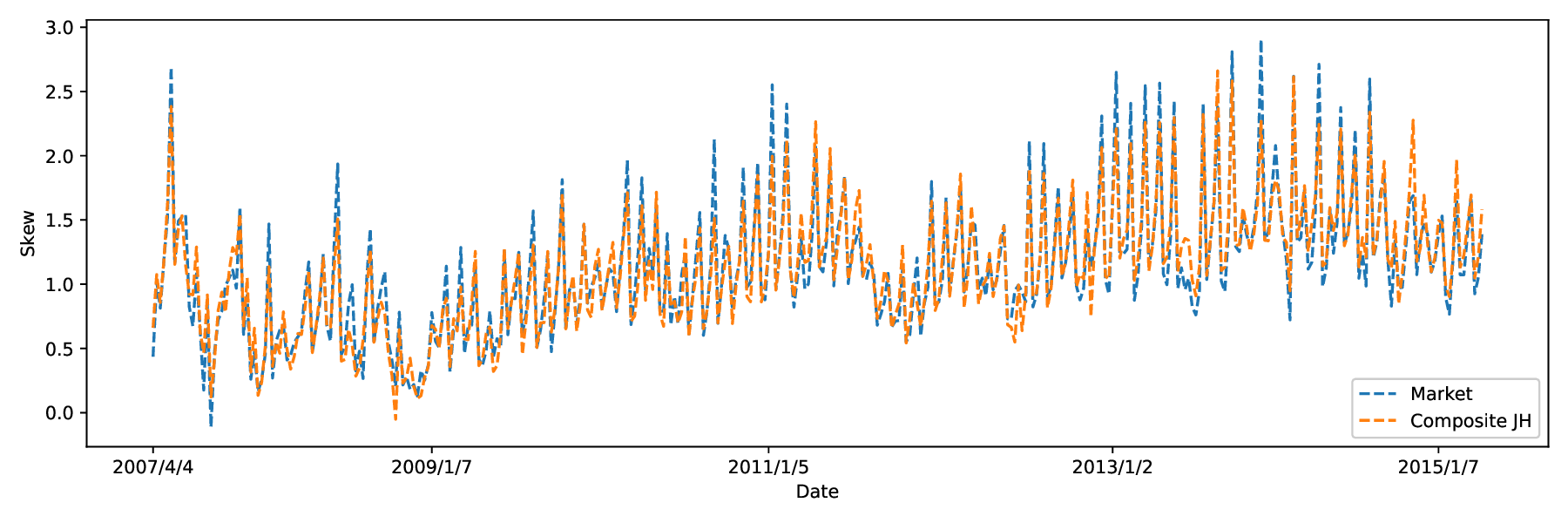}
}
 \end{minipage} \vspace{-1em}
 \begin{minipage}{\textwidth}
\subfloat[][Volatility skew difference]{
\includegraphics[width=\textwidth]{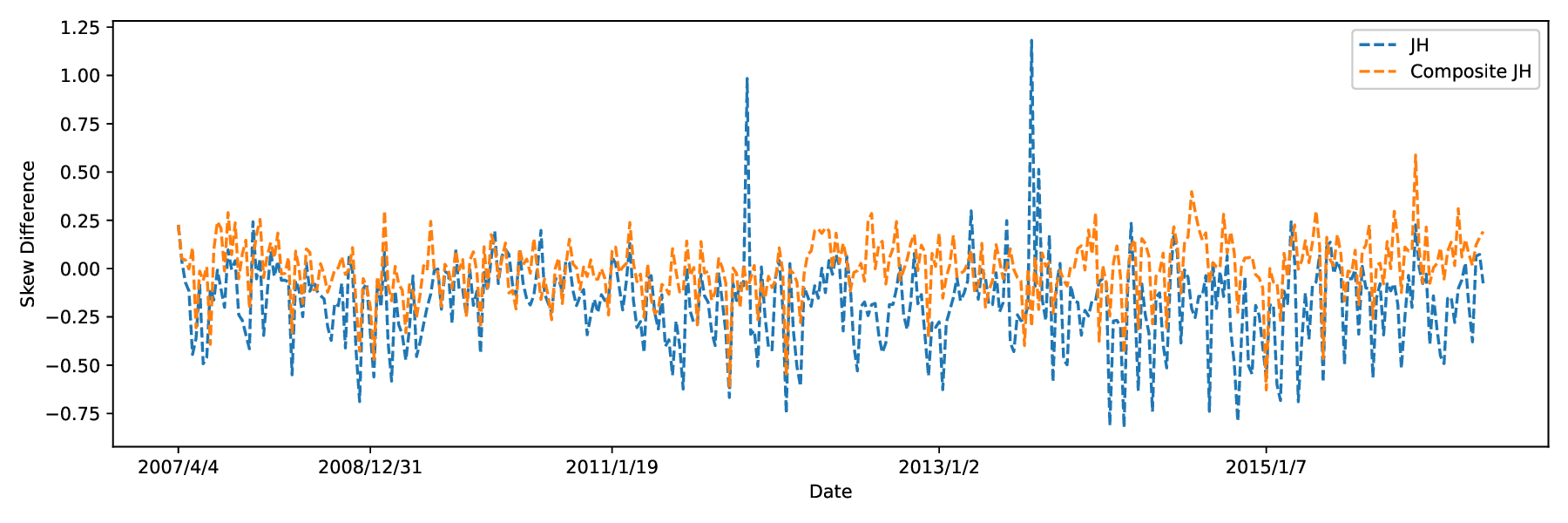}
}
\end{minipage}
\vspace{-2pt}
\footnotesize \textit{Note. } The plots compare the short-term volatility skew, defined by Eq. \eqref{Skew}, between the models (JH and Composite JH) and the market. The first row plots the volatility skew of JH model, the second row plots that of Composite JH model, and the last row compares the corresponding difference (model - market) between models and the market. Maturities of SPX options smaller than 30 days are considered short-term.\label{JH IV VIX Skew}
\end{figure}

Figure \ref{JH IV VIX long term} plots the long-term daily ATM IV fits in the VIX option market. We observe similar patterns of performance as in the short-term scenarios, where Composite JH model performs consistently better and more robustly.

\begin{figure}[!ht]
\caption{The long-term ATM implied volatility of JH and Composite JH in VIX markets}
 \begin{minipage}{\linewidth}
 \subfloat[][Long-term ATM IV]{
 \includegraphics[width=\linewidth]{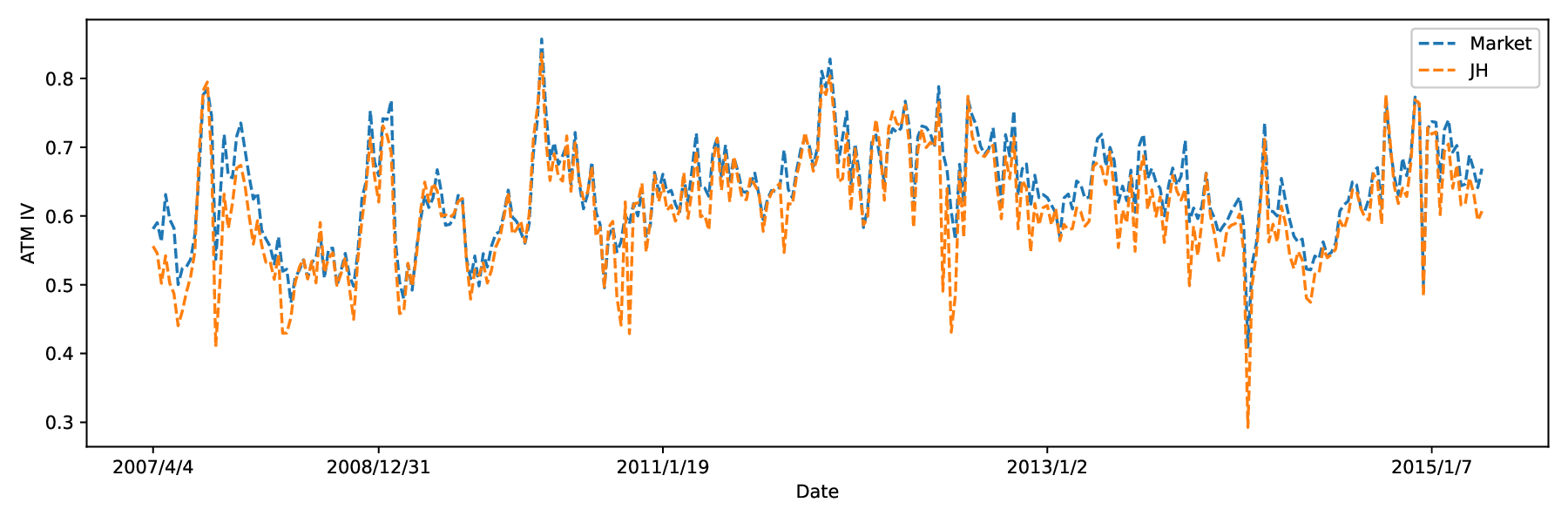}
 }
 \end{minipage} \vspace{-1em}
  \begin{minipage}{\linewidth}
 \subfloat[][Long-term ATM IV]{
 \includegraphics[width=\linewidth]{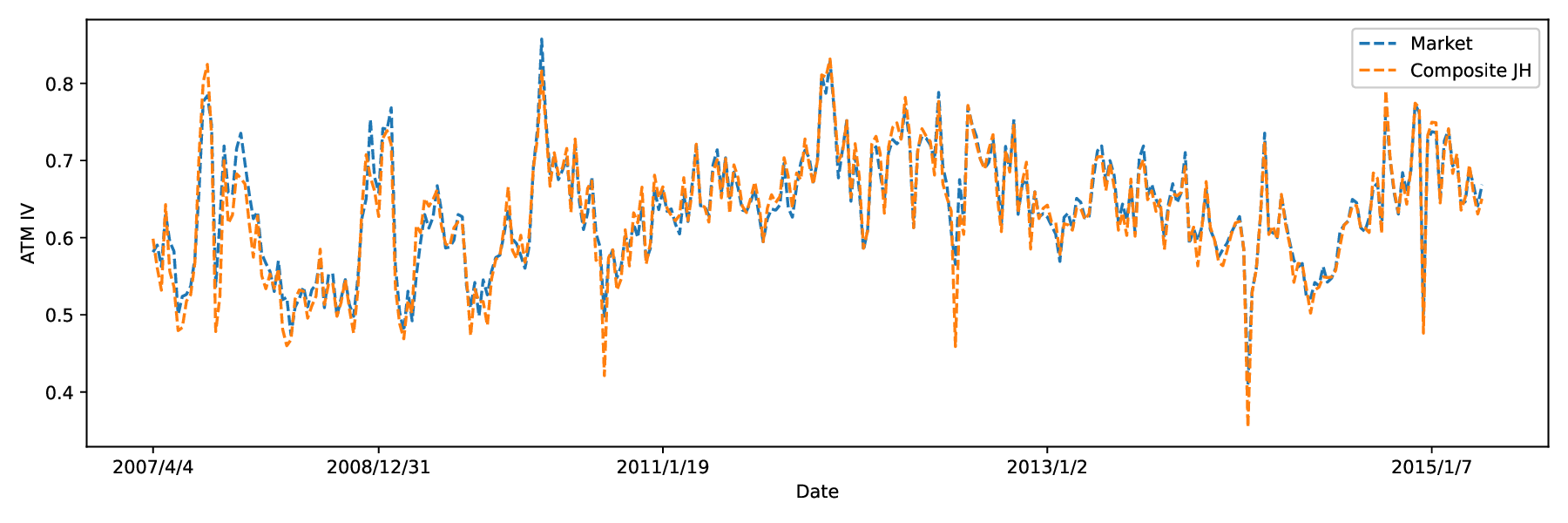}
}
 \end{minipage} \vspace{-1em}
 \begin{minipage}{\textwidth}
\subfloat[][ATM IV difference]{
\includegraphics[width=\textwidth]{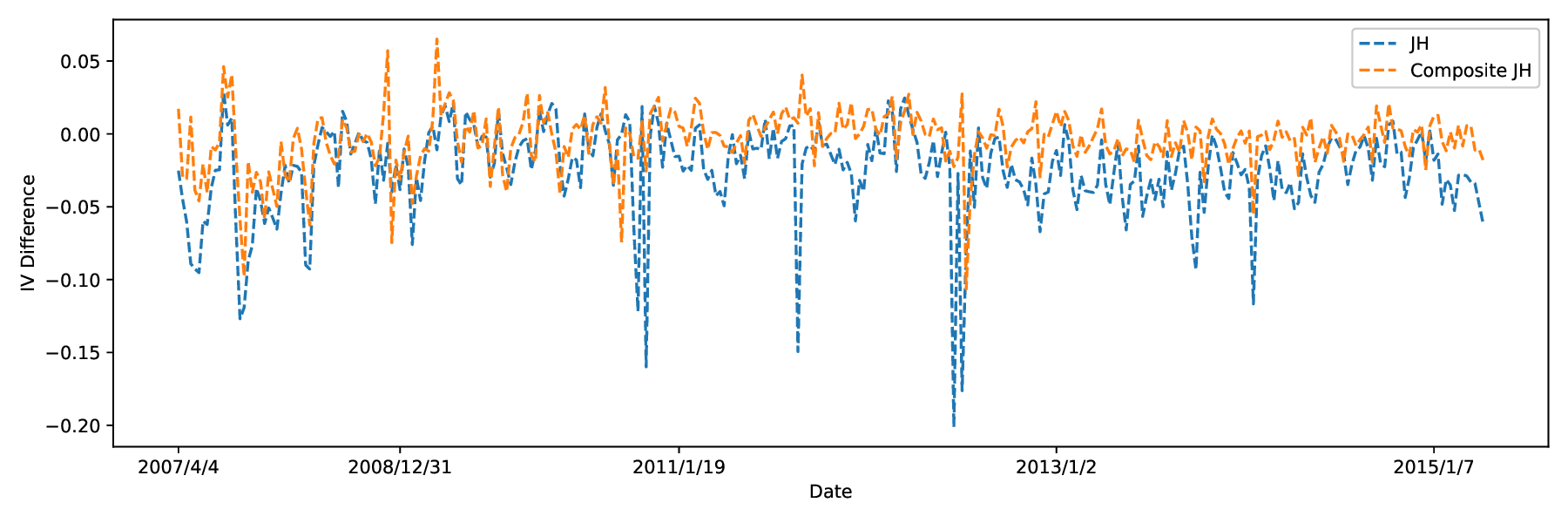}
}
\end{minipage}
\vspace{-2pt}
\footnotesize \textit{Note. } The plots compare the long-term ATM IV, defined by Eq. \eqref{IV}, between the models (JH and Composite JH) and the market. The first row plots the ATM IV of JH model, the second row plots that of Composite JH model, and the last row compares the corresponding difference (model - market) between models and the market. Maturities of SPX options smaller than 270 days are considered short-term.\label{JH IV VIX long term}
\end{figure}

And lastly, Figure \ref{JH IV SPX skew long term} shows the long-term daily fits of volatility skew for the two models in VIX option market, where Composite JH model is again much more robust. Such outstanding overall performance of Composite JH in VIX option markets is also confirmed in Table \ref{characteristics}.

Combining the results above, we find superior and balanced calibration performance of Composite JH model in the joint market.
\begin{figure}[!ht]
\caption{The long-term volatility skew of JH and Composite JH in VIX markets}
 \begin{minipage}{\linewidth}
 \subfloat[][Long-term volatility skew]{
 \includegraphics[width=\linewidth]{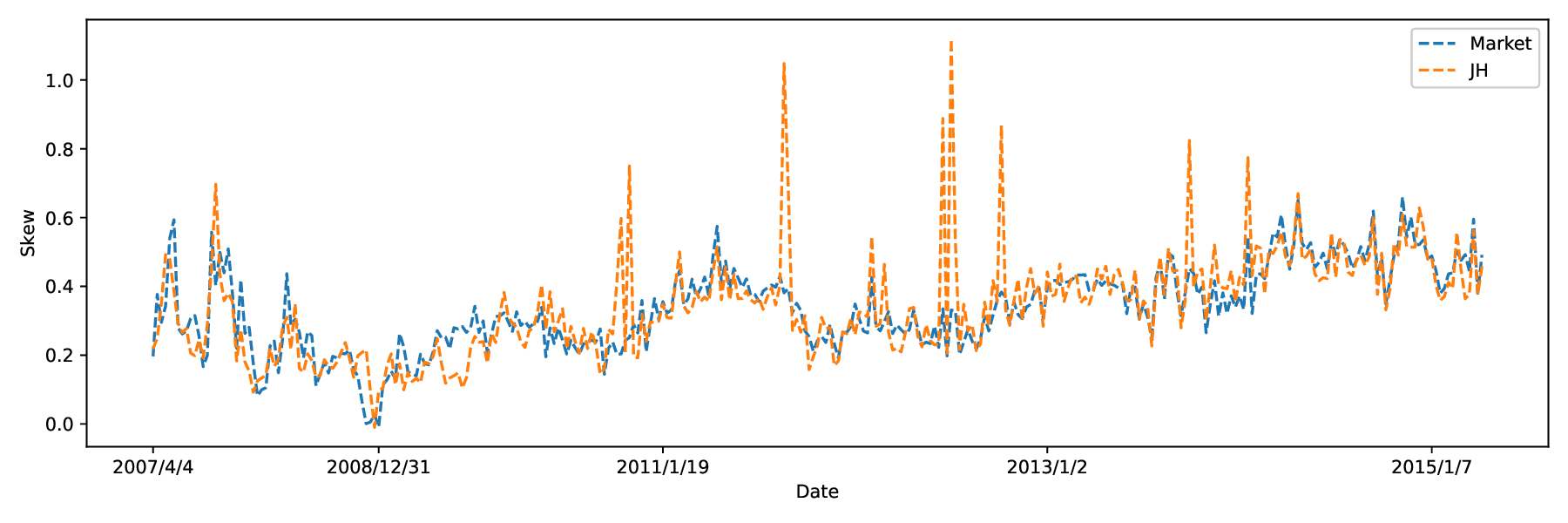}
 }
 \end{minipage} \vspace{-1em}
  \begin{minipage}{\linewidth}
 \subfloat[][Long-term volatility skew]{
 \includegraphics[width=\linewidth]{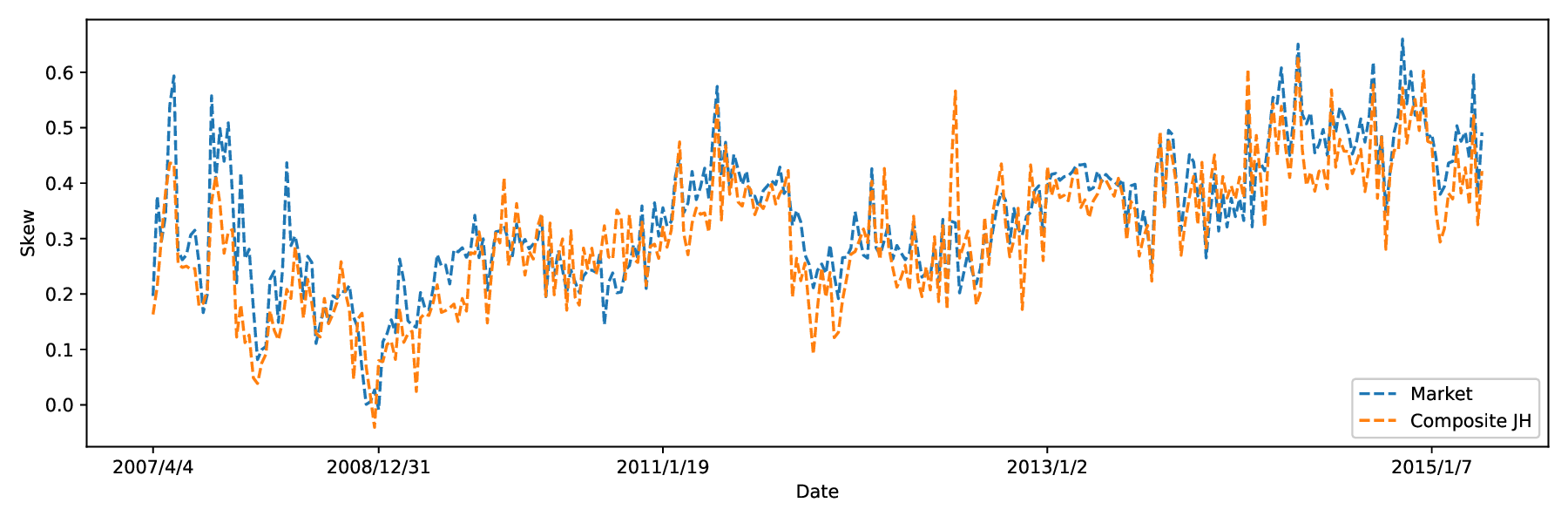}
}
 \end{minipage} \vspace{-1em}
 \begin{minipage}{\textwidth}
\subfloat[][Volatility skew difference]{
\includegraphics[width=\textwidth]{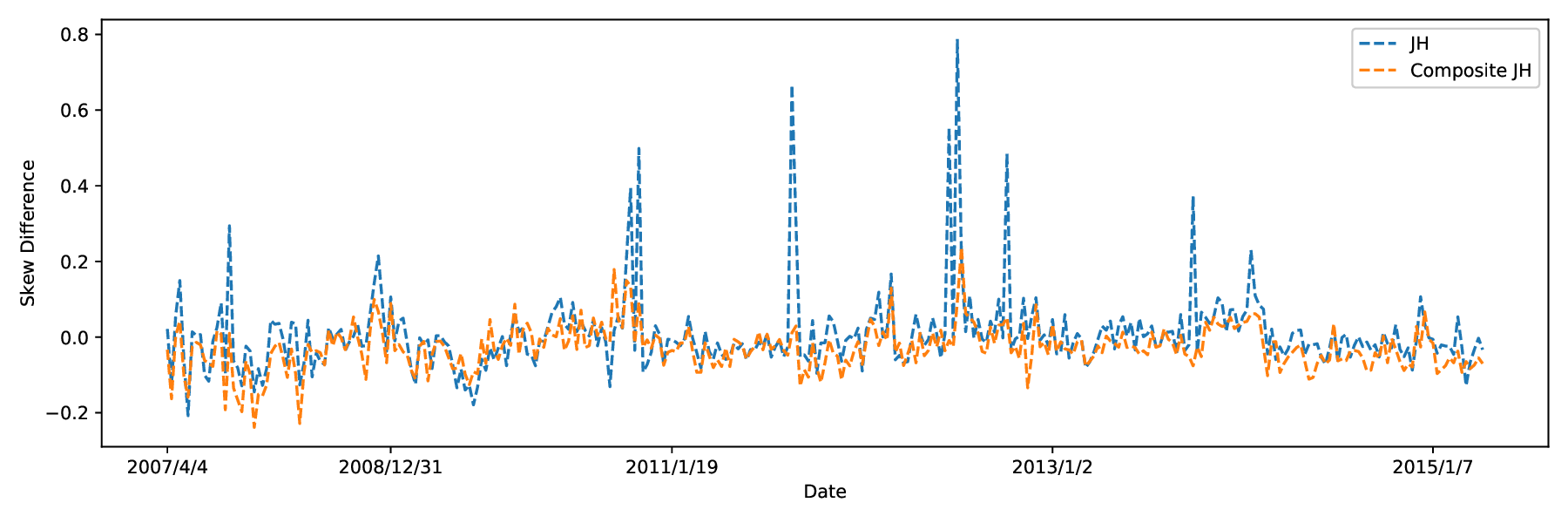}
}
\end{minipage}
\vspace{-2pt}
\footnotesize \textit{Note. } The plots compare the long-term volatility skew, defined by Eq. \eqref{Skew}, between the models (JH and Composite JH) and the market. The first row plots the volatility skew of JH model, the second row plots that of Composite JH model, and the last row compares the corresponding difference (model - market) between models and the market. Maturities of SPX options smaller than 270 days are considered short-term.\label{JH IV VIX skew long term}
\end{figure}

Next, we consider how the models fit the ATM skew in both markets. The decay rate of ATM skew is widely concerned in option markets. We characterize the goodness of the fit in short-term ATM skew by computing the RMSRE of near-the-money implied volatilities for all short maturities. We denote by $\sigma_{i}(k, \tau), \;i \in\{\text{mkt}, \text{mdl}\}$ the (market/model) implied volatility at maturity $\tau$ and moneyness $k$. The daily RMSRE is defined as 
\begin{equation}
\label{D RMSRE}\text{RMSRE}_t = \frac{1}{|\Gamma_t|} \sum_{\tau \in \Gamma_t }\left(\sqrt{\sum_{\exp{(k)}\in \tilde K}\left(\frac{\sigma_{\text{mdl}}(k, \tau) - \sigma_{\text{mkt}}(k, \tau)}{\sigma_{\text{mkt}}(k, \tau)}\right)^2}\right),
\end{equation}
where $\Gamma_t$ is the set of all maturities on date $t$ and $\tilde K$ is the moneyness range used to compute the RMSRE. We set $\tilde K = [0.9, 1.1]$ for the SPX market and $\tilde K = [0.8, 1.2]$ for the VIX market, a larger range due to the sparsity of moneyness in VIX markets.

Table \ref{decay table} and Figure \ref{decay plot} jointly assess the in-sample performance of composite models in fitting short-term near-the-money implied volatilities.

From Table \ref{decay table}, composite models consistently outperform their ordinary counterparts, with far more significant improvements in VIX markets. For SPX options: Composite Heston reduces the average daily RMSRE by 21.2\% (from 0.1044 to 0.0823), while Composite JH cuts it by 19.3\% (from 0.0658 to 0.0531). The pairwise t-statistics (4.74 for Composite Heston vs. Heston, 9.36 for Composite JH vs. JH) confirm these improvements are statistically significant.
In VIX markets, the advantage of composite models is even more striking. Composite Heston slashes the average daily RMSRE by 50.0\% (from 0.1712 to 0.0856), and Composite JH achieves a 45.2\% reduction (from 0.0620 to 0.0340), nearly halving the error. The $t$-statistics here are much larger (14.71 and 12.18, respectively), indicating the outperformance of composite models in VIX is highly statistically reliable.

Figure \ref{decay plot} provides intuitive visual evidence, focusing on the top-performing JH and Composite JH. Subfigure (a) (SPX) shows Composite JH’s daily RMSRE curve lies consistently below JH’s across the in-sample period, with smaller day-to-day fluctuations. Subfigure (b) (VIX) amplifies this gap: Composite JH’s RMSRE is drastically lower than JH’s, and the difference widens over time, especially after 2011, as noted. Subfigure (c) (RMSRE difference: JH minus Composite JH) further confirms this: nearly all daily differences are positive, meaning Composite JH outperforms JH on almost every calibrated date, with larger positive values (i.e., bigger performance gains) in VIX markets post-2011.

\begin{table}[!ht]
\caption{The in-sample average of daily RMSRE (Eq. \eqref{D RMSRE}) of short-term near-the-money implied volatilities in SPX and VIX option markets\label{decay table}}
\begin{center}
    \begin{tabular}{ccccc}
         & Heston & Composite Heston & JH & Composite JH \\\hline
       Error (SPX) & 0.1044 & 0.0823 & 0.0658 & 0.0531\\
        Error (VIX) & 0.1712 & 0.0856 & 0.0620 & 0.0340\\
        \makecell[c]{Pairwise $t$-statistics\\(SPX)} & - &4.74 & - & 9.36 \\
        \makecell[c]{Pairwise $t$-statistics\\(VIX)} & - & 14.71 & - & 12.18 \\\hline
    \end{tabular}
    \end{center}
    {\textit{Note.} This table reports the in-sample average of daily RMSRE (defined in Equation \eqref{D RMSRE}) for short-term near-the-money implied volatilities, comparing four models across SPX and VIX option markets. For SPX options, we select options with maturities smaller than 90 days and moneyness within $[0.9, 1.1]$. For VIX options, we select options with maturities smaller than 45 days and moneyness within $[0.8, 1.2]$. The pairwise $t$-statistics, calculated via Equation \eqref{t}, compare the RMSRE of each composite model to its corresponding ordinary counterpart (e.g., Composite Heston vs. Heston, Composite JH vs. JH). A significantly positive $t$-statistic confirms the composite model has a lower RMSRE (i.e., better performance) than the ordinary model.}
\end{table}

\begin{figure}
\caption{Comparison of daily RMSRE of short-term near-the-money implied volatilities between JH and Composite JH in joint markets\label{decay plot}}
 \begin{minipage}{\linewidth}
 \centering
\subfloat[][Daily RMSRE in SPX option markets.]{ \includegraphics[width=\textwidth]{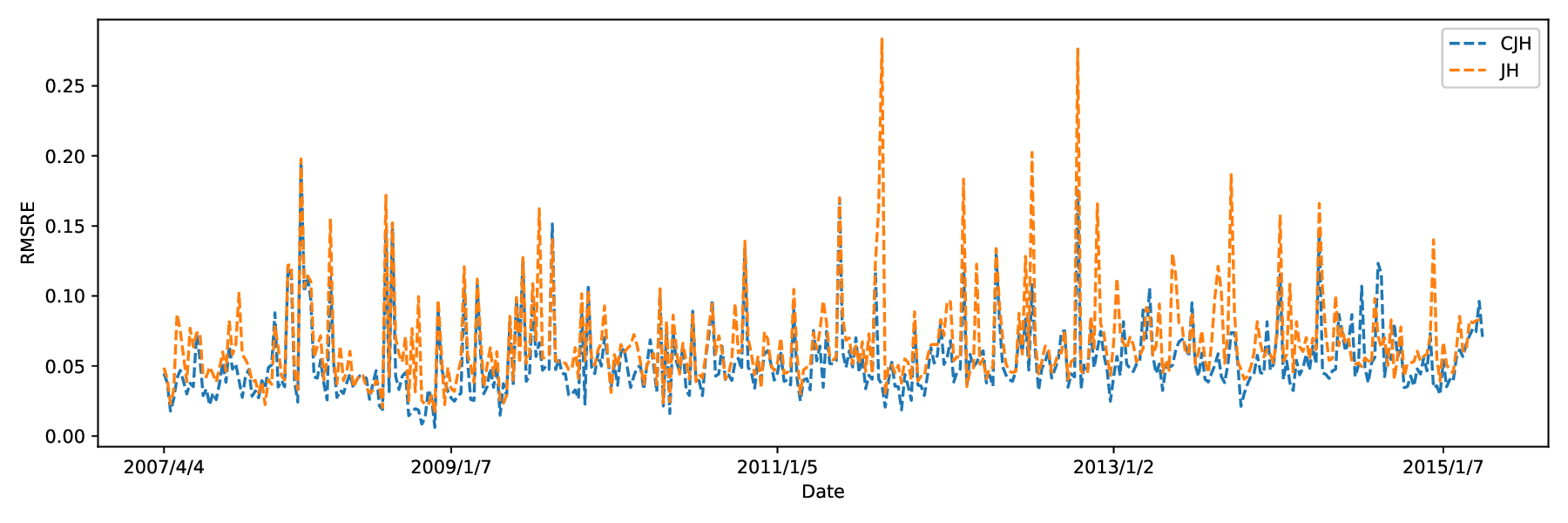}}
 \end{minipage}
 \begin{minipage}{\linewidth}
\centering\subfloat[][Daily RMSRE in VIX option market.]{ 
\includegraphics[width=\textwidth]{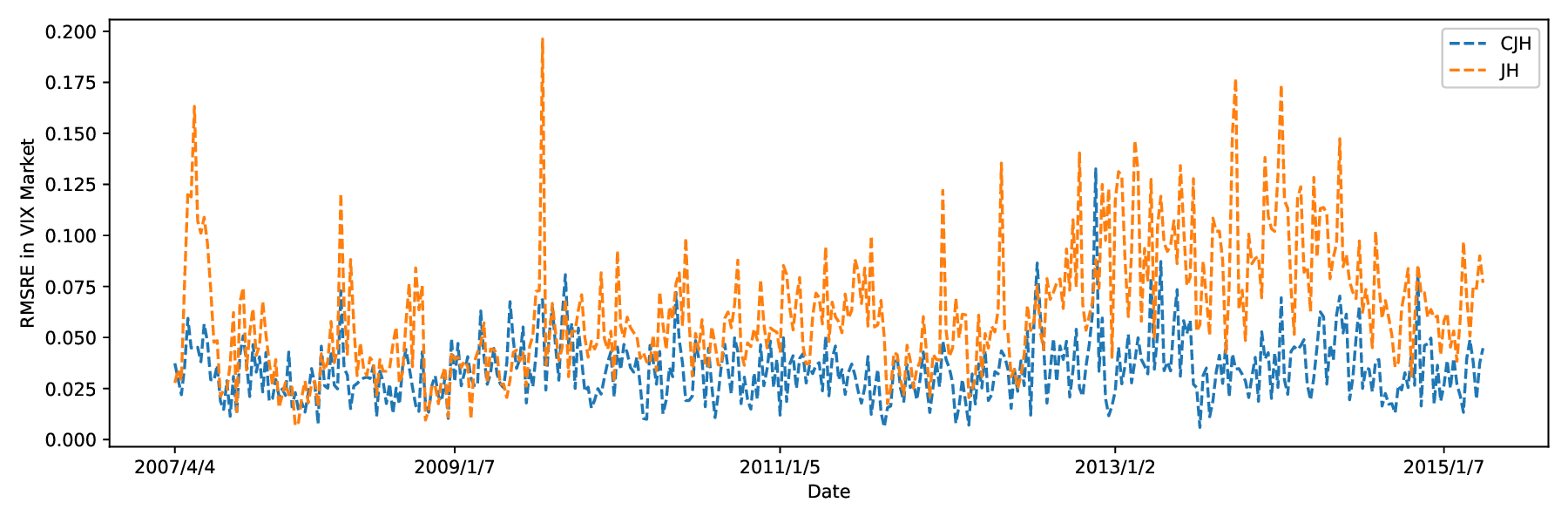}}
\end{minipage}
 \begin{minipage}{\linewidth}
\centering\subfloat[][Difference of RMSRE.]{ 
\includegraphics[width=\textwidth]{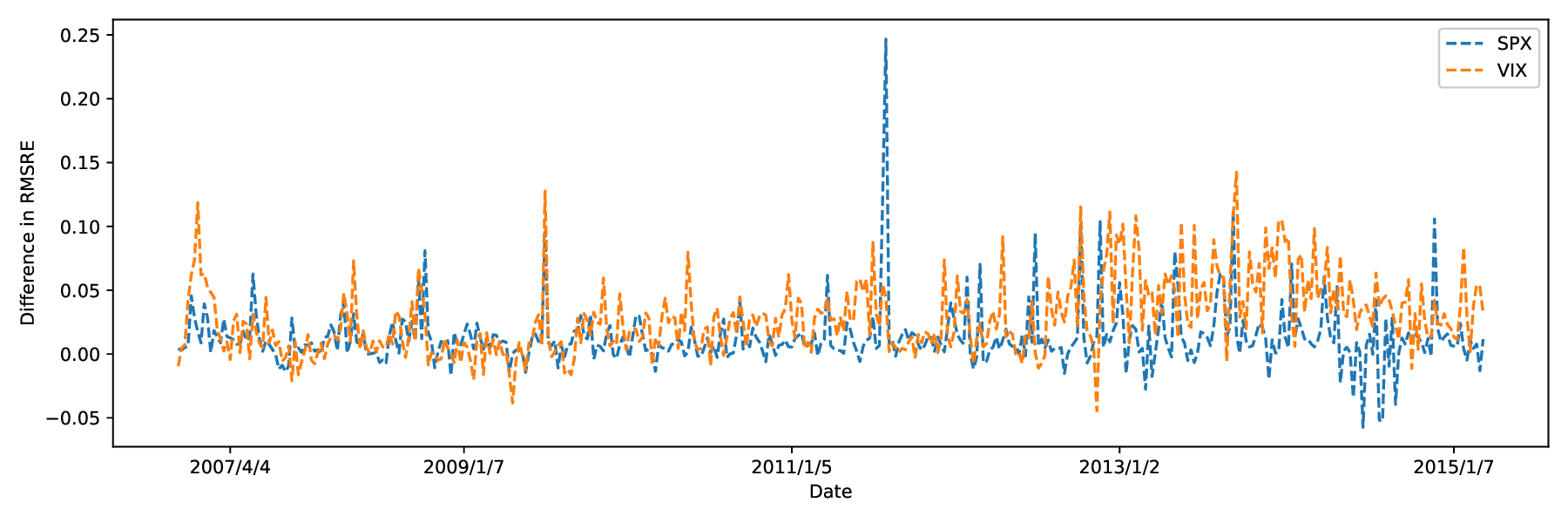}}
\end{minipage}
\footnotesize \textit{Note. } The plot compares the daily RMSRE (defined by Eq. \eqref{D RMSRE}) of short-term near-the-money implied volatilities between JH and Composite JH in SPX and VIX markets. For SPX options, we select options with maturities smaller than 90 days and moneyness within $[0.9, 1.1]$. For VIX options, we select options with maturities smaller than 45 days and moneyness within $[0.8, 1.2]$. In the third row, the difference of RMSRE is defined by the daily RMSRE of JH - Composite JH.
\end{figure}

Together, the table and figure demonstrate that the composite structure is particularly effective at fitting short-term near-the-money implied volatilities, especially in VIX markets, where ordinary models struggle most with their ordinary time change dynamics.
\section{Conclusion}
This paper has explored the application of CTC models to the joint pricing of SPX and VIX derivatives, providing a comprehensive solution to the consistent modeling problem across these interconnected markets. The theoretical framework introduces a novel decomposition of volatility and volatility-of-volatility through the composite structure $X_t = L_{U_{V_t}}$, where the theoretical decompositions $$\begin{aligned}
  &  \operatorname{VIX}_t^2 =-2E[L_1] u_{V_t}v_t + O(\bar{\tau}),\\
  & \operatorname{VVIX}_t^2 =2r- \mathcal{A}^v \ln (v_t) - \mathcal{A}^u \ln (u_{V_t}) v_t + o(1).
\end{aligned}$$
enable simultaneous calibration to both SPX and VIX implied volatility surfaces. The model specification naturally incorporates stochastic volatility through time changes $U$ and $V$, leverage effects via correlated Brownian motions, and jump components through L\'{e}vy processes, while subsuming many classical models as special cases.

The numerical implementation combines two efficient approaches for derivative pricing. For European-style options, we developed the CTC-COS method, extending the Fourier-cosine expansion technique to handle composite time changes. This method achieves $O(ND)$ computational complexity through decoupling composite time change in the numerical approximation, where $N$ represents the number of cosine terms and $D$ the discretization level. The pricing formula also involves separation of the strike dependence and efficient pricing along maturities, enabling fast calibration across the entire volatility surface. For VIX derivatives, we derived an explicit VIX-spot relationship under the affine assumptions. Then we introduced an exact simulation scheme that combines gamma expansion for conditional distributions with characteristic function inversion for $u_{V_{\bar T}}$. This approach avoids path simulation by directly sampling terminal distributions.

Empirical results demonstrate the model's superior performance in joint calibration exercises using real-market SPX and VIX options data. The composite structure improved in-sample calibration errors by 32.62\% for the Heston model and 23.00\% for the Jump Heston model, while maintaining robust out-of-sample performance. The Composite JH model achieved particularly balanced results, with short-term ATM implied volatility errors of 0.0203 for SPX options versus 0.0338 for VIX options, which outperformed existing multi-factor benchmarks.

These results establish the CTC framework as a powerful paradigm for joint derivatives pricing, combining theoretical elegance with practical computational efficiency. The methodology's success suggests several promising extensions, including applications to multi-asset derivatives, incorporation of rough volatility features, and potential acceleration through machine learning techniques for high-dimensional parameter spaces. The demonstrated ability to decouple market smiles while maintaining tractability makes this approach particularly valuable for both pricing and risk management applications in complex derivatives markets.
\newpage
\appendix

\section{Representation of VIX and VVIX under Composite Time Changes}
\label{re_VIX}
Denote by $E_t[\cdot] = \mathbb{E}[\cdot \mid \mathcal{F}_t^X]$ the conditional expectation under $\mathbb{P}$ given $\mathcal{F}_t^X$. We further impose the following assumptions for time changes $U$ and $V$:

\begin{assumption}\label{assump:UV}
    Time changes $U$, $V$ are absolutely continuous with $\mathrm{d}U_t = u_t\mathrm{d}t$ and $\mathrm{d} V_t = v_t\mathrm{d}t$.
\end{assumption}

\begin{assumption}\label{assump:regularity}
$u$ and $v$ are independent Feller processes and admit infinitesimal generators $\mathcal{A}^u$ and $\mathcal{A}^v$, respectively. Both are sufficiently regular in the sense that: for any continuous functions $f, g$, there exists a constant $\varepsilon > 0$ such that the family $\{f(v_s)g(u_{V_s})\}_{s\in [t, t+\varepsilon]}$ is uniformly integrable under $E_t[\cdot]$ for all $t\ge 0$.
\end{assumption}

\begin{assumption}\label{assump:generator}
For $f(x) = I(x) := x$ and $f(x) = \ln x$, $$E^\mathbb{P}[f(u_{t+\varepsilon}) - f(u_t)\mid\mathcal{F}^X_t] = \mathcal{A}^u f(u_t)\varepsilon + o(\varepsilon), \quad E^\mathbb{P}[f(v_{t+\varepsilon}) - f(v_t)\mid\mathcal{F}^X_t] = \mathcal{A}^vf(v_t) \varepsilon + o(\varepsilon),$$
almost surely with continuous $\mathcal{A}^u f$ and $\mathcal{A}^vf$.
\end{assumption}

Since Feller processes $u$ and $v$ are stochastically continuous, as a result of Assumption \ref{assump:regularity}, we have from probability theory that $E_t|f(v_{t+\bar\tau})g(u_{V_{t+\bar\tau}}) - f(v_{t+\bar\tau})g(u_{V_{t+\bar\tau}})| = o(1)$ almost surely for continuous $f$, $g$.
\begin{proposition}\label{prop:VIX}
Under Assumptions \ref{assump:UV}-\ref{assump:generator}, viewing $\operatorname{VIX}$ as a function of $\bar \tau$, the following approximation holds as $\bar{\tau} \to 0$:
$$
\operatorname{VIX}_t^2 = -2E[L_1] u_{V_t} v_t + O(\bar{\tau}).
$$
\end{proposition}

\begin{proof}{Proof}
By Assumption \ref{assump:generator}, we first compute the first conditional moment of $V_{t+\varepsilon} - V_t$:
\begin{align}
E_t[V_{t+\varepsilon} - V_t] &= \int_0^{\varepsilon} E_t[v_{t+s}]\mathrm{d}s \notag \\
&= v_t\varepsilon + \int_0^\varepsilon E_t[v_{t+s} - v_t]\mathrm{d}s \notag \\
&= v_t \varepsilon + \frac{\varepsilon^2}{2}\mathcal{A}^v I(v_t) + o(\varepsilon^2). \label{eq:V_first_moment}
\end{align}

For the second conditional moment, using Assumption \ref{assump:regularity}, we have:
\begin{align}
E_t[(V_{t+\varepsilon} - V_t)^2] &= 2\int_0^\varepsilon \int_0^s E_t[v_{t+s}v_{t+r}]\mathrm{d}r\mathrm{d}s \notag \\
&= 2\int_0^\varepsilon \int_0^s \left[E_t[v_{t+r}(v_{t+s} - v_{t+r})] + E_t[v_{t+r}^2]\right]\mathrm{d}r\mathrm{d}s \notag \\
&= 2\int_0^\varepsilon \int_0^s \left[E_t[v_{t+r}\mathcal{A}^v I(v_{t+r})](s-r) + o(1) + v_t^2 + o(1)\right]\mathrm{d}r\mathrm{d}s \notag \\
&= v_t^2 \varepsilon^2 + o(\varepsilon^2), \label{eq:V_second_moment}
\end{align}
where the third equality follows from Assumption \ref{assump:generator} (continuity of $\mathcal{A}^v I$) and continuity of $f(x)=x^2$.

By the independence of $u$, $v$ and Assumption \ref{assump:generator} w.r.t. $I(x)$, we condition on $\{V_r\}_{r\ge 0}$ and $\mathcal{F}_t^X$:
\begin{align}
E_t\left[U_{V_{t+\bar{\tau}}} - U_{V_t} \mid \mathcal{F}_t^X, \{V_r\}_{r\ge 0}\right] &= u_{V_t}(V_{t+\bar{\tau}} - V_t) + \frac{1}{2}\mathcal{A}^u I(u_{V_t})(V_{t+\bar{\tau}} - V_t)^2 \notag \\
&\quad + o((V_{t+\bar{\tau}} - V_t)^2). \label{eq:U_expansion}
\end{align}

Moreover, since $E[U_{V_t}] < \infty$ for every $t\ge 0$, $X_t - E[L_1] U_{V_t}$ is a zero-mean martingale (see the argument in proof of Proposition \ref{affine spec}). Substitute this into the VIX formula:
\begin{align}
\operatorname{VIX}_t^2 &= -\frac{2}{\bar{\tau}} E_t[X_{t+ \bar\tau} - X_t] \notag\\
&= -\frac{2}{\bar{\tau}} \mathbb{E}[L_1] E_t[U_{V_{t+\bar{\tau}}} - U_{V_t}] \notag \\
&= -\frac{2\mathbb{E}[L_1]}{\bar{\tau}} E_t\left[E_t\left[U_{V_{t+\bar{\tau}}} - U_{V_t} \mid \mathcal{F}_t^X, \{V_r\}_{r\ge 0}\right]\right] \notag \\
&= -\frac{2\mathbb{E}[L_1]}{\bar{\tau}} E_t\left[ u_{V_t}(V_{t+\bar{\tau}} - V_t) + \frac{1}{2}\mathcal{A}^u I(u_{V_t})(V_{t+\bar{\tau}} - V_t)^2 + o((V_{t+\bar{\tau}} - V_t)^2)\right] \notag \\
&= -\frac{2\mathbb{E}[L_1]}{\bar{\tau}} \left[ u_{V_t}E_t[V_{t+\bar{\tau}} - V_t] + \frac{1}{2}\mathcal{A}^u I(u_{V_t})E_t[(V_{t+\bar{\tau}} - V_t)^2] + o(\bar{\tau}^2)\right]. \label{eq:VIX_intermediate}
\end{align}
Substitute \eqref{eq:V_first_moment} and \eqref{eq:V_second_moment} (with $\varepsilon = \bar{\tau}$) into \eqref{eq:VIX_intermediate}:
\begin{align}
\operatorname{VIX}_t^2 &= -\frac{2\mathbb{E}[L_1]}{\bar{\tau}} \left[ u_{V_t}\left(v_t \bar{\tau} + \frac{\bar{\tau}^2}{2}\mathcal{A}^v I(v_t) + o(\bar{\tau}^2)\right) + \frac{1}{2}\mathcal{A}^u I(u_{V_t})\left(v_t^2 \bar{\tau}^2 + o(\bar{\tau}^2)\right) + o(\bar{\tau}^2)\right] \notag \\
&= -2\mathbb{E}[L_1] u_{V_t} v_t - \mathbb{E}[L_1] \bar{\tau}\left(u_{V_t}\mathcal{A}^v I(v_t) + v_t^2 \mathcal{A}^u I(u_{V_t})\right) + o(\bar{\tau}) \label{detailed VIX} \\
&= -2\mathbb{E}[L_1] u_{V_t} v_t + O(\bar{\tau}).\notag
\end{align}
\end{proof}

\begin{proposition}\label{prop:VVIX}
Under Assumptions \ref{assump:UV}-\ref{assump:generator}, viewing $\operatorname{VVIX}$ as a function of $\bar \tau$, the following approximation holds as $\bar{\tau} \to 0$:
$$
\operatorname{VVIX}_t^2 = 2r - \mathcal{A}^v \ln(v_t) - \mathcal{A}^u \ln(u_{V_t}) v_t + o(1).
$$
\end{proposition}

\begin{proof}{Proof}
Note from the continuity of $\mathcal{A}^u I$ and $\mathcal{A}^v I$ that \begin{equation}
\label{eq:o1}
E_t[u_{V_{t+\bar\tau}}\mathcal{A}^vI(v_{t+\bar\tau}) + v_{t+\bar\tau}^2 \mathcal{A}^uI(u_{V_{t+\bar\tau}}) - u_{V_t}\mathcal{A}^vI(v_t) - v_t^2 \mathcal{A}^uI(u_{V_t})] = o(1).
\end{equation}
Then from Eq. \eqref{detailed VIX} and Eq. \eqref{eq:o1}:
$$
E_t[\operatorname{VIX}_{t+\bar{\tau}}^2 - \operatorname{VIX}_t^2] = -2\mathbb{E}[L_1] E_t\left[ u_{V_{t+\bar{\tau}}}v_{t+\bar{\tau}} - u_{V_t}v_t\right] + o(\bar{\tau}).
$$

Thus, \begin{align}
\label{ln VIX}
E_t\left[\ln \frac{\operatorname{VIX}_{t+\bar \tau}^2}{\operatorname{VIX}_t^2}\right] &= E_t\left[\ln \left(1 + \frac{\operatorname{VIX}_{t+\bar\tau}^2 - \operatorname{VIX}_t^2}{\operatorname{VIX}_t^2}\right)\right] \notag\\
&=E_t\left[\ln\left(1 + \frac{u_{V_{t+\tau}}v_{t+\bar\tau} - u_{V_t}v_t + o(\bar\tau)}{\operatorname{VIX}_t^2}\right)\right]\notag\\
&= E_t\left[\ln\left(1 + \frac{u_{V_{t+\tau}}v_{t+\bar\tau} - u_{V_t}v_t}{u_{V_t}v_t + O(\bar\tau)}\right)\right] + o(\bar\tau)\notag\\
& = E_t\left[\ln\left(1 + \frac{u_{V_{t+\tau}}v_{t+\bar\tau} - u_{V_t}v_t}{u_{V_t}v_t}\right)\right] - O(\bar\tau)E_t\left[\frac{u_{V_{t+\bar\tau}}v_{t+\bar\tau} - u_{V_t}v_t}{(u_{V_t}v_t)^2}\right] + o(\bar\tau)\notag\\
&= E_t\left[\ln\left(\frac{u_{V_{t+\tau}}v_{t+\bar\tau} }{u_{V_t}v_t}\right)\right] + o(\bar\tau).\end{align}

Substitute \eqref{ln VIX} into the VVIX formula and apply Assumption \ref{assump:generator} w.r.t. $f(x) = \ln (x)$:
\begin{align}
\operatorname{VVIX}_t^2 &= -\frac{2}{\bar{\tau}} E_t\left[  -r\bar{\tau} + \frac{1}{2}\ln\frac{\operatorname{VIX}_{t+\bar\tau}^2}{\operatorname{VIX}_{t}^2}\right] \notag\\
&= -\frac{2}{\bar{\tau}} E_t\left[ -r\bar{\tau} + \frac{1}{2}\left(\ln \frac{v_{t+\bar{\tau}}}{v_t} + \ln \frac{u_{V_{t+\bar{\tau}}}}{u_{V_t}}\right)\right]  + o(1) \notag \\
&= -\frac{2}{\bar{\tau}} \left[ -r\bar{\tau} + \frac{1}{2}\left(\mathcal{A}^v \ln(v_t) \bar{\tau} + \mathcal{A}^u \ln(u_{V_t}) v_t \bar{\tau} + o(\bar{\tau})\right) \right] + o(1)\notag \\
&= 2r - \mathcal{A}^v \ln(v_t) - \mathcal{A}^u \ln(u_{V_t}) v_t + o(1).\notag
\end{align}
\end{proof}

\section{A Rigorous Formulation of the CTC Model}
\label{CTC argument3}
In the following, we show how the CTC model Eq. \eqref{CTC} is formulated. Given a filtered probability space $(\Omega, \mathcal{F}, \{\mathcal{F}^X_t\}_{t\ge 0}, \mathbb{P})$ with $\{\mathcal{F}^X_t\}_{t \ge 0}$ satisfying usual conditions (i.e., $\{\mathcal{F}^X_t\}_{t \ge 0}$ is a right-continuous standard filtration), we suppose that $v$ is a $\mathcal{F}^X$-adapted Markov process in some state space $D_v\subset \mathbb{R}_+$ with $v_0 > 0$ solving the SDE:
\begin{equation}
\label{SDE v}\mathrm{d}v_t = \alpha^v(v_t)\mathrm{d}t + \beta^v(v_t)\sqrt{v_t} \mathrm{d}\bar B_t + \gamma^v(v_{t-})\mathrm{d}\bar J^v_t,
\end{equation}
where $\bar B$ is an $\mathcal{F}^X$-standard Brownian motion and $\bar J^v$ is an $\mathcal{F}^X$-adapted time-changed L\'{e}vy subordinator with activity rate $v_t$, i.e., the local characteristics of $\bar{J}^v$ are given by $(0, 0, v_t \nu(\mathrm{d}x))$ with L\'{e}vy measure $\nu(\mathrm{d}x)$ on $(0, +\infty)$, and $\gamma^v(\cdot) \ge 0$ is the multiplier of jump size. To be precise, the infinitesimal generator of $v$ is given by
$$\begin{aligned}\mathcal{A}^vf(x) &= \alpha^v(x) f^\prime(x) + \frac{1}{2}(\beta^v(x))^2x f^{\prime\prime}(x)\\
& \quad \quad + x\int_{\mathbb{R}_+} \left(f(x+ \gamma^v(x)y) - f(x) - \gamma^v(x) yf^\prime(x) 1_{\{y\le 1\}}\right)\nu(\mathrm{d}y)\\
&= (\alpha^v(x) - a^v(x)x)f^\prime(x) + \frac{1}{2}(\beta^v(x))^2x f^{\prime\prime}(x) \\
&\quad \quad + x\int_{\mathbb{R}_+} \left(f(x+ y) - f(x) - yf^\prime(x) 1_{\{y\le 1\}}\right)N_\nu(x,\mathrm{d}y),\end{aligned}$$
where 
$N_\nu(x, \mathrm{d}y):= \nu(\frac{\mathrm{d}y}{\gamma^v(x)})$ ($N_\nu(x, \mathrm{d}y) = 0$ if $\gamma^v(x) = 0$) and $$a^v(x) = \int_1^{\gamma^v(x)} y N_\nu(x,\mathrm{d}y).$$
We then define $V_t := \int_0^t v_s\mathrm{d}s$ and $\hat{V}_t := \inf\{s: V_s >t\}$. Since $V$ is absolutely continuous and strictly increasing, $V_{\hat{V}_t} = t$ and $\hat{V}_t = \int_0^t v_{\hat{V_s}}^{-1}\mathrm{d}s$. And we have, for every $s \ge 0$,
$$\{\hat{V}_t \le s\} = \{V_s \ge t \}\in \mathcal{F}^X_s,$$
from which $\hat{V}$ is a $\mathcal{F}^X$-time change. 

Next, we assume that $u$ is an $\mathcal{F}_{\hat{V}}^X$-adapted Markov process in some state process $D_u \subset \mathbb{R}_+$ with $u_0 > 0$ solving the SDE:
\begin{equation}
\label{SDE u}
\mathrm{d}u_t = \alpha^u(u_t)\mathrm{d}t + \beta^u(u_t)\sqrt{u_t} \mathrm{d}\tilde Z_t + \gamma^u(u_{t-})\mathrm{d}\tilde J^u_t,
\end{equation}
where $\tilde Z$ is an $\mathcal{F}^X_{\hat{V}}$-standard Brownian motion and $\tilde J^u$ is an $\mathcal{F}^X_{\hat{V}}$-adapted time-changed L\'{e}vy subordinator with activity rate $u_t$, i.e., the local characteristics of $\tilde{J}^u$ are given by $(0, 0, u_t \mu(\mathrm{d}x))$ with L\'{e}vy measure $\mu(\mathrm{d}x)$ on $(0, +\infty)$, and $\gamma^u(\cdot) \ge 0$ is the multiplier of jump size. The infinitesimal generator of $u$ is given by
$$\begin{aligned}\mathcal{A}^uf(x) &= \alpha^u(x) f^\prime(x) + \frac{1}{2}(\beta^u(x))^2x f^{\prime\prime}(x) \\
& \quad \quad + x\int_{\mathbb{R}_+} \left(f(x+ \gamma^u(x)y) - f(x) - \gamma^u(x) f^\prime(x)y 1_{\{y\le 1\}}\right)\mu(\mathrm{d}y)\\
&= (\alpha^u(x) - a^u(x)x)f^\prime(x) + \frac{1}{2}(\beta^u(x))^2x f^{\prime\prime}(x) \\
&\quad \quad + x\int_{\mathbb{R}_+} \left(f(x + y) - f(x) - yf^\prime(x) 1_{\{y\le 1\}}\right)N_\mu(x,\mathrm{d}y),\end{aligned}$$
where 
$N_\mu(x, \mathrm{d}y):= \mu(\frac{\mathrm{d}y}{\gamma^u(x)})$ ($N_\mu(x, \mathrm{d}y) = 0$ if $\gamma^u(x) = 0$) and $$a^u(x) = \int_1^{\gamma^u(x)} y N_\mu(x,\mathrm{d}y).$$

\begin{proposition}
\label{prop:TC SDE}
    If there exists a unique, strong, and strictly positive solution for both Eq. \eqref{SDE v} and Eq. \eqref{SDE u}, then there exists a filtration $\{\mathcal{F}_t\}_{t\ge 0}$ satisfying usual conditions such that $U$ is a $\mathcal{F}$-time change, $V$ is a $\mathcal{F}_U$-time change, and $u$, $v$ satisfy the following SDEs:
    $$
\begin{aligned}
& \mathrm{d} u_t = \alpha^u(u_{t})\mathrm{d}t + \beta^u(u_{t})\mathrm{d} Z(U_t) + \gamma^u(u_{t-}) \mathrm{d} J^u(U_t), \\
& \mathrm{d} v_t = \alpha^v(v_{t})\mathrm{d}t + \beta^v(v_{t})\mathrm{d}\tilde B(V_t) + \gamma^v(v_{t-}) \mathrm{d}\tilde J^v(V_t),
\end{aligned}
$$
where $Z$, $J^u$ are $\mathcal{F}$-adapted Brownian motion and subordinator, respectively, and $\tilde B$, $\tilde J^v$ are $\mathcal{F}_{U}$-adapted Brownian motion and subordinator, respectively.
\end{proposition}

\begin{proof}{Proof}
    Since $\hat{V}$ is a $\mathcal{F}^X$-time change with $V_{\hat{V}_t} \equiv t$, we define
    $$\tilde B_t := \int_0^{\hat{V}_t} \sqrt{v_s}\mathrm{d}\bar B_s, \quad \tilde J^v_t := \bar J^v(\hat{V}_t),$$
then $\tilde B$ is $\mathcal{F}^X_{\hat{V}}$-adapted and its quadratic variation $$\langle \tilde B\rangle_t = \int_0^t v_s \mathrm{d}\langle \bar B\rangle_{\hat{V}_s} = \int_0^t v_s \mathrm{d}\hat{V}_s = t.$$ 
It follows from L\'{e}vy's characterization theorem that $\tilde B$ is an $\mathcal{F}^X_{\hat{V}}$-standard Brownian motion. Moreover, it follows (e.g., from \cite{eberlein2019mathematical} Proposition 4.14) that the local characteristics of $\tilde J^v$ are given by $(0, 0, \nu(\mathrm{d}x))$, which implies that $\tilde J^v$ is an $\mathcal{F}^X_{\hat{V}}$-subordinator.

And following the arguments regarding $v$, we define $U$ and $\hat{U}$ likewise. Then $\hat{U}$ is a time change with respect to $\mathcal{F}_{\hat{V}}^X$. Similarly, denoting by $\mathcal{F}_t = (\mathcal{F}_{\hat{V}}^X)_{\hat{U}_t}$ the filtration of $\mathcal{F}_{\hat{V}}^X$ time-changed by $\hat{U}$, $Z_t := \int_0^{\hat{U}_t} \sqrt{u_s} \mathrm{d}\tilde Z_s$ is an $\mathcal{F}_t$-standard Brownian motion and $J^u_t := \tilde J^u(\hat{U}_t)$ is an $\mathcal{F}$-adapted subordinator. 

Let $t_n \downarrow t$, $A \in \bigcap_{n}\mathcal{F}^X_{\hat{V}_{t_n}}$, and $r_m \downarrow r$, then, by the continuity of $\hat{V}$ and the right-continuity of $\mathcal{F}^X$,
$$A \cap \{\hat{V}_{t} \le r\} = \bigcap_m \bigcup_n \left(A \cap \{\hat{V}_{t_n} \le r_m\}\right)\in \bigcap_m \mathcal{F}_{r_m}^X = \mathcal{F}^X_r.$$
Consequently, $\mathcal{F}^X_{\hat{V}_t} = \bigcap_n \mathcal{F}_{\hat{V}_{t_n}}^X$. By applying the same argument to the filtration $\mathcal{F}_{\hat{V}}$ and its time change $\hat U$, $\mathcal{F}$ is shown to be a right-continuous standard filtration. 

Moreover, since 
$$\{V_t \le s\} = \{\hat{V}_s \ge t\}\in \mathcal{F}_{\hat{V}_s}^X, \quad \forall s \ge 0,$$
$V_t$ is a stopping time with respect to $\mathcal{F}_{\hat{V}}^X \equiv \mathcal{F}_{U}$ for every $t\ge 0$. Thus, $V$ is an $\mathcal{F}_U$-time change. Similarly, since
$$\{U_t \le s\} = \{\hat{U}_s \ge t\}\in \mathcal{F}_s, \quad \forall s\ge 0,$$
$U_t$ is a stopping time with respect to $\mathcal{F}_t$ for every $t\ge 0$, from which $U$ is an $\mathcal{F}$-time change. 

Consequently, we recover the model setup in the beginning of Section \ref{CTC model} and the dynamics of $u$, $v$ in a time-changed form:
$$
\begin{aligned}
& \mathrm{d} u_t = \alpha^u(u_{t})\mathrm{d}t + \beta^u(u_{t})\mathrm{d} Z(U_t) + \gamma^u(u_{t-}) \mathrm{d} J^u(U_t), \\
& \mathrm{d} v_t = \alpha^v(v_{t})\mathrm{d}t + \beta^v(v_{t})\mathrm{d}\tilde B(V_t) + \gamma^v(v_{t-}) \mathrm{d}\tilde J^v(V_t).
\end{aligned}
$$
\end{proof}

\begin{corollary}
\label{coro:TC SDE}
    Under the condition of Corollary \ref{affine spec}, where $\alpha^u(u_t) = \kappa_u(\theta_u - u_t)$, $\alpha^v(v_t) = \kappa_v(\theta_v - v_t)$ with $\kappa_i \theta_i \ge 0$, $\beta^i(\cdot) \equiv \sigma_i \ge 0$, and $\gamma^i(\cdot) \equiv \eta_i \ge 0$, $i = u,v$, and, in addition, $2\kappa_i\theta_i > \sigma_i$, $i = u, v$, the assumption in Proposition~\ref{prop:TC SDE} holds.
\end{corollary}

\begin{proof}{Proof}
    Under the affine specifications of $u$ and $v$ in Corollary \ref{affine spec}, there exists unique strong solutions of SDEs \eqref{SDE v} and \eqref{SDE u} due to \cite{kawazu1971branching} (the converse part of Theorem 1.1). Since $2\kappa_i\theta_i > \sigma_i$, $i = u, v$ also holds, $u$ and $v$ are almost surely positive by Feller's condition. Thus, the condition in Proposition~\ref{prop:TC SDE} is met.
\end{proof}

\section{Numerical Accuracy of CTC-COS Method}
\label{accuracy of COS}
Table \ref{Cos Accuracy} validates the accuracy of CTC-COS method and compares the numerical call option prices with the simulated values. Simulations are performed with path number $5 \times 10 ^5$ and discretization degree $\delta = 5 \times 10 ^ {-6}$. The relative error is defined as $$
\text{Rel. Err.} = \frac{|C_{\text{CTC-COS}} (K, \tau) - C_{\text{Simu}}(K, \tau)|}{C_{\text{Simu}}(K,\tau)},$$
and the ratio in the table is defined as 
$$\text{Ratio} =  \frac{|C_{\text{CTC-COS}} (K, \tau) - C_{\text{Simu}}(K,\tau)|}{\text{Std.Err.}(K,\tau)},$$
where $\text{Std.Err.}(K,\tau)$ is the standard error of the simulated option prices for time to maturity $\tau$ and strike price $K$. The relative error in Table \ref{Cos Accuracy2} is defined likewise by substituting the option prices by the corresponding implied volatilities.

As shown in Table \ref{Cos Accuracy} and Table \ref{Cos Accuracy2}, the numerical method exhibits high accuracy with respect to both option prices and implied volatility.
\begin{table}[H]
\caption{Validation of numerical accuracy of the CTC-COS method for the call option prices under Composite Heston model\label{Cos Accuracy}}
\begin{center}
 \scalebox{0.6}{
\begin{tabular}{ccccccccccccc}  
\hline ($K / F$, $\tau$) & CTC-COS & Simu. & Std. Err. ($10^{-2}$) & Rel. Err. (\%)& Ratio\\  
\hline  
(1.0, 0.02) & 0.9293 & 0.9292 & 0.177 & 0.016 & 0.082 \\  
(0.9, 0.05) & 10.1159 & 10.1078 & 0.569 & 0.081 & 1.433 \\  
(1.0, 0.05) & 1.5805 & 1.5808 & 0.309 & 0.019 & -0.096 \\  
(1.1, 0.05) & 0.0173 & 0.0170 & 0.039 & 6.686 & -0.382 \\  
(0.8, 0.1) & 20.0877 & 20.0671 & 0.940 & 0.103 & 2.249 \\  
(0.9, 0.1) & 10.4737 & 10.4590 & 0.815 & 0.141 & 2.235 \\  
(1.0, 0.1) & 2.4958 & 2.4956 & 0.497 & 0.007 & 0.035 \\  
(1.1, 0.1) & 0.1521 & 0.1529 & 0.154 & 0.503 & -0.500 \\  
(0.8, 0.15) & 20.2457 & 20.2386 & 1.192 & 0.035 & 0.933 \\  
(0.9, 0.15) & 10.9084 & 10.9086 & 1.015 & 0.001 & -0.013 \\  
(1.0, 0.15) & 3.3141 & 3.3252 & 0.667 & 0.333 & -1.658 \\  
(1.1, 0.15) & 0.4108 & 0.4136 & 0.285 & 0.672 & -0.631 \\  
(0.7, 0.2) & 30.1369 & 30.1289 & 1.511 & 0.026 & 0.526 \\  
(0.8, 0.2) & 20.4502 & 20.4468 & 1.405 & 0.016 & 0.239 \\  
(0.9, 0.2) & 11.3681 & 11.3719 & 1.190 & 0.034 & -0.322 \\  
(1.0, 0.2) & 4.0634 & 4.0764 & 0.822 & 0.317 & -1.572 \\  
(1.1, 0.2) & 0.7634 & 0.7683 & 0.417 & 0.633 & -1.165 \\  
(1.2, 0.2) & 0.1411 & 0.1423 & 0.200 & 1.017 & -0.319 \\  
(0.6, 0.3) & 40.1203 & 40.1167 & 2.007 & 0.009 & 0.181 \\  
(0.7, 0.3) & 30.3522 & 30.3502 & 1.925 & 0.007 & 0.103 \\  
(0.8, 0.3) & 20.9276 & 20.9294 & 1.769 & 0.008 & -0.100 \\  
(0.9, 0.3) & 12.2883 & 12.2981 & 1.503 & 0.080 & -0.657 \\  
(1.0, 0.3) & 5.4028 & 5.4252 & 1.110 & 0.412 & -2.015 \\  
(1.1, 0.3) & 1.6247 & 1.6444 & 0.683 & 1.193 & -2.871 \\  
(1.2, 0.3) & 0.4321 & 0.4412 & 0.391 & 3.428 & -1.941 \\  
(0.6, 0.5) & 40.3760 & 40.3794 & 2.712 & 0.008 & -0.124 \\  
(0.7, 0.5) & 30.8987 & 30.9100 & 2.572 & 0.037 & -0.439 \\  
(0.8, 0.5) & 21.9617 & 21.9840 & 2.348 & 0.102 & -0.952 \\  
(0.9, 0.5) & 13.9954 & 14.0244 & 2.021 & 0.206 & -1.432 \\  
(1.0, 0.5) & 7.6360 & 7.6664 & 1.603 & 0.396 & -1.895 \\  
(1.1, 0.5) & 3.4802 & 3.5084 & 1.158 & 0.802 & -2.430 \\  
(1.2, 0.5) & 1.4046 & 1.4184 & 0.784 & 0.970 & -1.755 \\  
(0.6, 0.7) & 40.6905 & 40.6779 & 3.280 & 0.031 & 0.385 \\  
(0.7, 0.7) & 31.5027 & 31.4981 & 3.097 & 0.015 & 0.151 \\  
(0.8, 0.7) & 22.9855 & 22.9916 & 2.827 & 0.026 & -0.214 \\  
(0.9, 0.7) & 15.5063 & 15.5194 & 2.467 & 0.084 & -0.531 \\  
(1.0, 0.7) & 9.4844 & 9.5094 & 2.035 & 0.262 & -1.226 \\  
(1.1, 0.7) & 5.2188 & 5.2536 & 1.583 & 0.663 & -2.200 \\  
(1.2, 0.7) & 2.6356 & 2.6677 & 1.174 & 1.201 & -2.730 \\  
(0.6, 0.9) & 41.0381 & 41.0180 & 3.765 & 0.049 & 0.535 \\  
(0.7, 0.9) & 32.1250 & 32.1155 & 3.550 & 0.030 & 0.267 \\  
(0.8, 0.9) & 23.9659 & 23.9684 & 3.249 & 0.010 & -0.076 \\  
(0.9, 0.9) & 16.8586 & 16.8730 & 2.868 & 0.085 & -0.501 \\  
(1.0, 0.9) & 11.0831 & 11.1079 & 2.430 & 0.224 & -1.023 \\  
(1.1, 0.9) & 6.7949 & 6.8278 & 1.975 & 0.482 & -1.666 \\  
(1.2, 0.9) & 3.9225 & 3.9574 & 1.551 & 0.883 & -2.252 \\  
(1.3, 0.9) & 2.1762 & 2.2072 & 1.192 & 1.405 & -2.602 \\  
\hline  
\end{tabular}
}
\end{center}
{\textit{Note.} This table validates the numerical accuracy of the CTC-COS method for the call option prices under Composite Heston model. The Std. Simu refers to the standard error of the simulated option price. The parameters are set as $S_0 = 100$, $\kappa_u = 6.0$, $\theta_u = 0.08$, $\sigma_u = 1.5$, $\kappa_v = 3.0$, $\theta_v = 1.5$, $\sigma_v = 0.5$, $\rho_u = -0.5$, $u_0 = 0.02$, $v_0 = 1.3$. Simulations are performed with path number $5 \times 10 ^5$ and discretization degree $\delta = 5 \times 10 ^ {-6}$. Std.Err. denotes the standard error of the simulated option prices. Rel. Err. is the $L^1$ relative error defined as $|C_{\text{CTC-COS}} (K,\tau) - C_{\text{Simu}}(K,\tau)| / C_{\text{Simu}}(K,\tau)$. And the ratio in the table is defined as $|C_{\text{CTC-COS}} (K,\tau) - C_{\text{Simu}}(K,\tau)| / {\text{Std.Err.}(K,\tau)}$.} 
\end{table}

\begin{table}[H]
\caption{Validation of numerical accuracy of the CTC-COS method for the implied volatility of Composite Heston model
\label{Cos Accuracy2}}
\begin{center}
 \scalebox{0.7}{
\begin{tabular}{ccccccccccccc}
\hline ($K/ F$, $\tau$) & CTC-COS & Simu. & Rel. Err. (\%)& \\
\hline 
(1.0, 0.02) & 0.1647 & 0.1647 & 0.0157 \\  
(0.9, 0.05) & 0.2813 & 0.2784 & 1.031 \\  
(1.0, 0.05) & 0.1772 & 0.1772 & 0.0188 \\  
(1.1, 0.05) & 0.1874 & 0.1869 & 0.244 \\  
(0.8, 0.1) & 0.3548 & 0.3606 & 1.593 \\  
(0.9, 0.1) & 0.2788 & 0.2805 & 0.575 \\  
(1.0, 0.1) & 0.1979 & 0.1978 & 0.007 \\  
(1.1, 0.1) & 0.1902 & 0.1904 & 0.107 \\  
(0.8, 0.15) & 0.3482 & 0.3461 & 0.594 \\  
(0.9, 0.15) & 0.2805 & 0.2805 & 0.005 \\  
(1.0, 0.15) & 0.2146 & 0.2153 & 0.333 \\  
(1.1, 0.15) & 0.1977 & 0.1981 & 0.188 \\  
(0.7, 0.2) & 0.4039 & 0.4001 & 0.957 \\  
(0.8, 0.2) & 0.3440 & 0.3434 & 0.176 \\  
(0.9, 0.2) & 0.2836 & 0.2839 & 0.113 \\  
(1.0, 0.2) & 0.2279 & 0.2286 & 0.317 \\  
(1.1, 0.2) & 0.2070 & 0.2074 & 0.213 \\  
(1.2, 0.2) & 0.2219 & 0.2222 & 0.148 \\  
(0.6, 0.3) & 0.4437 & 0.4417 & 0.451 \\  
(0.7, 0.3) & 0.3904 & 0.3900 & 0.113 \\  
(0.8, 0.3) & 0.3392 & 0.3394 & 0.055 \\  
(0.9, 0.3) & 0.2901 & 0.2907 & 0.204 \\  
(1.0, 0.3) & 0.2474 & 0.2485 & 0.413 \\  
(1.1, 0.3) & 0.2249 & 0.2261 & 0.511 \\  
(1.2, 0.3) & 0.2272 & 0.2284 & 0.486 \\  
(0.6, 0.5) & 0.4156 & 0.4163 & 0.169 \\  
(0.7, 0.5) & 0.3737 & 0.3749 & 0.320 \\  
(0.8, 0.5) & 0.3351 & 0.3365 & 0.412 \\  
(0.9, 0.5) & 0.3002 & 0.3015 & 0.408 \\  
(1.0, 0.5) & 0.2711 & 0.2722 & 0.398 \\  
(1.1, 0.5) & 0.2515 & 0.2526 & 0.436 \\  
(1.2, 0.5) & 0.2438 & 0.2446 & 0.319 \\  
(0.6, 0.7) & 0.3978 & 0.3962 & 0.407 \\  
(0.7, 0.7) & 0.3639 & 0.3635 & 0.093 \\  
(0.8, 0.7) & 0.3336 & 0.3338 & 0.084 \\  
(0.9, 0.7) & 0.3069 & 0.3074 & 0.147 \\  
(1.0, 0.7) & 0.2848 & 0.2856 & 0.264 \\  
(1.1, 0.7) & 0.2685 & 0.2696 & 0.406 \\  
(1.2, 0.7) & 0.2588 & 0.2600 & 0.482 \\  
(0.6, 0.9) & 0.3859 & 0.3840 & 0.483 \\  
(0.7, 0.9) & 0.3577 & 0.3572 & 0.148 \\  
(0.8, 0.9) & 0.3330 & 0.3331 & 0.028 \\  
(0.9, 0.9) & 0.3116 & 0.3120 & 0.138 \\  
(1.0, 0.9) & 0.2938 & 0.2945 & 0.225 \\  
(1.1, 0.9) & 0.2800 & 0.2809 & 0.318 \\  
(1.2, 0.9) & 0.2703 & 0.2714 & 0.403 \\  
(1.3, 0.9) & 0.2646 & 0.2659 & 0.469 \\  
\hline
\end{tabular}}
\end{center}
{\textit{Note.} The parameters are set as $S_0 = 100$, $\kappa_u = 6.0$, $\theta_u = 0.08$, $\sigma_u = 1.5$, $\kappa_v = 3.0$, $\theta_v = 1.5$, $\sigma_v = 0.5$, $\rho_u = -0.5$, $u_0 = 0.02$, $v_0 = 1.3$. Simulation is performed with path number $5 \times 10 ^5$ and discretization $\delta = 5 \times 10 ^ {-6}$. Rel. Err. is the $L^1$ relative error defined as $|\sigma_{\text{CTC-COS}} (k,\tau) - \sigma_{\text{Simu}}(k,\tau)| / \sigma_{\text{Simu}}(k,\tau)$.}
\end{table}  

\section{Calibration Details}
Spot index prices such as SPX and VIX are not used in the calibration because they are not directly traded in the market, see also \cite{lian2013pricing}. They are recovered from the market according to the put-call parity as discounted futures price. Meanwhile, such implied spot prices contain the term structure of future dividend expectations.
\begin{itemize}
    \item SPX options with AM settlement are specially treated with 1 day less maturity
    \item The risk-free rate is quoted from daily U.S. treasury bond rates with Spline interpolation
    \item Implied volatility is a function of log moneyness $k_t = \log (\frac{K}{F_t})$ (and no additional rates) because, for $\tau = \bar T - t$, Black-Scholes formula yields
    $$
\begin{aligned}
e^{-r\tau} E_t[S_{\bar T} - K]_+ &  = C_t(K, \bar T) \equiv C^{B S}\left(k_t, K, \tau, I V_t\right) \\
& =K e^{-r \tau}\left[e^{k_t} \Phi\left(d_1\right)-\Phi\left(d_2\right)\right],
\end{aligned}
$$
which is reduced to
$$e^{k_t} \Phi\left(d_1\left(k_t, I V_t, \tau\right)\right)-\Phi\left(d_2\left(k_t, I V_t, \tau\right)\right)=E_t\left[e^{k_{\bar T}}-1\right]_{+}.$$

To price and compute $\text{IV}$, it's enough to obtain the moneyness. By put-call parity,
$$C - P = e^{-r\tau}(F - K),$$
if $r$ is known, then 
$$k_t = \ln\frac{K}{K + (C_t-P_t)e^{r\tau}}.$$
\end{itemize}

\section{Supplementary Analysis of Calibration Results}
\subsection{Out-of-sample Calibration Performance}

Tables \ref{out1}, \ref{out2}, and \ref{out3} present the out-of-sample joint calibration performance of the four models across three consecutive sub-periods (P1, P2, P3), evaluated via multi-error metrics (RMSRE, RMSE, MAE) and pairwise $t$-statistics for statistical significance.

In period 1, composite models consistently outperform their ordinary counterparts in P1. For SPX options, the Composite Heston reduces RMSE by 28.2\% (from 0.0383 to 0.0275), while the Composite JH cuts it by 9.1\% (from 0.0231 to 0.0210). The improvement is far more pronounced in VIX markets: the Composite Heston’s VIX RMSE drops by 14.3\% (from 0.1789 to 0.1534), and the Composite JH achieves a 40.9\% reduction (from 0.1303 to 0.0769). Aggregate metrics confirm this superiority—Composite JH’s RMSRE (0.0837) and MAE (0.0384) are the lowest among all models. Panel B’s $t$-statistics (e.g., Composite JH vs. JH: -6.65; Composite Heston vs. Heston: -9.23) indicate all improvements are statistically significant.

Peiod 2 sees the most dramatic performance gains from composite models. The Composite Heston slashes RMSRE by 45.8\% (from 0.3685 to 0.1999) and aggregate RMSE by 14.8\% (from 0.1881 to 0.1603). The Composite JH outperforms even more substantially: its SPX RMSE drops by 45.4\% (from 0.0268 to 0.0145), VIX RMSE by 29.1\% (from 0.1450 to 0.1028), and MAE by 34.7\% (from 0.0684 to 0.0447) versus JH. Pairwise $t$-statistics (e.g., Composite JH vs. Heston: -23.74; Composite Heston vs. Heston: -15.33) confirm the outperformance is highly significant, highlighting composite models’ robustness in volatile market conditions.

The superiority of composite models persists in period 3, with minor exceptions. The Composite Heston underperforms Heston slightly in SPX RMSE (0.0421 vs. 0.0289) but still improves VIX RMSE by 28.0\% (from 0.2670 to 0.1925) and aggregate RMSRE by 28.5\% (from 0.2491 to 0.1780). The Composite JH remains the top performer: it reduces JH’s VIX RMSE by 49.9\% (from 0.2218 to 0.1111), aggregate RMSE by 43.6\% (from 0.1184 to 0.0669), and MAE by 52.6\% (from 0.1046 to 0.0496). Panel B’s $t$-statistics (e.g., Composite JH vs. JH: -9.99; Composite JH vs. Composite Heston: -22.96) validate the statistical significance of these gains, with the exception that the Composite Heston and JH show negligible performance differences ($t$-statistic = -0.94).

Overall, composite models deliver robust out-of-sample improvements across all three periods, with the most notable gains in VIX markets. The Composite JH consistently achieves the lowest error metrics and strongest statistical significance, confirming the value of integrating composite time changes with jump dynamics for joint calibration.

\begin{table}[H]
\caption{Comparison of the out-of-sample joint calibration performance (P1)\label{out1}}
\begin{center}
    \begin{tabular}{ccccc}
       \hline & Heston & \makecell[c]{Composite\\ Heston} & JH & 
       \makecell[c]{Composite\\ JH}\\\hline
       \multicolumn{5}{c}{Panel A. Summary of Statistics of Performance}\\\hline
        \makecell[c]{RMSRE\\ \quad} & \makecell[c]{0.1784\\ (0.0375)} & \makecell[c]{0.1297\\ (0.0285)} & \makecell[c]{0.1114\\ (0.0239)} & \makecell[c]{0.0837\\
       (0.0197)}\\
       \makecell[c]{RMSE (SPX)\\ \quad} & \makecell[c]{0.0383\\ (0.0104)}& \makecell[c]{0.0275 \\(0.0085)} & \makecell[c]{0.0231 \\(0.0052)} & \makecell[c]{0.0210 \\(0.0052)}\\
       \makecell[c]{RMSE (VIX)\\ \quad} & \makecell[c]{0.1789\\ (0.0569)} & \makecell[c]{0.1534 \\(0.0533)} & \makecell[c]{0.1303 \\(0.0383)} & \makecell[c]{0.0769 \\(0.0256)}\\
       \makecell[c]{RMSE (Aggregate)\\ \quad} & \makecell[c]{0.1086 \\(0.0295)} & \makecell[c]{0.0905\\ (0.0277)} & \makecell[c]{0.0767\\ (0.0194)} & \makecell[c]{0.0489 \\(0.0135)}\\
       \makecell[c]{MAE\\ \quad} & \makecell[c]{0.0823 \\(0.0220)} & \makecell[c]{0.0717\\ (0.0203)} & \makecell[c]{0.0638\\ (0.0164)} & \makecell[c]{0.0384 \\ (0.0105)}\\\hline
                \multicolumn{5}{c}{Panel B. Pairwise $t$-statistics}\\\hline
        Heston & 0 & 9.23 & 13.25 & 18.07\\
        Composite Heston & -9.23 & 0 & 3.95 & 12.72\\
        JH & -13.25 & -3.95 & 0 & 6.65\\
        Composite JH & -18.07 & -12.72 & -6.65 & 0\\\hline
    \end{tabular}
    \end{center}
    {\textit{Note.} Panel A compares the out-of-sample (P1) joint calibration performance of four models (Heston, Composite Heston, JH, Composite JH) using error metrics including RMSE, RMSRE and MAE. All values are the daily average of errors, with the standard deviation of the errors shown in parentheses. Panel B reports pairwise $t$-statistics (defined in Equation \eqref{t}), where the $i$-th row corresponds to Model $i$ and the $j$-th column to Model $j$. A significantly negative statistic indicates Model $i$ outperforms Model $j$, while a positive statistic indicates the opposite.}
\end{table}

\begin{table}[H]
\caption{Comparison of the out-of-sample joint calibration performance (P2)\label{out2}}
\begin{center}
    \begin{tabular}{ccccc}
       \hline & Heston & \makecell[c]{Composite\\ Heston} & JH & 
       \makecell[c]{Composite\\ JH}\\\hline
       \multicolumn{5}{c}{Panel A. Summary of Statistics of Performance}\\\hline
        \makecell[c]{RMSRE\\ \quad} & \makecell[c]{0.3685\\ (0.0688)} & \makecell[c]{0.1999\\ (0.0312)} & \makecell[c]{0.1441\\ (0.0445)} & \makecell[c]{0.0894\\
       (0.0184)}\\
       \makecell[c]{RMSE (SPX)\\ \quad} & \makecell[c]{0.0593\\ (0.0080)}& \makecell[c]{0.0371 \\(0.0067)} & \makecell[c]{0.0268 \\(0.0062)} & \makecell[c]{0.0145 \\(0.0036)}\\
       \makecell[c]{RMSE (VIX)\\ \quad} & \makecell[c]{0.3170\\ (0.0817)} & \makecell[c]{0.2834 \\(0.0836)} & \makecell[c]{0.1450 \\(0.0728)} & \makecell[c]{0.1028 \\(0.0435)}\\
       \makecell[c]{RMSE (Aggregate)\\ \quad} & \makecell[c]{0.1881 \\(0.0428)} & \makecell[c]{0.1603\\ (0.0424)} & \makecell[c]{0.0859 \\ (0.0378)} & \makecell[c]{0.0587 \\(0.0219)}\\
       \makecell[c]{MAE\\ \quad} & \makecell[c]{0.1351 \\(0.0367)} & \makecell[c]{0.1174\\ (0.0304)} & \makecell[c]{0.0684\\ (0.0326)} & \makecell[c]{0.0447 \\ (0.0141)}\\\hline
              \multicolumn{5}{c}{Panel B. Pairwise $t$-statistics}\\\hline
        Heston & 0 & 15.33 & 34.51 & 23.74\\
        Composite Heston & -15.33 & 0 & 8.54 & 26.83\\
        JH & -34.51 & -8.54 & 0 & 7.82\\
        Composite JH & -23.74 & -26.83 & -7.82 & 0\\\hline
    \end{tabular}
    \end{center}
    {\textit{Note.} Panel A compares the out-of-sample (P2) joint calibration performance of four models (Heston, Composite Heston, JH, Composite JH) using error metrics including RMSE, RMSRE and MAE. All values are the daily average of errors, with the standard deviation of the errors shown in parentheses. Panel B reports pairwise $t$-statistics (defined in Equation \eqref{t}), where the $i$-th row corresponds to Model $i$ and the $j$-th column to Model $j$. A significantly negative statistic indicates Model $i$ outperforms Model $j$, while a positive statistic indicates the opposite.}
\end{table}

\begin{table}[H]
\caption{Comparison of the out-of-sample joint calibration performance (P3)\label{out3}}
\begin{center}
    \begin{tabular}{ccccc}
       \hline & Heston & \makecell[c]{Composite\\ Heston} & JH & 
       \makecell[c]{Composite\\ JH}\\\hline
       \multicolumn{5}{c}{Panel A. Summary of Statistics of Performance}\\\hline
        \makecell[c]{RMSRE\\ \quad} & \makecell[c]{0.2491\\ (0.0672)} & \makecell[c]{0.1780\\ (0.0317)} & \makecell[c]{0.1723\\ (0.0501)} & \makecell[c]{0.1000\\
       (0.0307)}\\
       \makecell[c]{RMSE (SPX)\\ \quad} & \makecell[c]{0.0289\\ (0.0054)}& \makecell[c]{0.0421 \\(0.0115)} & \makecell[c]{0.0151 \\(0.0037)} & \makecell[c]{0.0226 \\(0.0142)}\\
       \makecell[c]{RMSE (VIX)\\ \quad} & \makecell[c]{0.2670\\ (0.1011)} & \makecell[c]{0.1925 \\(0.0470)} & \makecell[c]{0.2218 \\(0.0930)} & \makecell[c]{0.1111 \\(0.0455)}\\
       \makecell[c]{RMSE (Aggregate)\\ \quad} & \makecell[c]{0.1479 \\(0.0518)} & \makecell[c]{0.1173\\ (0.0256)} & \makecell[c]{0.1184\\ (0.0476)} & \makecell[c]{0.0669 \\(0.0251)}\\
       \makecell[c]{MAE\\ \quad} & \makecell[c]{0.1162 \\(0.0435)} & \makecell[c]{0.0935\\ (0.0216)} & \makecell[c]{0.1046\\ (0.0438)} & \makecell[c]{0.0496 \\ (0.0175)}\\\hline
              \multicolumn{5}{c}{Panel B. Pairwise $t$-statistics}\\\hline
        Heston & 0 & 7.45 & 17.54 & 13.90\\
        Composite Heston & -7.45 & 0 & 0.94 & 22.96\\
        JH & -17.54 & -0.94 & 0 & 9.99\\
        Composite JH & -13.90 & -22.96 & -9.99 & 0\\\hline
    \end{tabular}
    \end{center}
    {\textit{Note.} Panel A compares the out-of-sample (P3) joint calibration performance of four models (Heston, Composite Heston, JH, Composite JH) using error metrics including RMSE, RMSRE and MAE. All values are the daily average of errors, with the standard deviation of the errors shown in parentheses. Panel B reports pairwise $t$-statistics (defined in Equation \eqref{t}), where the $i$-th row corresponds to Model $i$ and the $j$-th column to Model $j$. A significantly negative statistic indicates Model $i$ outperforms Model $j$, while a positive statistic indicates the opposite.}
\end{table}
\subsection{ATM Characteristics}
We complement the analysis for Heston and Composite Heston models by examining the comparison of both the short-term 
(Figures~\ref{H IV SPX}, \ref{H IV SPX Skew} for SPX options and Figure~\ref{H IV VIX}, \ref{H IV VIX Skew} for VIX options) 
and the long-term
(Figures~\ref{H IV SPX long term}, \ref{H IV SPX skew long term} for SPX options and Figure~\ref{H IV VIX long term}, \ref{H IV VIX skew long term} for VIX options) calibration error for ATM implied volatility and ATM volatility skew.

\begin{figure}[H]
\caption{The short-term ATM implied volatility of Heston and Composite Heston in SPX markets}
 \begin{minipage}{\linewidth}
 \subfloat[][Short-term ATM IV]{
 \includegraphics[width=\linewidth]{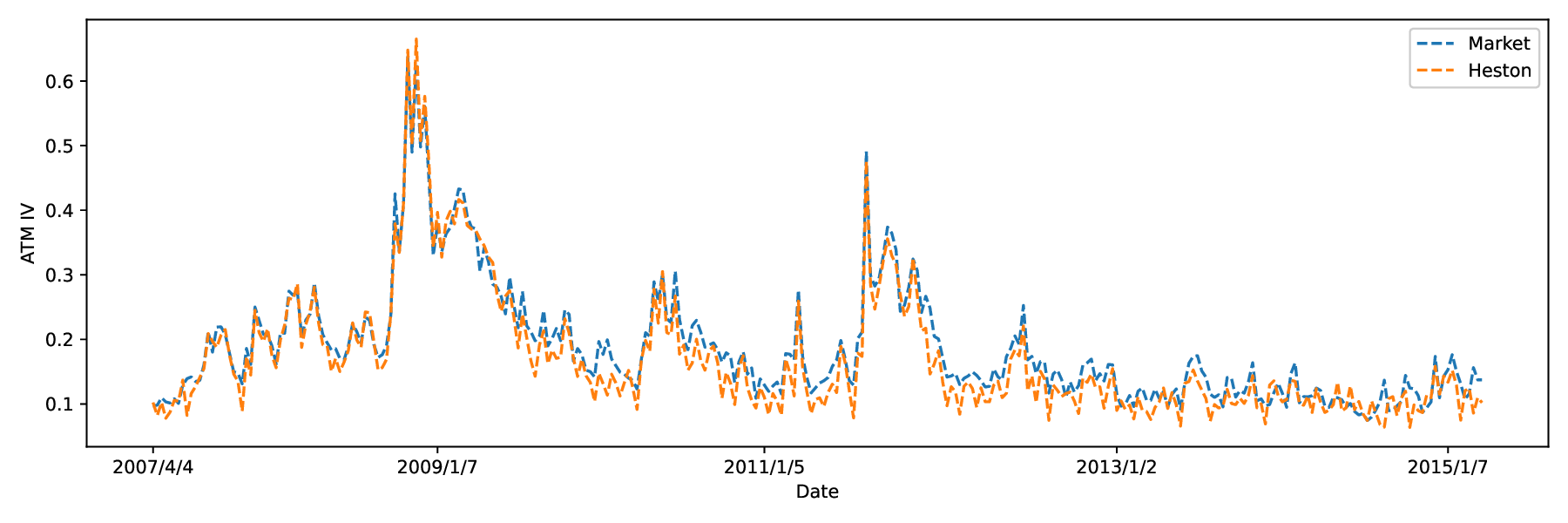}
 }
 \end{minipage} \vspace{-1em}
  \begin{minipage}{\linewidth}
 \subfloat[][Short-term ATM IV]{
 \includegraphics[width=\linewidth]{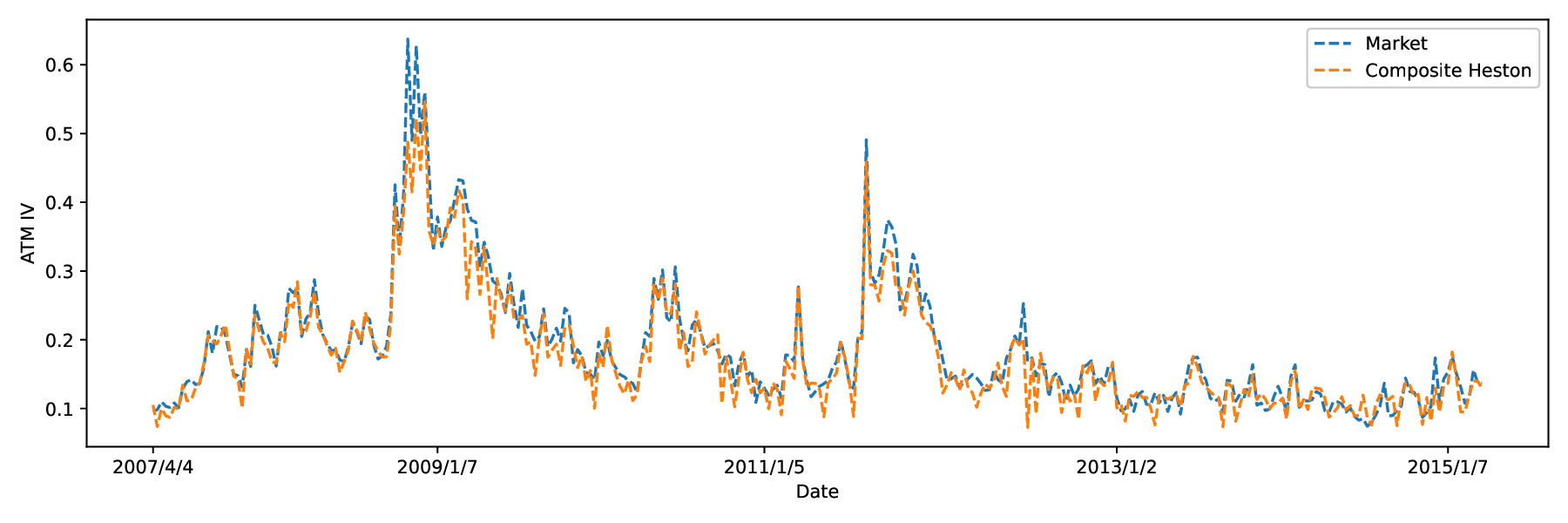}
}
 \end{minipage} \vspace{-1em}
 \begin{minipage}{\textwidth}
\subfloat[][ATM IV difference]{
\includegraphics[width=\textwidth]{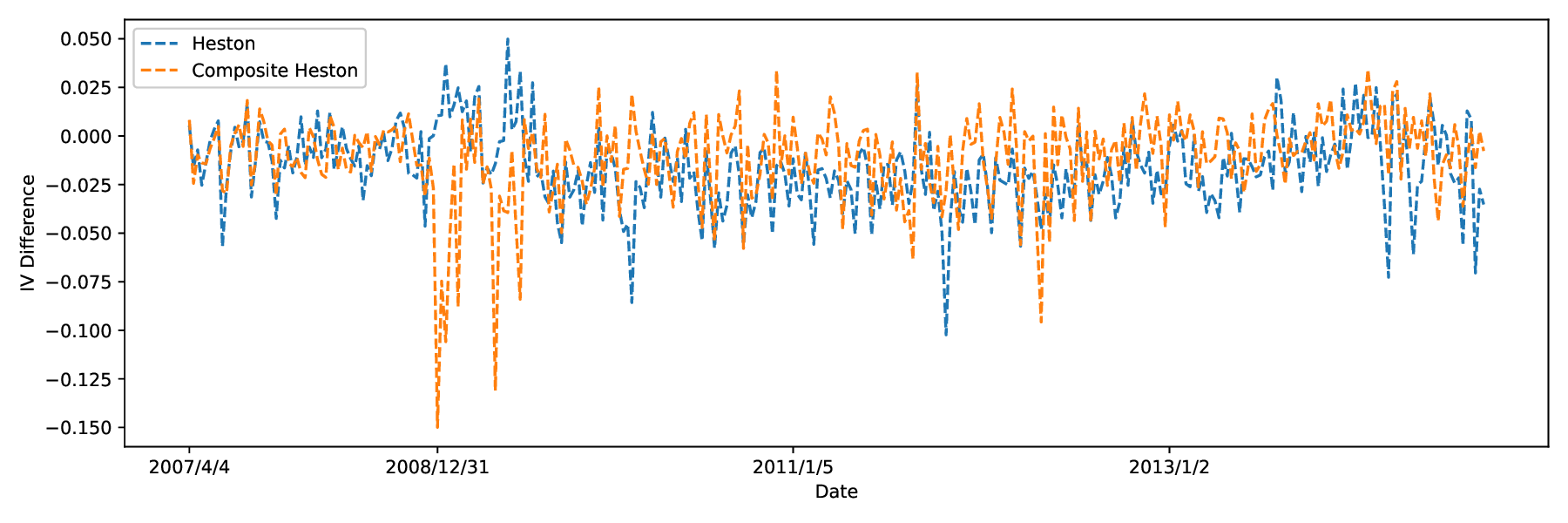}
}
\end{minipage}
\vspace{-2pt}
\footnotesize \textit{Note. } The plots compare the short-term ATM IV, defined by Eq. \eqref{IV}, between the models (Heston and Composite Heston) and the market. The first row plots the ATM IV of Heston model, the second row plots that of Composite Heston model, and the last row compares the corresponding difference (model - market) between models and the market. Maturities of SPX options smaller than 30 days are considered short-term.\label{H IV SPX}
\end{figure}

\begin{figure}[H]
\caption{The short-term volatility skew of Heston and Composite Heston in SPX markets}
 \begin{minipage}{\linewidth}
 \subfloat[][Short-term volatility skew]{
 \includegraphics[width=\linewidth]{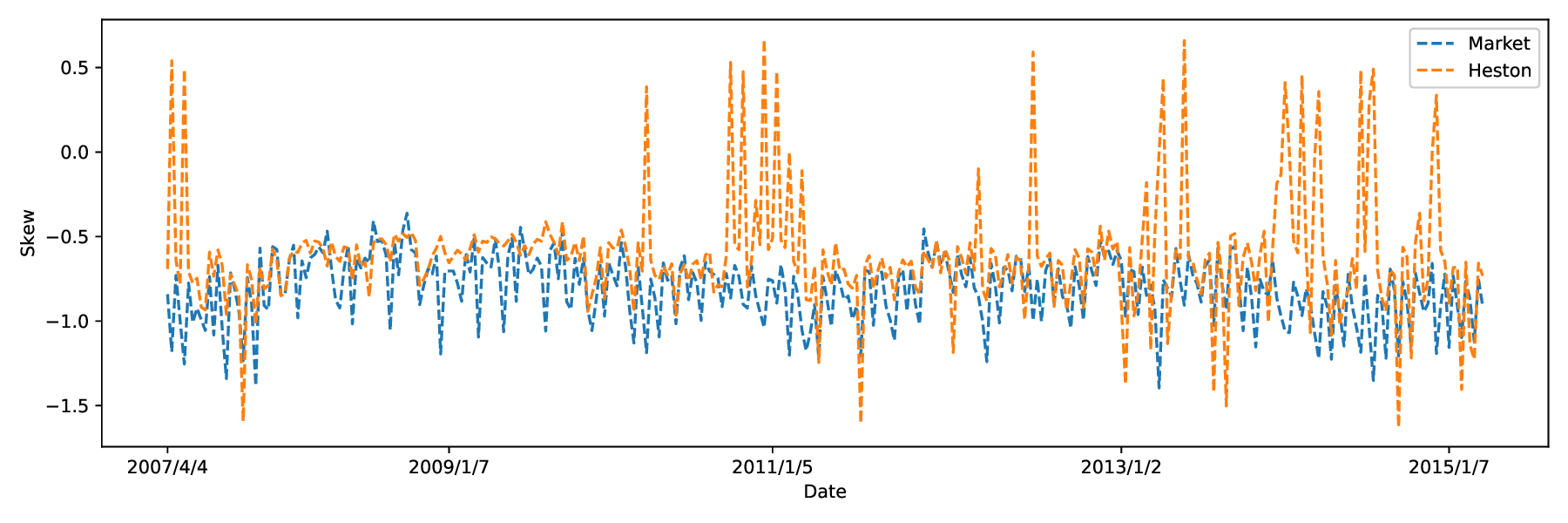}
 }
 \end{minipage} \vspace{-1em}
  \begin{minipage}{\linewidth}
 \subfloat[][Short-term volatility skew]{
 \includegraphics[width=\linewidth]{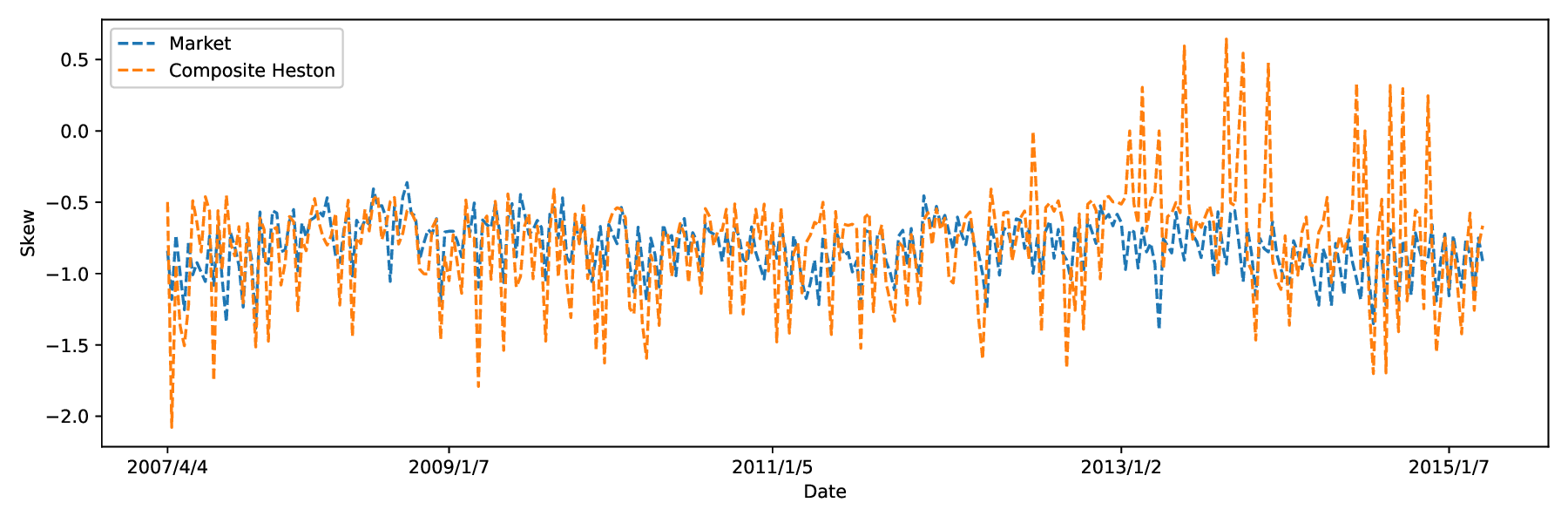}
}
 \end{minipage} \vspace{-1em}
 \begin{minipage}{\textwidth}
\subfloat[][Volatility skew difference]{
\includegraphics[width=\textwidth]{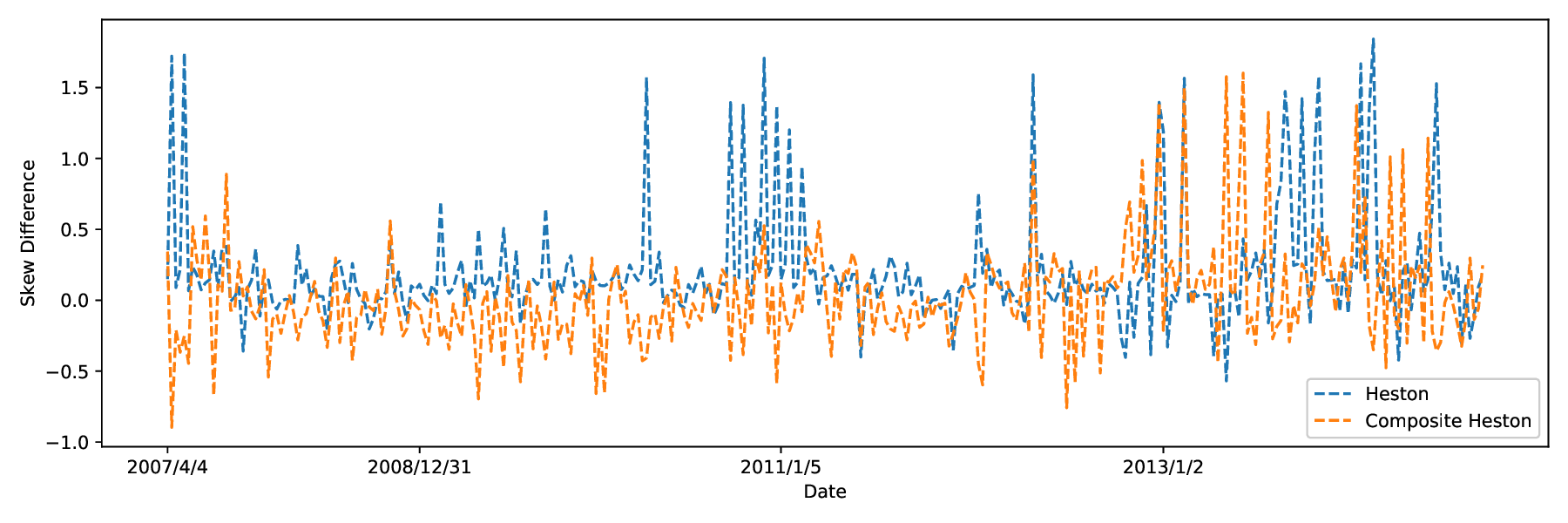}
}
\end{minipage}
\vspace{-2pt}
\footnotesize \textit{Note. } The plots compare the short-term volatility skew, defined by Eq. \eqref{Skew}, between the models (Heston and Composite Heston) and the market. The first row plots the volatility skew of Heston model, the second row plots that of Composite Heston model, and the last row compares the corresponding difference (model - market) between models and the market. Maturities of SPX options smaller than 30 days are considered short-term.\label{H IV SPX Skew}
\end{figure}

\begin{figure}[H]
\caption{The long-term ATM implied volatility of Heston and Composite Heston in SPX markets}
 \begin{minipage}{\linewidth}
 \subfloat[][Long-term ATM IV]{
 \includegraphics[width=\linewidth]{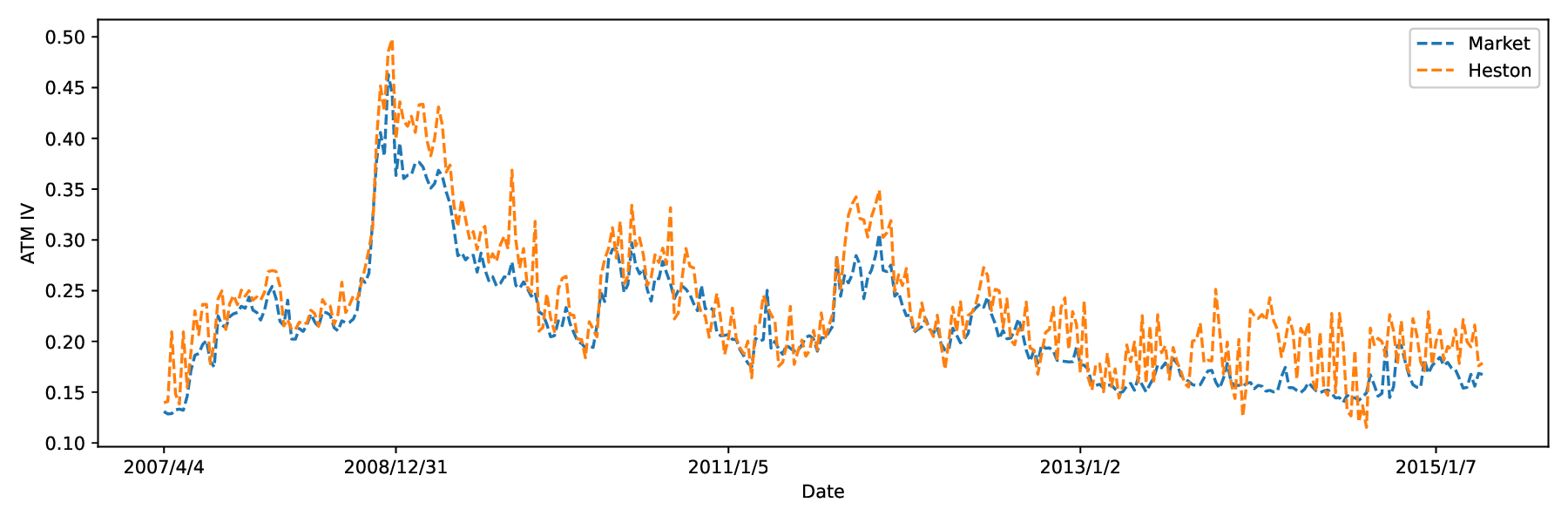}
 }
 \end{minipage} \vspace{-1em}
  \begin{minipage}{\linewidth}
 \subfloat[][Long-term ATM IV]{
 \includegraphics[width=\linewidth]{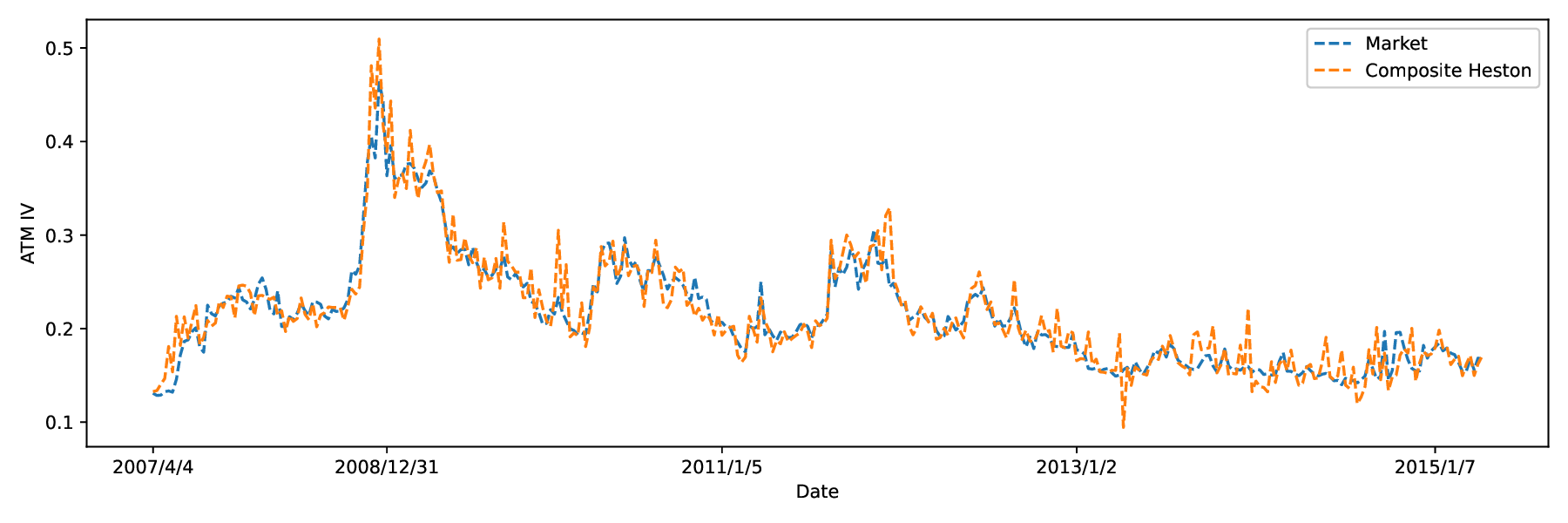}
}
 \end{minipage} \vspace{-1em}
 \begin{minipage}{\textwidth}
\subfloat[][ATM IV difference]{
\includegraphics[width=\textwidth]{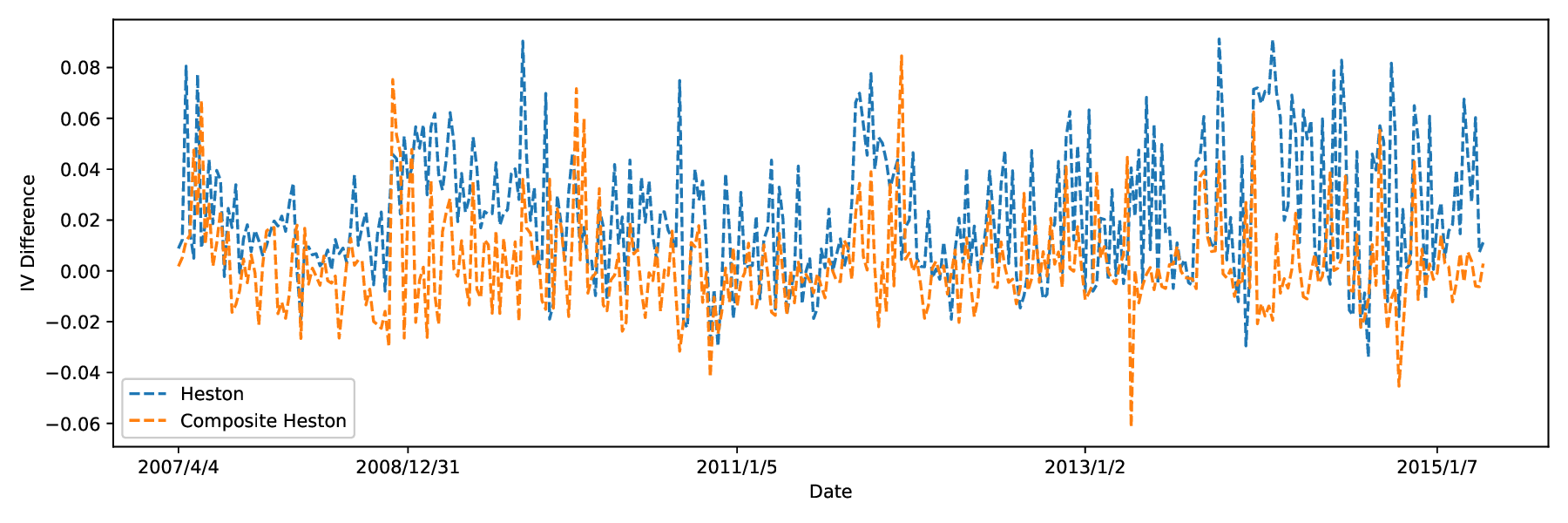}
}
\end{minipage}
\vspace{-2pt}
\footnotesize \textit{Note. } The plots compare the long-term ATM IV, defined by Eq. \eqref{IV}, between the models (Heston and Composite Heston) and the market. The first row plots the ATM IV of Heston model, the second row plots that of Composite Heston model, and the last row compares the corresponding difference (model - market) between models and the market. Maturities of SPX options smaller than 270 days are considered short-term.\label{H IV SPX long term}
\end{figure}

\begin{figure}[H]
\caption{The long-term volatility skew of Heston and Composite Heston in SPX markets}
 \begin{minipage}{\linewidth}
 \subfloat[][Long-term volatility skew]{
 \includegraphics[width=\linewidth]{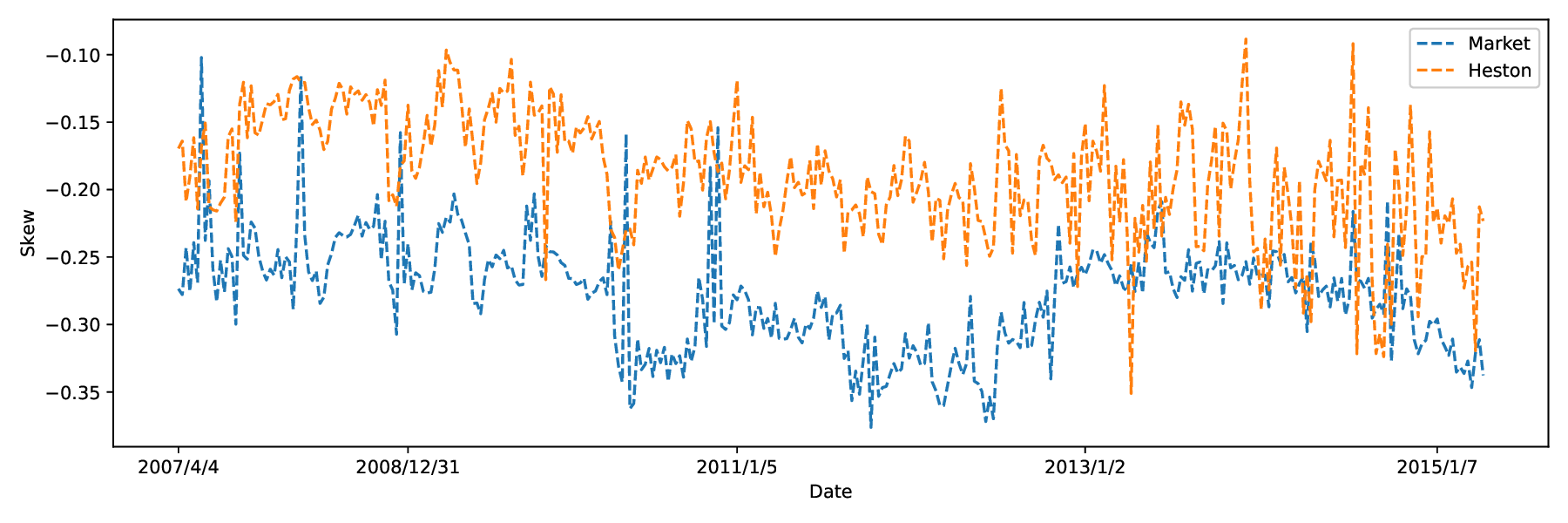}
 }
 \end{minipage} \vspace{-1em}
  \begin{minipage}{\linewidth}
 \subfloat[][Long-term volatility skew]{
 \includegraphics[width=\linewidth]{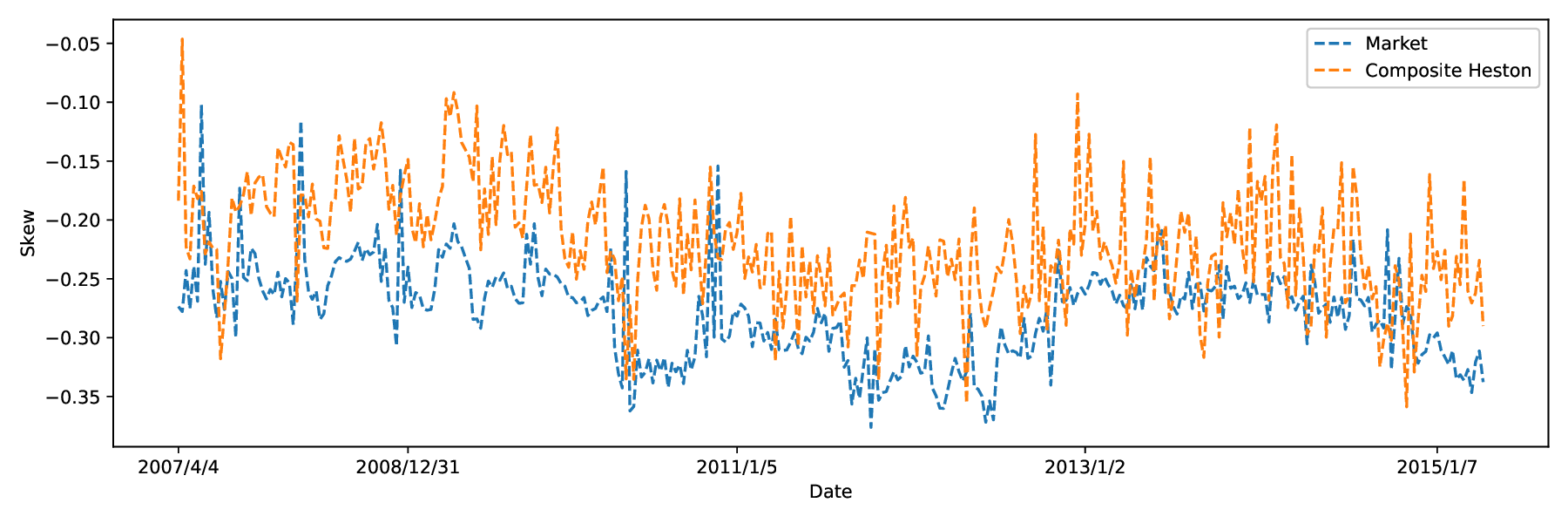}
}
 \end{minipage} \vspace{-1em}
 \begin{minipage}{\textwidth}
\subfloat[][Volatility skew difference]{
\includegraphics[width=\textwidth]{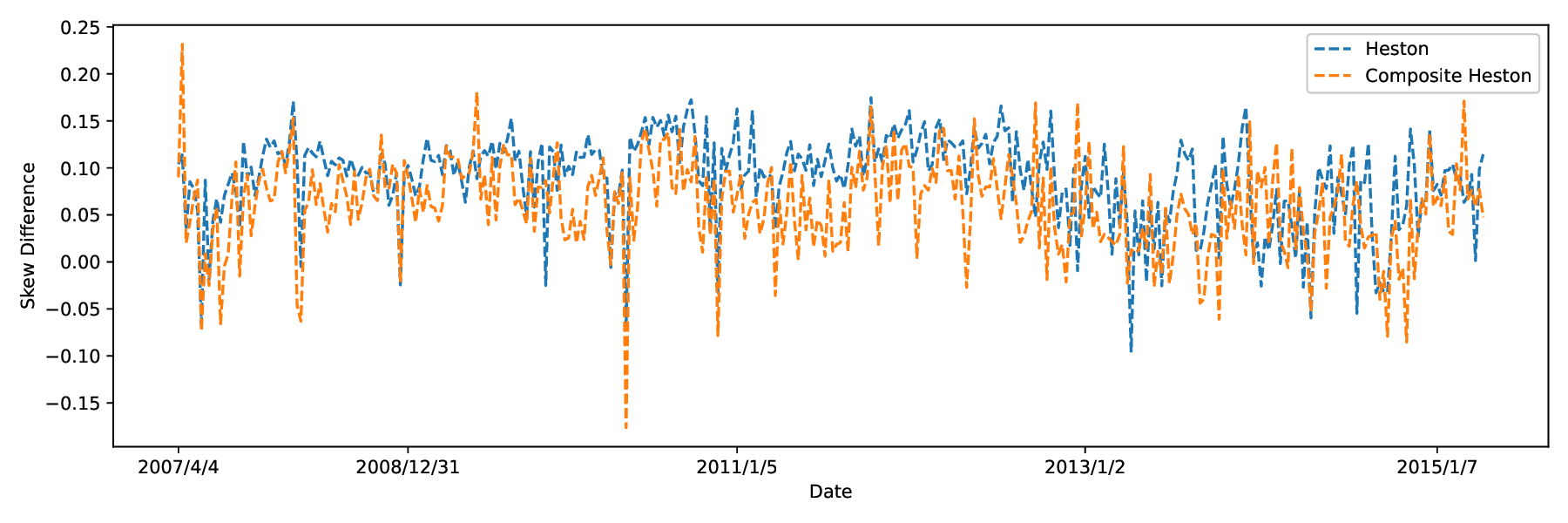}
}
\end{minipage}
\vspace{-2pt}
\footnotesize \textit{Note. } The plots compare the long-term volatility skew, defined by Eq. \eqref{Skew}, between the models (Heston and Composite Heston) and the market. The first row plots the volatility skew of Heston model, the second row plots that of Composite Heston model, and the last row compares the corresponding difference (model - market) between models and the market. Maturities of SPX options smaller than 270 days are considered short-term.\label{H IV SPX skew long term}
\end{figure}

\begin{figure}[H]
\caption{The short-term ATM implied volatility of Heston and Composite Heston in VIX markets}
 \begin{minipage}{\linewidth}
 \subfloat[][Short-term ATM IV]{
 \includegraphics[width=\linewidth]{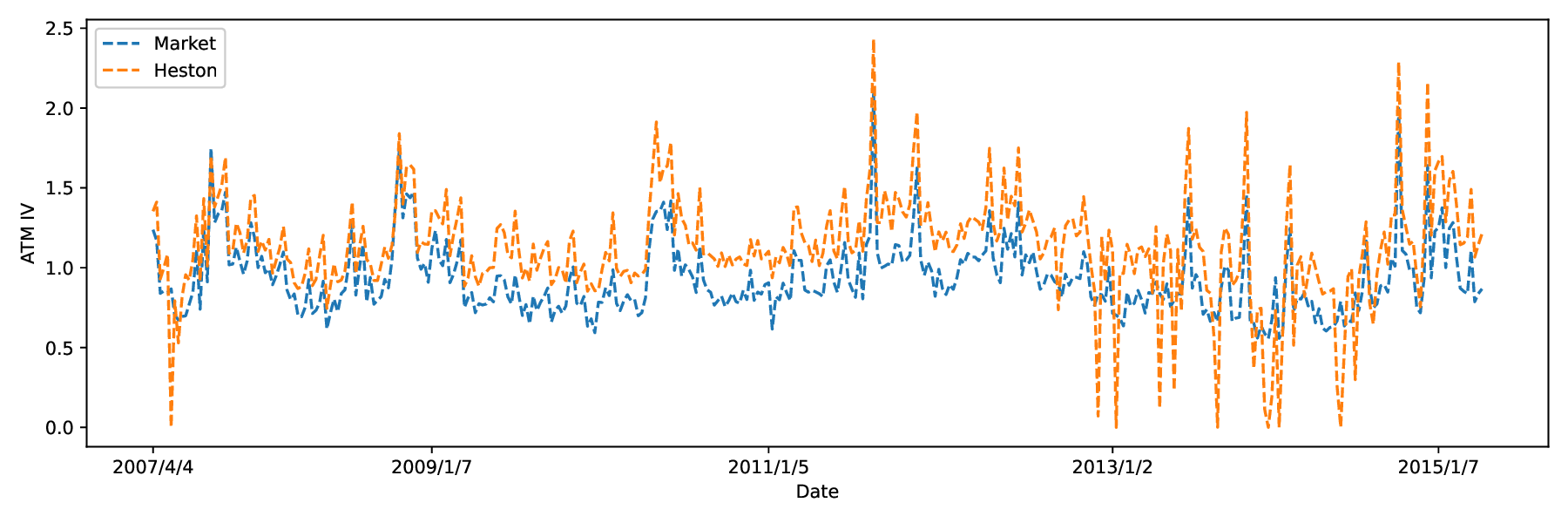}
 }
 \end{minipage} \vspace{-1em}
  \begin{minipage}{\linewidth}
 \subfloat[][Short-term ATM IV]{
 \includegraphics[width=\linewidth]{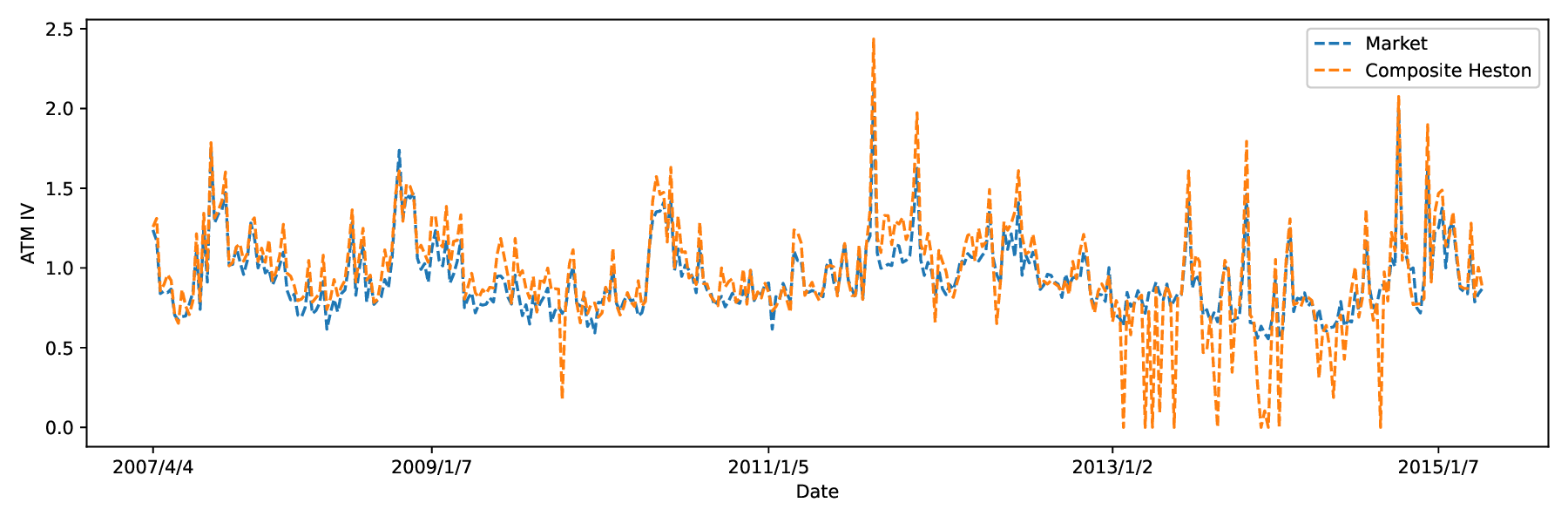}
}
 \end{minipage} \vspace{-1em}
 \begin{minipage}{\textwidth}
\subfloat[][ATM IV difference]{
\includegraphics[width=\textwidth]{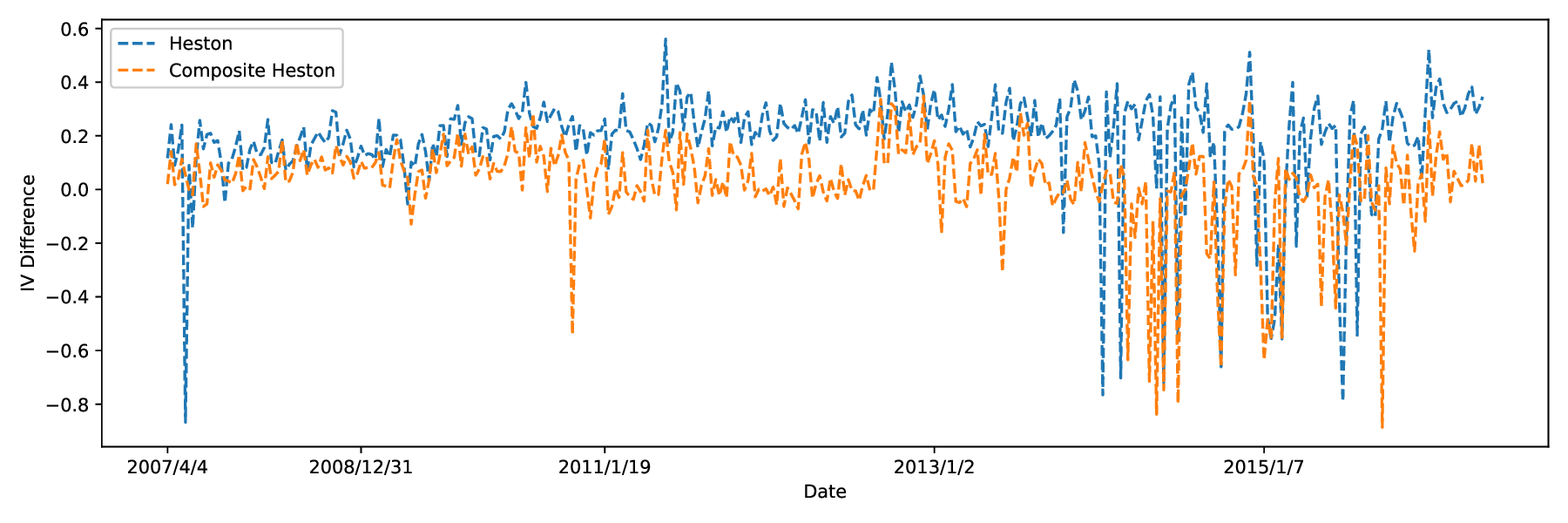}
}
\end{minipage}
\vspace{-2pt}
\footnotesize \textit{Note. } The plots compare the short-term ATM IV, defined by Eq. \eqref{IV}, between the models (Heston and Composite Heston) and the market. The first row plots the ATM IV of Heston model, the second row plots that of Composite Heston model, and the last row compares the corresponding difference (model - market) between models and the market. Maturities of SPX options smaller than 30 days are considered short-term.\label{H IV VIX}
\end{figure}

\begin{figure}[H]
\caption{The short-term volatility skew of Heston and Composite Heston in VIX markets}
 \begin{minipage}{\linewidth}
 \subfloat[][Short-term volatility skew]{
 \includegraphics[width=\linewidth]{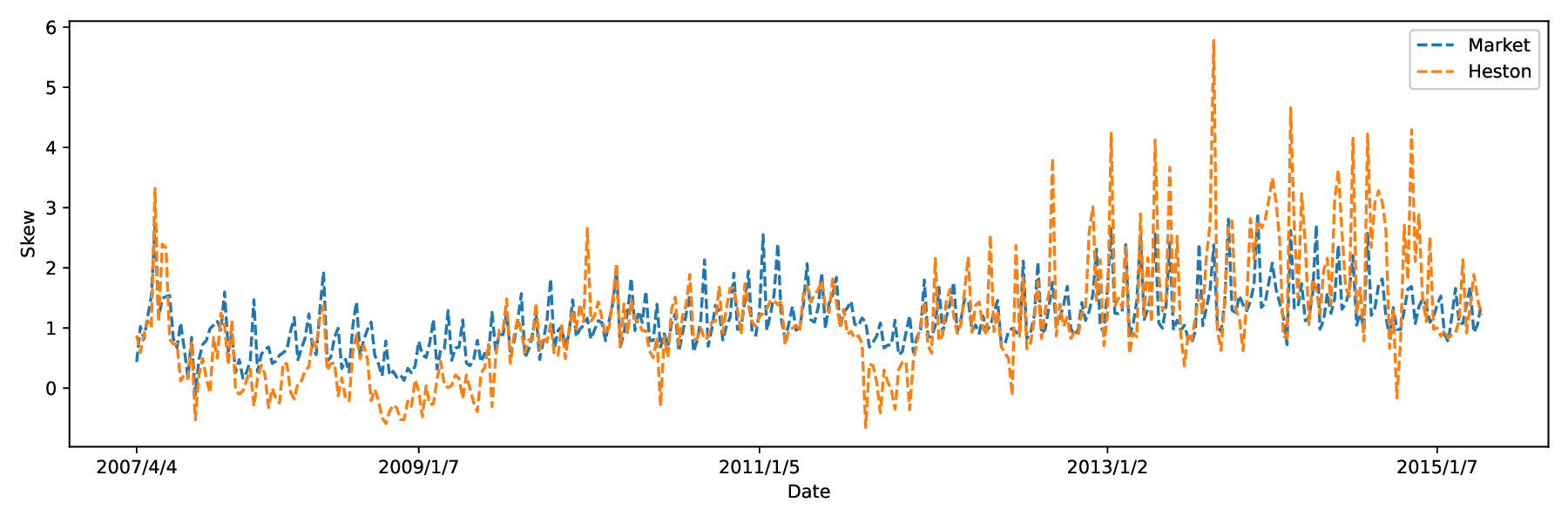}
 }
 \end{minipage} \vspace{-1em}
  \begin{minipage}{\linewidth}
 \subfloat[][Short-term volatility skew]{
 \includegraphics[width=\linewidth]{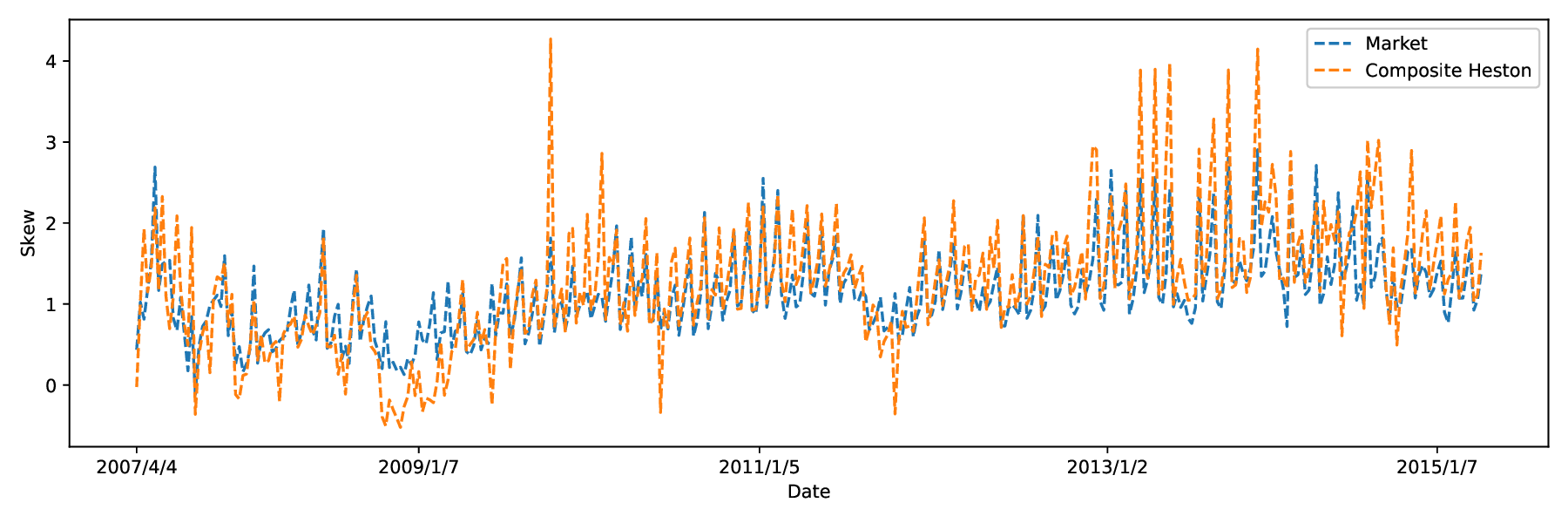}
}
 \end{minipage} \vspace{-1em}
 \begin{minipage}{\textwidth}
\subfloat[][Volatility skew difference]{
\includegraphics[width=\textwidth]{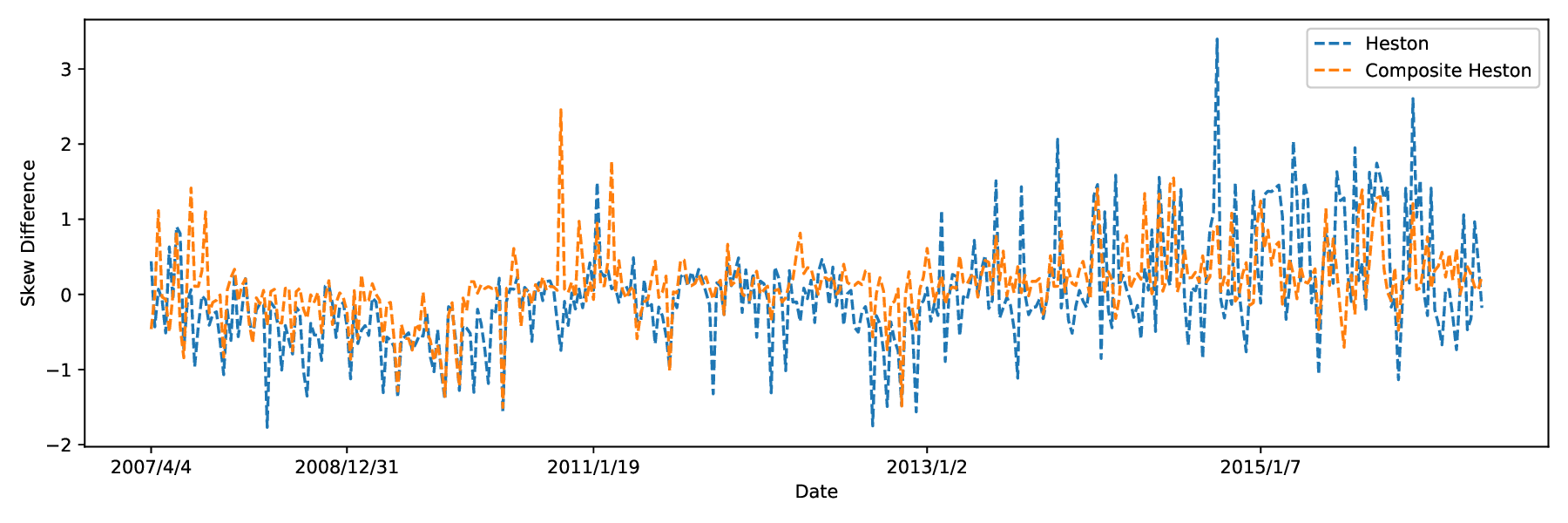}
}
\end{minipage}
\vspace{-2pt}
\footnotesize \textit{Note. } The plots compare the short-term volatility skew, defined by Eq. \eqref{Skew}, between the models (Heston and Composite Heston) and the market. The first row plots the volatility skew of Heston model, the second row plots that of Composite Heston model, and the last row compares the corresponding difference (model - market) between models and the market. Maturities of SPX options smaller than 30 days are considered short-term.\label{H IV VIX Skew}
\end{figure}

\begin{figure}[H]
\caption{The long-term ATM implied volatility of Heston and Composite Heston in VIX markets}
 \begin{minipage}{\linewidth}
 \subfloat[][Long-term ATM IV]{
 \includegraphics[width=\linewidth]{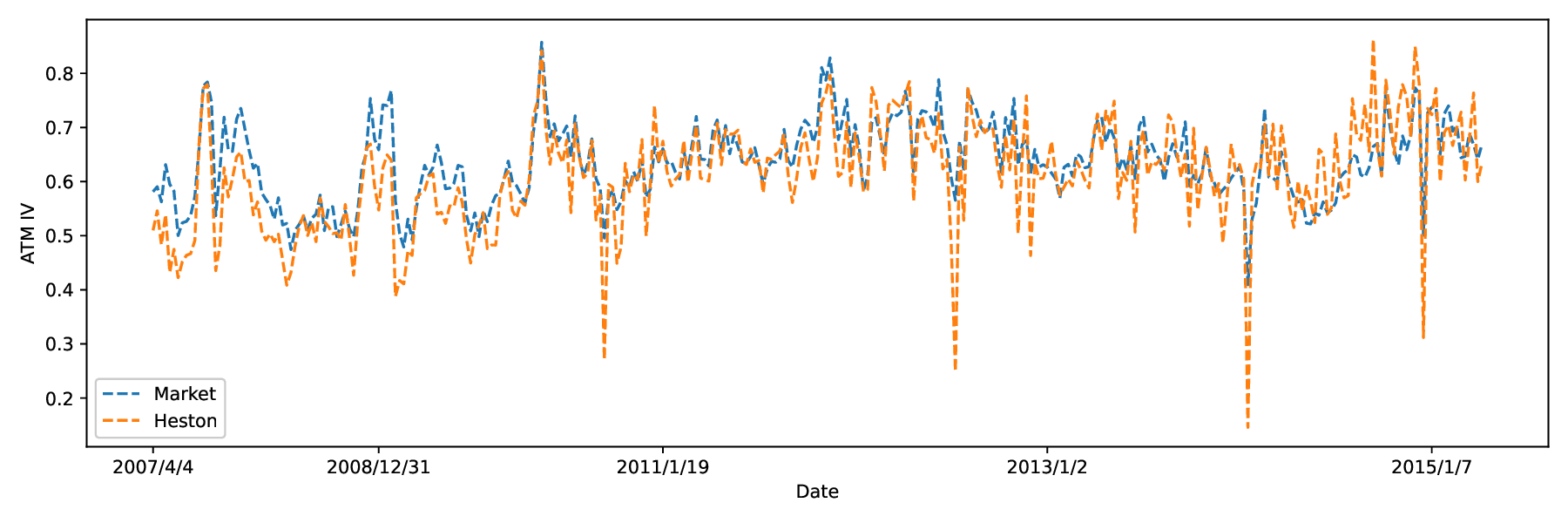}
 }
 \end{minipage} \vspace{-1em}
  \begin{minipage}{\linewidth}
 \subfloat[][Long-term ATM IV]{
 \includegraphics[width=\linewidth]{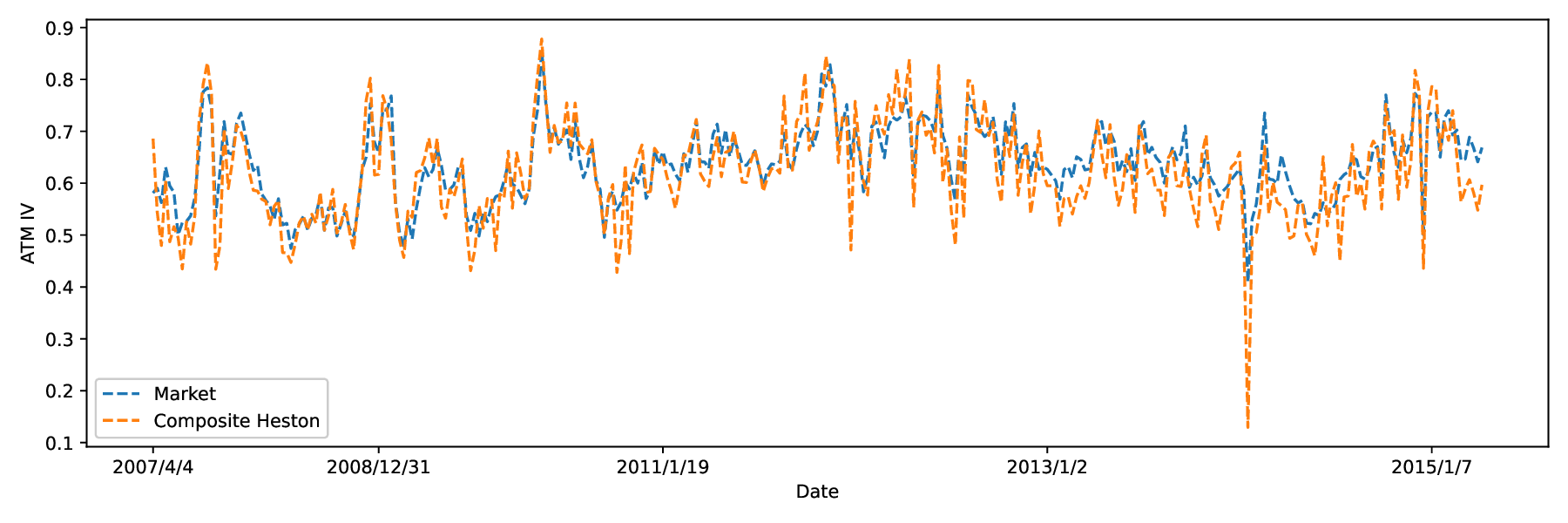}
}
 \end{minipage} \vspace{-1em}
 \begin{minipage}{\textwidth}
\subfloat[][ATM IV difference]{
\includegraphics[width=\textwidth]{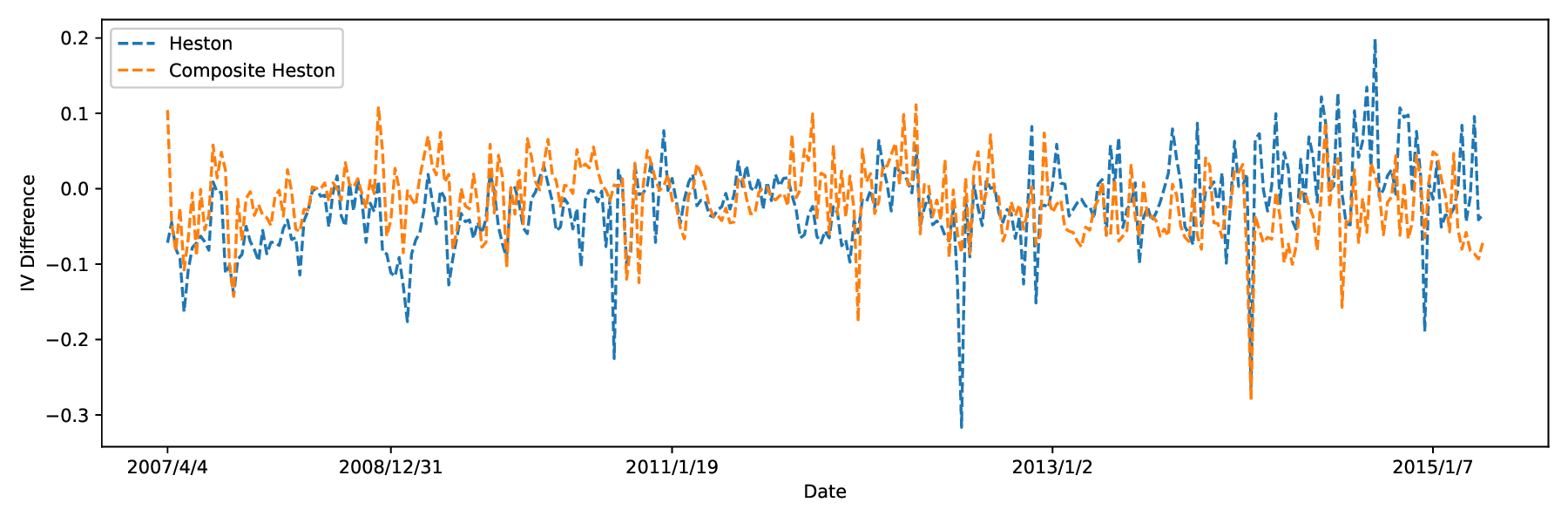}
}
\end{minipage}
\vspace{-2pt}
\footnotesize \textit{Note. } The plots compare the long-term ATM IV, defined by Eq. \eqref{IV}, between the models (Heston and Composite Heston) and the market. The first row plots the ATM IV of Heston model, the second row plots that of Composite Heston model, and the last row compares the corresponding difference (model - market) between models and the market. Maturities of SPX options smaller than 270 days are considered short-term.\label{H IV VIX long term}
\end{figure}

\begin{figure}[H]
\caption{The long-term volatility skew of Heston and Composite Heston in VIX markets}
 \begin{minipage}{\linewidth}
 \subfloat[][Long-term volatility skew]{
 \includegraphics[width=\linewidth]{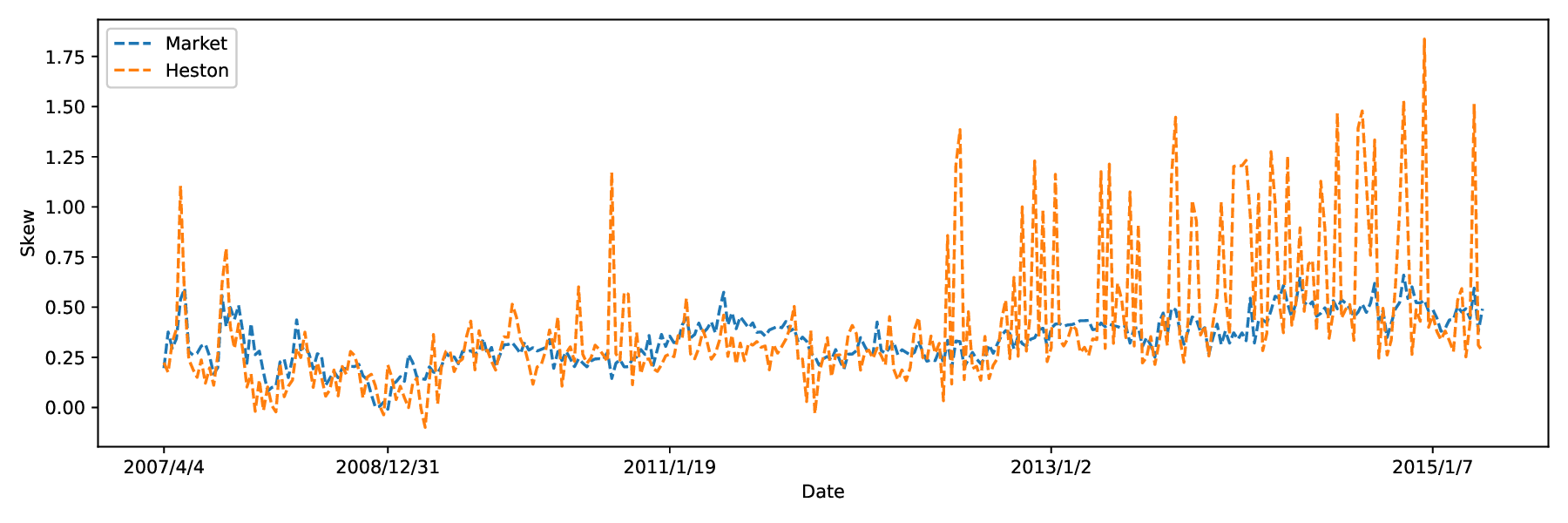}
 }
 \end{minipage} \vspace{-1em}
  \begin{minipage}{\linewidth}
 \subfloat[][Long-term volatility skew]{
 \includegraphics[width=\linewidth]{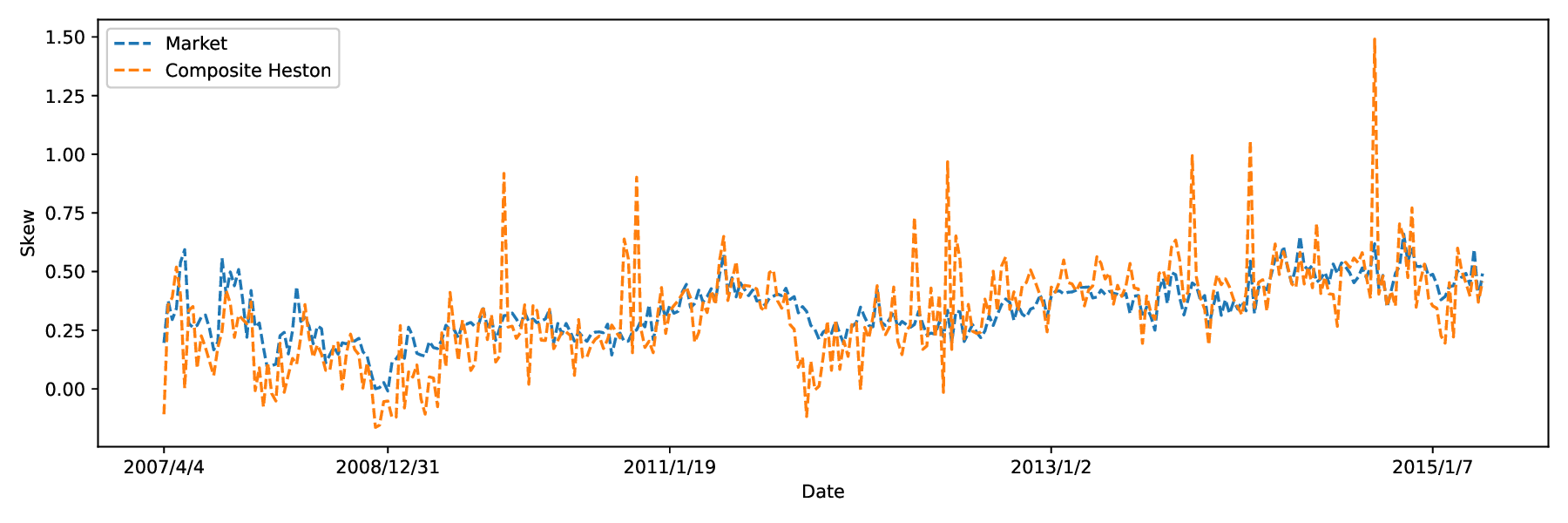}
}
 \end{minipage} \vspace{-1em}
 \begin{minipage}{\textwidth}
\subfloat[][Volatility skew difference]{
\includegraphics[width=\textwidth]{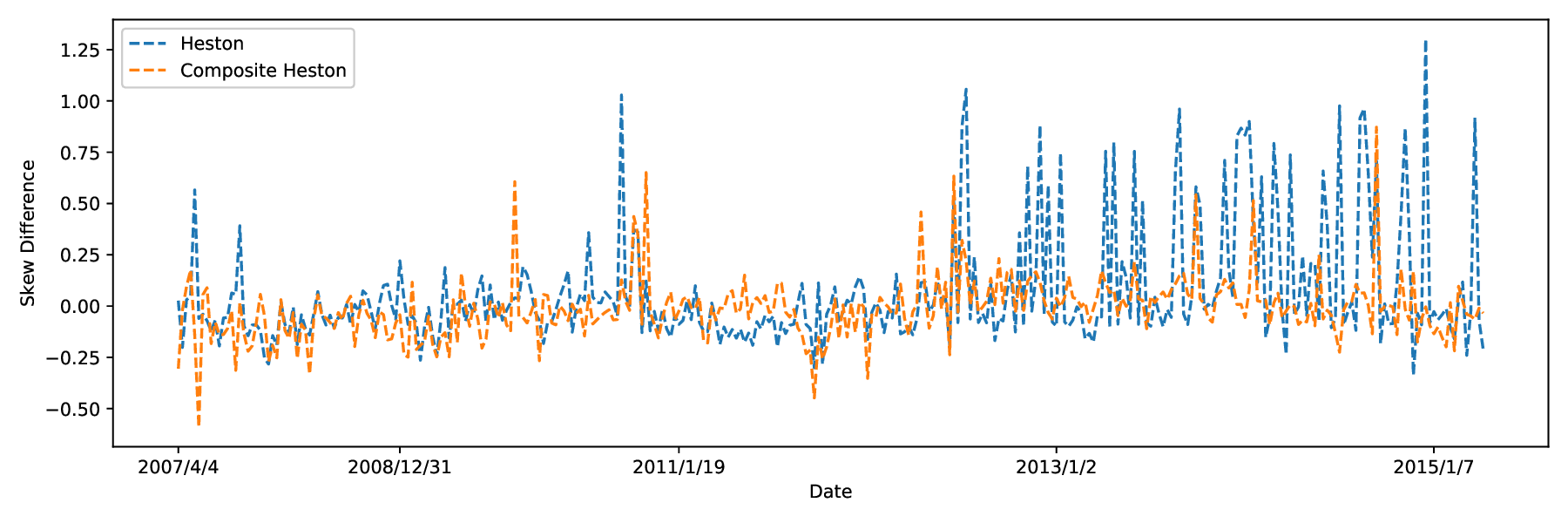}
}
\end{minipage}
\vspace{-2pt}
\footnotesize \textit{Note. } The plots compare the long-term volatility skew, defined by Eq. \eqref{Skew}, between the models (Heston and Composite Heston) and the market. The first row plots the volatility skew of Heston model, the second row plots that of Composite Heston model, and the last row compares the corresponding difference (model - market) between models and the market. Maturities of SPX options smaller than 270 days are considered short-term.\label{H IV VIX skew long term}
\end{figure}

\subsection{Robustness Test of Skew}
\label{robustness}

To assess the robustness of the joint calibration to implied volatility characteristics, 
we compute the volatility skew using alternative moneyness ranges compared to those in Table~\ref{characteristics}. 
In Table~\ref{ro 1}, the skew is derived from implied volatilities with moneyness closest to 
\text{0.8} (0.7) and \text{1.1} (1.2) for the SPX (VIX) options market.
And in Table~\ref{ro 2}, the corresponding moneyness ranges are \text{0.6} (0.7) and \text{1.3} (1.8).

\begin{table}[H]
\caption{Short-term and long-term calibration accuracy of implied volatility characteristics\label{ro 1} (robust test I)}
\begin{center}
    \begin{tabular}{ccccc}
         & Heston & Composite Heston & JH & Composite JH \\\hline
  \multicolumn{5}{c}{Panel A: S\&P 500 options}\\
   \multicolumn{5}{l}{Panel A.1: Daily RMSE}\\\hline
    Short-term Skew & 0.4157 & 0.3446 &  0.1882 & 0.1237\\
    Long-term Skew & 0.1104 & 0.0828 & 0.0532 & 0.0695\\
     \multicolumn{5}{l}{Panel A.2: Sample RMSE}\\\hline
     Short-term Skew & 0.4114 & 0.3482 &  0.1841 & 0.1266\\
     Long-term Skew &  0.1087 &  0.0818 & 0.0523 & 0.0681 \\\hline
    \multicolumn{5}{c}{Panel B: VIX options}
         \\
    \multicolumn{5}{l}{Panel B.1: Daily RMSE}\\\hline
       Short-term Skew & 1.1024 & 0.7245 & 0.3499 & 0.2289\\
        Long-term Skew & 0.2874 & 0.1785 & 0.1083 & 0.0805\\
       \multicolumn{5}{l}{Panel B.2: Sample RMSE}\\\hline
     Short-term Skew & 1.1024 & 0.7245 & 0.3499 & 0.2289\\
     Long-term Skew & 0.3005 & 0.1810 &  0.0985 & 0.0817\\\hline
    \end{tabular}
    \end{center}
    {\textit{Note.} This table evaluates the calibration accuracy of short-term and long-term volatility skew for four models separately for SPX and VIX options. Panel A presents results for SPX options (A.1 = Daily RMSE; A.2 = Sample RMSE), and Panel B presents results for VIX options (B.1 = Daily RMSE; B.2 = Sample RMSE). Volatility skew is defined as the slope of two implied volatilities (Equation \eqref{Skew})—for SPX options, moneyness is closest to 0.8 and 1.1; for VIX options, moneyness is closest to 0.7 and 1.2. Maturities of SPX (VIX) options smaller than 30 (30) days are considered short-term, and larger than 270 (90) days are considered long-term. Daily RMSE is computed by the root mean square of daily mean absolute errors, and sample RMSE takes the root mean square error over all qualified option maturities.}
\end{table}

\begin{table}[H]
\caption{Short-term and long-term calibration accuracy of implied volatility characteristics\label{ro 2}(robust test II)}
   \begin{tabular}{ccccc}
         & Heston & Composite Heston & JH & Composite JH \\\hline
  \multicolumn{5}{c}{Panel A: S\&P 500 options}\\
   \multicolumn{5}{l}{Panel A.1: Daily RMSE}\\\hline
    Short-term Skew & 0.4350 & 0.3618 &  0.1851 & 0.1081\\
    Long-term Skew & 0.0944 & 0.0721 & 0.0397 & 0.0572\\
     \multicolumn{5}{l}{Panel A.2: Sample RMSE}\\\hline
     Short-term Skew & 0.4373 & 0.3636 &  0.1796 & 0.1102\\
     Long-term Skew &  0.0930 &  0.0714 & 0.0392 & 0.0560 \\\hline
    \multicolumn{5}{c}{Panel B: VIX options}
         \\
    \multicolumn{5}{l}{Panel B.1: Daily RMSE}\\\hline
       Short-term Skew & 0.6298 & 0.4135 & 0.2612 & 0.1597\\
        Long-term Skew & 0.2269 & 0.1357 & 0.0961 & 0.0602\\
       \multicolumn{5}{l}{Panel B.2: Sample RMSE}\\\hline
     Short-term Skew & 0.6298 & 0.4135 & 0.2612 & 0.1597\\
     Long-term Skew & 0.2390 & 0.1351 &  0.1044 & 0.0594\\\hline
    \end{tabular}
    
    {\textit{Note.} This table evaluates the calibration accuracy of short-term and long-term volatility skew for four models separately for SPX and VIX options. Panel A presents results for SPX options (A.1 = Daily RMSE; A.2 = Sample RMSE), and Panel B presents results for VIX options (B.1 = Daily RMSE; B.2 = Sample RMSE). Volatility skew is defined as the slope of two implied volatilities (Equation \eqref{Skew})—for SPX options, moneyness is closest to 0.6 and 1.3; for VIX options, moneyness is closest to 0.7 and 1.8. Maturities of SPX (VIX) options smaller than 30 (30) days are considered short-term, and larger than 270 (90) days are considered long-term. Daily RMSE is computed by the root mean square of daily mean absolute errors, and sample RMSE takes the root mean square error over all qualified option maturities.}
\end{table}

\bibliography{references}
\bibliographystyle{plainnat}
\end{document}